\documentclass[11pt]{scrartcl} %{article} %,letterpaper 
\usepackage[sort&compress,numbers]{natbib}%colon %%%% natbib package must be before algorithm2e
\usepackage{amsmath,amssymb,latexsym,amsopn,amsfonts}
\usepackage[boxed,ruled,vlined,linesnumbered,noresetcount]{algorithm2e}
\usepackage{url,fullpage}
\usepackage{framed,microtype}
\usepackage[normalem]{ulem}
\usepackage{array}
\usepackage{graphicx}
\usepackage[hidelinks]{hyperref} %% Hyperref after all packages
\usepackage{color}
\usepackage[dvipsnames]{xcolor}
\definecolor{oskar_green}{rgb}{0.0, 0.5, 0.0}

\newcommand{\ol}[1]{\textcolor{oskar_green}{#1}}
\newcommand{\VCalgSize}{footnotesize}%{smaller} %scriptsize 
\newcommand{\VCalgSizeSmall}{scriptsize}%{smaller} %scriptsize 
\newcommand{\bigO}{\mathcal{O}}
\newcommand{\remove}[1]{}

\newcommand{\baseSnapshot}{\mathrm{baseSnapshot}}
\newcommand{\exceed}{\Delta}
\newcommand{\pndSnap}{\mathit{pndTsk}}
\newcommand{\repSnap}{\mathit{repSnap}}
\newcommand{\etal}{\emph{et al.}\xspace}
\newcommand{\eg}{\emph{e.g.,}\xspace}
\newcommand{\ie}{\emph{i.e.,}\xspace}

\newcommand{\figScale}{0.5}

\usepackage{etoolbox}
\patchcmd{\thebibliography}{\chapter*}{\section*}{}{}

\newtheorem{remark}{Remark}[section]
\newtheorem{theorem}{Theorem}[section]
\newtheorem{lemma}[theorem]{Lemma}
\newtheorem{definition}{Definition}[section]
\newtheorem{claim}[theorem]{Claim}

\newenvironment{claimProof}[1]{\par\noindent\underline{Proof:}\space#1}{\hfill $\blacksquare$}

\newcommand{\sS}{\mathcal{S}}
\newcommand{\sP}{\mathcal{P}}

\newcommand{\N}{\mathbb{N}}
\newcommand{\capacity}{\mathsf{capacity}}

\newcommand{\op}{\mathsf{op}}
\newcommand{\sWrite}{\mathsf{write}}
\newcommand{\sSnapshot}{\mathsf{snapshot}}
\newcommand{\sStore}{\mathsf{safeReg}}

\newcommand{\sdelta}{S_i \cap \Delta_i}

\usepackage{caption}
\setcapindent{0pt}
\addtokomafont{caption}{\centering}

\SetKwInOut{Input}{input}
\SetKwInOut{Output}{output}

\newenvironment{proof}{\noindent \textbf{Proof.}}{\hfill$\blacksquare$}
\renewenvironment{claimProof}{\noindent\textbf{Proof of claim.}}{\hfill$\Box$}

\usepackage[utf8]{inputenc}

\begin{document}
\setcounter{footnote}{2}
%\renewcommand{\thefootnote}{\fnsymbol{footnote}}
%\appendix	
\title{Self-Stabilizing Snapshot Objects\\ for Asynchronous Failure-Prone Networked Systems\\~{(preliminary version)}} %--- to be upload to  arXiv with the paper acceptance)}}
%{\Huge Appendix}\\~~\\  \normalsize
%~\footnote{\cgr{I would remove this, especially since this is a double-blind submission}This document is a preliminary version of a technical report~\cite{CS2018}}\\\Large{\cg{\sout{(preliminary report)}}}

\author{Chryssis Georgiou~\footnote{Department of Computer Science, University of Cyprus, Nicosia, Cyprus. E-mail: \texttt{chryssis@cs.ucy.ac.cy}} \and Oskar Lundstr\"om~\footnote{Department of Computer Science and Engineering, Chalmers University of Technology, Gothenburg, SE-412 96, Sweden, E-mail: \texttt{osklunds@student.chalmers.se}.} \and Elad M.\ Schiller~\footnote{Department of Computer Science and Engineering, Chalmers University of Technology, Gothenburg, SE-412 96, Sweden, E-mail: \texttt{elad@chalmers.se}.}} 

%\author{Chryssis Georgiou \inst{1} \and Elad M.\ Schiller \inst{2}}
%\authorrunning{Ch. Georgiou and E.M. Schiller}
%\institute{Computer Science, University of Cyprus, Cyprus. \email{chryssis@cs.ucy.ac.cy}  
%	\and Computer Science and Engineering, Chalmers Univ. of Tech.
%	\email{elad@chalmers.se}}
	
\date{}

\maketitle

\begin{abstract}
A \emph{snapshot object} simulates the behavior of an array of single-writer/multi-reader shared registers that can be read atomically. Delporte-Gallet \etal proposed two fault-tolerant algorithms for snapshot objects in asynchronous crash-prone message-passing systems. Their first algorithm is \emph{non-blocking}; it allows snapshot operations to terminate once all write operations have ceased.
%
%, in the sense that it assumes that eventually all write operations cease for a period that is sufficiently long to allow pending snapshot operations to terminate. 
%
It uses $\bigO(n)$ messages of $\bigO(n \cdot \nu)$ bits, where $n$ is the number of nodes and $\nu$ is the number of bits it takes to represent the object. Their second algorithm allows snapshot operations to \emph{always terminate} independently of write operations. It incurs $\bigO(n^2)$ messages.

The fault model of Delporte-Gallet \etal considers node failures (crashes). We aim at the design of even more robust snapshot objects. We do so through the lenses of \emph{self-stabilization}---a very strong notion of fault-tolerance. In addition to Delporte-Gallet \etal's fault model, our self-stabilizing algorithm can recover after the occurrence of \emph{transient faults}; these faults represent arbitrary violations of the assumptions according to which the system was designed to operate (as long as the code stays intact).   

In particular, in this work, we propose self-stabilizing variations of Delporte-Gallet \etal's non-blocking algorithm and always-terminating algorithm. Our algorithms have similar communication costs to the ones by Delporte-Gallet \etal and $\bigO(1)$ recovery time (in terms of asynchronous cycles) from transient faults. The main differences are that our proposal considers repeated gossiping of $\bigO(\nu)$ bit messages and deals with bounded space (which is a prerequisite for self-stabilization). We also consider an input parameter, $\delta$, for which we claim an ability to balance the costs of snapshot operations. We validate our correctness proof, evaluate the performance of Delporte-Gallet \etal's algorithms and our proposed variations and investigate the properties of $\delta$ via PlanetLab experiments, where significant latency and communication costs reduction are observed. 

%Lastly, we explain how to extend the proposed solutions to reconfigurable ones. \OL{How to make them reconfigurable is not last. The evaluation is last.}

\end{abstract}

\thispagestyle{empty}

\section{Introduction}
\label{sec:intro}
We propose self-stabilizing implementations of shared memory snapshot objects for asynchronous networked systems whose nodes may crash.

\smallskip

\noindent \textbf{Context and Motivation.~~} 
Shared registers are fundamental objects that facilitate synchronization in distributed systems. In the context of networked systems, they provide a higher abstraction level than a simple end-to-end communication. Share registers  provide persistent and consistent distributed storage that can simplify the design and analysis of dependable distributed systems. Snapshot objects extend shared registers. They provide a way to further make the design and analysis of algorithms that base their implementation on shared registers easier. Snapshot objects allow an algorithm to construct consistent global states of the shared storage in a way that does not disrupt the system computation. Their efficient and fault-tolerant implementation is a fundamental problem, as  there are many examples of algorithms that are built on top of snapshot objects; see textbooks such as~\cite{DBLP:series/synthesis/2011Welch,DBLP:series/synthesis/2018Taubenfeld} and recent reviews, such as~\cite{DBLP:reference/algo/Ruppert16}.

\smallskip

\noindent \textbf{Task description.~~} 
Consider a fault-tolerant distributed system of $n$ asynchronous nodes that are prone to failures. Their interaction is based on the emulation of Single-Writer/Multi-Reader (SWMR) shared registers over a message-passing communication system. Snapshot objects can read the entire array of system registers~\cite{DBLP:journals/jacm/AfekADGMS93,DBLP:journals/dc/Anderson94}. The system lets each node update its own register via $\sWrite()$ operations and retrieve the value of all shared registers via $\sSnapshot()$ operations. Note that these snapshot operations may occur concurrently with the write operations that individual nodes perform. We are particularly interested in the study of atomic snapshot objects that are \emph{linearizable}~\cite{DBLP:journals/toplas/HerlihyW90}: the operations $\sWrite()$ and $\sSnapshot()$ appear as if they have been executed instantaneously, one after the other (in other words, they appear to preserve real-time ordering).   

\smallskip

\noindent \textbf{Fault Model.~~} 
We consider an asynchronous message-passing system that has no guarantees on the communication delay. Moreover, there is no notion of global (or universal) clocks and we do not assume that the algorithm can explicitly access the local clock (or  timeout mechanisms). Our fault model includes $(i)$ node crashes, and $(ii)$ communication failures, such as packet omission, duplication, and reordering. In addition to the failures captured in our model, we also aim to recover from \emph{transient faults}, \ie any temporary violation of assumptions according to which the system and network were designed to behave, \eg the corruption of control variables, such as the program counter and operation indices, which are responsible for the correct operation of the studied system, or operational assumptions, such as that at least half of the system nodes never fail. Since the occurrence of these failures can be combined, we assume that these transient faults can alter the system state in unpredictable ways. 

We follow the design criteria of self-stabilization, which was proposed by Dijkstra~\cite{DBLP:journals/cacm/Dijkstra74} and detailed in~\cite{DBLP:books/mit/Dolev2000,DBLP:series/synthesis/2019Altisen}. In particular, when modeling the system, we assume that the above violations bring the system to an arbitrary state from which a \emph{self-stabilizing algorithm} should recover the system. 
%
%\EMS{Thus, a self-stabilizing can recover after the occurrence of transient faults, which can model transient memory corruptions as well as temporary violations of operational assumptions, such as that at least half of the system nodes never fails. These violations and memory corruptions can bring the system to an arbitrary state from which the algorithm should recover the system. 
%	
Therefore, starting from an arbitrary state, the correctness proof of self-stabilizing systems~\cite{DBLP:journals/cacm/Dijkstra74} has to demonstrate the return to a ``correct behavior'' within a bounded period, which brings the system to a \emph{legitimate state}. The complexity measure of self-stabilizing systems is the length of the recovery period.      

As transient faults can occur at any point in a system's lifetime, self-stabilizing systems need to keep communicating their state structures for cleaning any potential corrupted (stale) information; to this respect, a self-stabilizing system cannot really terminate~\cite[Chapter 2.3]{DBLP:books/mit/Dolev2000}. Specifically, the proposed solution repeatedly broadcasts $\bigO(\nu)$-size gossip messages that facilitate the system clean-up from stale information, where $\nu$ is the number of bits it takes to represent the object. We note the trade-off between the cost related to these gossip messages and the recovery time. That is, one can balance this trade-off by, for example, reducing the rate of gossip messages, which prolongs the stabilization time. We clarify that the rate in which these repeated clean-up operations take place does not impact the execution time of the $\sWrite()$ and $\sSnapshot()$ operations. 

\smallskip

\noindent \textbf{Related work.~~} 
We now overview existing work related to ours. Our review does not focus on algorithms for shared memory system; although there are examples for both non-self-stabilizing~\cite{DBLP:journals/ipl/KirousisST96,DBLP:journals/tpds/KirousisST94} and self-stabilizing~\cite{Abraham07selfstabilizing} solutions. 
 
\noindent \emph{Shared registers emulation in message-passing systems:~~}
Attiya \etal~\cite{DBLP:journals/jacm/AttiyaBD95} implemented SWMR atomic shared memory in an asynchronous networked system. They assume that the majority of the nodes do not crash or get disconnected. Their work builds on this assumption in the following manner: Any majority subset of the system nodes includes at least one non-faulty node; thus, any two majority subsets of the system nodes have a non-empty intersection. They show that if a majority of the nodes acknowledge an update to the shared register, then that update can safely be considered visible to all non-faulty nodes that retrieve the latest update from a majority of nodes. Attiya \etal also show that this assumption is essential for solvability. Their seminal work has many generalizations and applications~\cite{DBLP:journals/eatcs/Attiya10}. The literature includes a large number of simulation of shared registers for networked systems, which differ in their fault tolerance properties, time complexity, storage costs, and system properties, \eg~\cite{DBLP:journals/jal/Attiya00,DBLP:journals/siamcomp/DuttaGLV10,DBLP:conf/netys/HadjistasiNS17,DBLP:conf/podc/MostefaouiR16,DBLP:books/sp/Raynal18,DBLP:books/daglib/0032304}.

In the context of self-stabilization, the literature includes a practically-self-stabilizing variation for the work of Attiya \etal~\cite{DBLP:journals/jacm/AttiyaBD95} by Alon \etal~\cite{DBLP:journals/jcss/AlonADDPT15}. Their proposal guarantees wait-free recovery from transient faults. However, there is no bound on the recovery time. Dolev \etal~\cite{DBLP:journals/corr/abs-1806-03498} consider MWMR atomic storage that is wait-free in the absence of transient faults. They guarantee a bounded time recovery from transient faults in the presence of a fair scheduler. They demonstrate the algorithm's ability to recover from transient faults using unbounded counters and in the presence of fair scheduling. Then they deal with the event of integer overflow via a consensus-based procedure. Since integer variables can have $64$-bits, their algorithm seldom uses this non-wait-free (consensus-based) procedure for dealing with integer overflows. In fact, they model integer overflow events as transient faults, which implies bounded recovery time from transient faults in the seldom presence of a fair scheduler (using bounded memory). They call these systems \emph{self-stabilizing systems in the presence of seldom fairness.} Our work adopts these design criteria. We also make use of their self-stabilizing quorum and gossip services~\cite[Section~13]{DBLP:journals/corr/abs-1806-03498}.

\noindent \emph{Implementing a snapshot object on top of a message-passing system:~~} A straightforward way for implementing snapshot objects is to consider a layer of $n$ SWMR atomic registers emulated in a networked system. On top of this layer of these emulated registers, the system can run any algorithm for implementing a snapshot object for a system with shared variables. Delporte-Gallet \etal~\cite{DBLP:journals/tpds/Delporte-Gallet18} avoid this composition, obtaining, in this way, a more efficient implementation with respect to the communication costs. Specifically, they claim that when stacking the shared-memory atomic snapshot algorithm of~\cite{DBLP:journals/jacm/AfekADGMS93} on the shared-memory emulation of \cite{DBLP:journals/jacm/AttiyaBD95} (with some improvements), the number of messages per snapshot operation is $8n$ and it takes four round trips. Their proposal, instead, takes $2n$ messages per snapshot operation and just one round trip to complete. The algorithms we propose in the present work follow the approach of Delporte-Gallet, which does not consider stacking. Our proposal have the same communication costs for write and snapshot operations Delporte-Gallet \etal~\cite{DBLP:journals/tpds/Delporte-Gallet18}. Moreover, they tolerate any failure (in any communication or operation invocation pattern) that~\cite{DBLP:journals/tpds/Delporte-Gallet18} can. Furthermore, our algorithms deal with transient faults by periodically removing stale information. To that end, the algorithms broadcast gossip message of $\bigO(\nu)$ bits, where $\nu$ is the number of bits it takes to represent the object.

In the context of self-stabilization, there exist algorithms for the propagation of information with feedback, \eg~\cite{DBLP:journals/jpdc/DelaetDNT10} that can facilitate the implementation of snapshot objects that can recover from transient faults, but not from node failures. For the sake of clarity, we note that ``stacking''  of self-stabilizing algorithms for asynchronous systems is not a straightforward process (since the existing ``stacking'' requires schedule\ol{r} fairness, see~\cite[Section 2.7]{DBLP:books/mit/Dolev2000}). Moreover, we are unaware of an attempt in the literature to stack a self-stabilizing shared-memory atomic snapshot algorithm (such as the weak snapshots algorithm of Abraham~\cite{Abraham07selfstabilizing} that uses $\bigO(n)$ register size) over a self-stabilizing shared-memory emulation, such as the one of Dolev \etal~\cite{DBLP:journals/siamcomp/DolevIM97}.

\smallskip

\noindent \textbf{Our Contributions.~~} We present an novel module for dependable distributed systems:  self-stabilizing algorithms for snapshot objects in networked systems. To the best of our knowledge, we are the first to provide a broad fault model that includes both node failures and transient faults. Specifically, we advance the state of the art as follows:

\begin{enumerate}
	\item As a first contribution, we offer a self-stabilizing variation of the non-blocking algorithm presented by Delporte-Gallet \etal~\cite{DBLP:journals/tpds/Delporte-Gallet18}. Their solution tolerates node failures as well as packet omission, duplication, and reordering. Each snapshot or write operation uses $\bigO(n)$ messages of $\bigO(\nu\cdot n)$ bits, where $n$ is the number of nodes and $\nu$ is the number of bits for encoding the object. The termination of a snapshot operation depends on the assumption that the invocation of all write operations cease eventually.
	
	Our solution broadens the set of failure types it can tolerate. Namely, the proposed algorithms can also recover after the occurrence of transient faults, which model any violation of the assumptions according to which the system was designed to operate (as long as the code stays intact). We increase the communication costs slightly by using gossip messages of $\bigO(\nu)$ bits that are sent on every link, where $\nu$ is the number of bits it takes to represent the object.
	
	\item Our second contribution offers a self-stabilizing variation of the always-terminating algorithm presented by Delporte-Gallet \etal~\cite{DBLP:journals/tpds/Delporte-Gallet18}. Our algorithm can: (i) recover from transient faults, and (ii) both write and snapshot operations always terminate (regardless of the invocation patterns of any operation).
	
	We achieve $(ii)$ by choosing to use \emph{safe registers} for storing the result of recent snapshot operations, rather than (the more evolved) \emph{reliable broadcast} mechanism. Moreover, instead of dealing with one snapshot task at a time, we take care of several at a time. We also	consider an input parameter, $\delta$. For the case of $\delta=0$, our self-stabilizing algorithm guarantees an always-termination behavior in a way the resembles the non-self-stabilizing always-terminating algorithm by Delporte-Gallet \etal~\cite{DBLP:journals/tpds/Delporte-Gallet18} that has the cost of $\bigO(n^2)$ messages per snapshot operation. For the case of $\delta>0$, our solution aims at using $\bigO(n)$ messages per snapshot operation while monitoring the number of concurrent write operations. Once our algorithm notices that a snapshot operation runs concurrently with at least $\delta$ write operations, it blocks all write operations and uses $\bigO(n^2)$ messages for completing the snapshot operations.

Thus, the proposed algorithm can trade communication costs with an  $\bigO(\delta)$ bound on snapshot operation latency. Moreover, between any two consecutive periods in which snapshot operations block the system for write operations, the algorithm guarantees that at least $\delta$ write operations occur.    

\item Using Lamport's happened-before relation~\cite{DBLP:journals/cacm/Lamport78}, we bound the cost of the proposed self-stabilizing operations to constants, for the case of non-blocking solution, and $\bigO(n+\delta)$, for the case of the always-terminating solution, where $\delta$ is a predefined constant that we introduced in the previous item. The two proposed algorithms presented in sections~\ref{sec:aussnba} and~\ref{sec:aussata} consider unbounded counters. In Section~\ref{sec:bounded}, we explain how to bound these counters.

% as well as how to extend our solutions to reconfigurable ones.

\item We validate our correctness proof, evaluate the performance of the algorithms and investigate the properties of $\delta$ via a Rust-based implementation and PlanetLab experiments. We observe that Delporte-Gallet \etal's non-self-stabilizing non-blocking solution and our self-stabilizing variation have similar behaviors with negligible differences in their overheads. Moreover, Delporte-Gallet \etal's non-self-stabilizing alway-terminating solution's snapshot latency grows linearly as the number of concurrent snapshotters increases. By varying $\delta$, our self-stabilizing solution gets various latencies. For $\delta = 0$, the write latency scales fairly well and the snapshot latency is a constant that appears to be twice as much as the latency of the non-blocking solutions. In our experiments we noticed that, on the one hand, a higher $\delta$ leads to shorter write latency, while on the other hand, a higher $\delta$ leads to longer snapshot latency.

\end{enumerate}

\smallskip

\noindent \textbf{Organization.} We state our system settings in Section~\ref{sec:sys}. We review the non-self-stabilizing solutions by Delporte-Gallet \etal~\cite{DBLP:journals/tpds/Delporte-Gallet18} in Section~\ref{sec:back}. Our self-stabilizing non-blocking and always-terminating algorithms are presented in Sections~\ref{sec:aussnba} and~\ref{sec:aussata}, respectively; they consider unbounded counters. In Section~\ref{sec:bounded}, we explain how to bound the counters of the proposed self-stabilizing algorithms. We empirically validate our correctness in Section~\ref{sec:eval}. We conclude in Section~\ref{sec:disc}.

% and mention an extension that considers reconfiguration

\section{System settings}
\label{sec:sys}
We consider an asynchronous message-passing system that has no guarantees on the communication delay. Moreover, there is no notion of global (or universal) clocks and we do not assume that the algorithms can explicitly access the local clock (or  timeout mechanisms). The system consists of a set, $\sP$, of $n$ failure-prone nodes (or processors) with unique identifiers. Any pair of nodes $p_i,p_j \in \sP$ have access to a bidirectional communication channel, $channel_{i,j}$, that, at any time, has at most $\capacity \in \N$ packets on transit from $p_i$ to $p_j$ (this assumption is due to a well-known impossibility~\cite[Chapter 3.2]{DBLP:books/mit/Dolev2000}).

\subsection{Execution model}
\label{sec:interModel}
Our analysis considers the \emph{interleaving model}~\cite{DBLP:books/mit/Dolev2000}, in which the node's program is a sequence of \emph{(atomic) steps}. Each step starts with an internal computation and finishes with a single communication operation, \ie a message $send$ or $receive$. The \emph{state}, $s_i$, of node $p_i \in \sP$ includes all of $p_i$'s variables and $channel_{j,i}$. The term \emph{system state} (or configuration) refers to the tuple $c = (s_1, s_2, \cdots,  s_n)$. We define an \emph{execution (or run)} $R={c[0],a[0],c[1],a[1],\ldots}$ as an alternating sequence of system states $c[x]$ and steps $a[x]$, such that each $c[x+1]$, except for the starting one, $c[0]$, is obtained from $c[x]$ by $a[x]$'s execution.

\subsection{The specifications of the snapshot object task}
\label{sec:snapshotTask}
The set of \emph{legal executions} ($LE$) refers to all the executions in which the requirements of the task $T$ hold. In this work, $T_{\text{snapshot}}$ denotes the task of snapshot object emulation and $LE_{\text{snapshot}}$ denotes the set of executions in which the system fulfills $T_{\text{snapshot}}$'s requirements, which consider the operations $\sWrite()$ and $\sSnapshot()$. When node $p_i$ invokes the $\sWrite(v)$ operation, the system stores value $v$ in the emulated shared memory associated with $p_i$. When node $p_i$ invokes the $\sSnapshot()$ operation, the system returns an array that includes all the emulated shared memories in a way that fulfills the properties of termination and linearizability. Termination is defined as the end of operation execution within a finite time. Linearizability is given by the definition of linearizable snapshot histories, which is based on the terms events and event histories. The definitions here are brief version of the ones that are given by Delporte-Gallet \etal~\cite{DBLP:journals/tpds/Delporte-Gallet18}.

\smallskip

\noindent 
\emph{Events:~~} Let $\op$ be a $\sWrite()$ or $\sSnapshot()$ operation. The execution of an operation $\op$ by a processor $p_i$ is modeled by two steps: the invocation step,
denoted by $invoc(\op)$, which calls the $\op$ operation,  and a response event, denoted $resp(\op)$ (termination), which occurs when $p_i$ terminates (completes) the operation. For the sake of simple presentation, by \emph{event} we refer to either an operation's start step or an operation's end step.

\smallskip

\noindent 
\emph{Effective operations:~~} We say that a $\sSnapshot()$ operation is \emph{effective} when the invoking processor does not fail during the operation's execution. We say that a $\sWrite()$ operation is effective when the invoking processor does not fail  during its execution, or in case it does fail, the operation's effect is returned by an effective snapshot operation.

\smallskip

\noindent 
\emph{Histories:~~} A history is a sequence of operation start and end steps that are totally ordered. We consider histories to compare in an abstract way between two executions of the studied algorithms. Given any two events
$e$ and $f$, $e < f$ if $e$ occurs before $f$ in the corresponding history. A history is denoted by $\widehat{H}=(E,<)$, where $E$ is the set of events. Given an infinite history $\widehat{H}$, we require that: (i) its first event is an invocation and (ii) each invocation is followed by its matching response event. If $\widehat{H}$ is finite, then  $\widehat{H}$ might not contain the matching response event of the last invocation event.

\smallskip

\noindent 
\emph{Linearizable snapshot history:~~} A snapshot-based history $\widehat{H}= (H,<)$ models a computation at the abstraction level at which the write and snapshot operations are invoked. It is linearizable if there is an equivalent sequential history $\widehat{H}_{seq}= (H,<_{seq})$ in which the sequence of effective $\sWrite()$ and $\sSnapshot()$ operations issued by the processes is such that: 
\begin{enumerate}
	\item Each effective operation appears as if executed at a single point of the timeline between its invocation event and its response event, and 
	\item Each effective $\sSnapshot()$ operation returns an array $reg$ such that: (i) {$reg[i]=(v,\bullet)$} if the operation $\sWrite(v)$ by $p_i$ appears previously in the sequence and it is the latest one. (ii) Otherwise $reg[i]=\bot$.
\end{enumerate}

\begin{figure*}[t!]
	\begin{\VCalgSize}
		\centering
		\begin{tabular}{llll}
			\cline{2-3}
			\multicolumn{1}{l|}{}                    & \multicolumn{2}{l|}{~~~~~~~~~~~~~~~~~~~~~~~~~~~~~~~~~~~~~~~~~~~~\textbf{Frequency}}                                                                               \\ \hline
			\multicolumn{1}{|l|}{\textbf{Duration}}  & \multicolumn{1}{l|}{\textit{Rare}}                          & \multicolumn{1}{l|}{\textit{Not rare}}                     \\ \hline
			\multicolumn{1}{|l|}{}                   & \multicolumn{1}{l|}{Any violation of the assumptions according}         & \multicolumn{1}{l|}{Packet failures: omissions,}        \\
			\multicolumn{1}{|l|}{\textit{Transient}}          & \multicolumn{1}{l|}{to which the system is assumed to}           & \multicolumn{1}{l|}{duplications, reordering} \\
			\multicolumn{1}{|l|}{\textit{}}          & \multicolumn{1}{l|}{operate (as long as the code stays intact).}    & \multicolumn{1}{l|}{(assuming communication}      \\
			\multicolumn{1}{|l|}{\textit{}} & \multicolumn{1}{l|}{This can result in any state corruption.} & \multicolumn{1}{l|}{fairness holds).}                                   \\ \cline{2-3} 
			%\multicolumn{1}{|l|}{}                   & \multicolumn{1}{l|}{}                                       & \multicolumn{1}{l|}{Link failures (assuming} \\
			%\multicolumn{1}{|l|}{}                   & \multicolumn{1}{l|}{}                                       & \multicolumn{1}{l|}{at most $\kappa$ links failures).}   \\ 
			\hline
			\multicolumn{1}{|l|}{\textit{Permanent}} & \multicolumn{1}{l|}{Crash failures.}                 & \multicolumn{1}{l|}{}                                   \\ \hline
			\vspace*{0.25em}
		\end{tabular}
		\includegraphics[clip=true,scale=0.7]{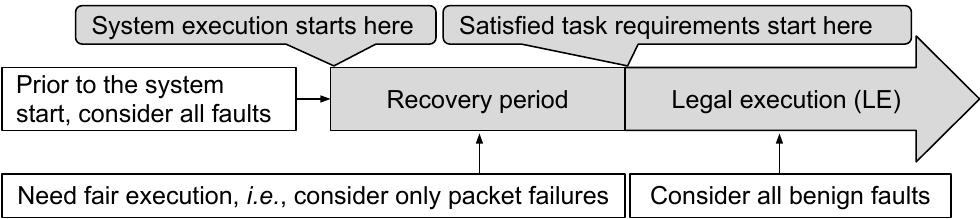}\\
		% source: https://docs.google.com/drawings/d/14xB2G0LG4kCBt-1Bj_TP-EpbYh_rpiE1lIexyYRsz7U/edit
		\caption{\label{fig:self-stab-SDN}\footnotesize{The table above details our fault model and the chart illustrates when each fault set is relevant. The chart's gray shapes represent the system execution, and the white boxes specify the failures considered to be possible at different execution parts and recovery guarantees of the proposed self-stabilizing algorithm. The set of benign faults includes both packet failures and crash failures.}}
	\end{\VCalgSize}
\end{figure*}

\subsection{The Fault Model and Self-stabilization}%Fault model}
We model a failure occurrence as a step that the environment takes rather than the algorithm. 

\smallskip

\noindent \textbf{Benign failures.~~}
\label{sec:benignFailures}
When the occurrence of a failure cannot cause the system execution to lose legality, \ie to leave $LE$, we refer to that failure as a benign one (Figure~\ref{fig:self-stab-SDN}).

\smallskip

\noindent \textit{Node failure.~~}
We consider \emph{crash} failures, in which nodes stop taking steps. We assume that the system does not have access to failure detectors~\cite{DBLP:journals/jacm/ChandraT96}. We assume that the number of failing nodes is bounded by $f$ and that $2f<n$ for the sake of guaranteeing correctness~\cite{DBLP:books/mk/Lynch96}. In Section~\ref{sec:bounded} we discuss how to relax this assumption, \eg assuming that $f=0$ or considering \emph{fail-stop failures}, which assume access to unreliable failure detectors~\cite{DBLP:journals/jacm/ChandraT96}.

%In the absence of transient faults, failing nodes can simply crash (or fail-stop and then resume at some arbitrary time), as in Delporte-Gallet \etal~\cite{DBLP:journals/tpds/Delporte-Gallet18}. In the presence of transient faults, we assume that failing nodes resume within some unknown finite time. The latter assumption is needed \emph{only} for recovering from transient faults; we bring more details in Section~\ref{sec:seldomFair}. 

\smallskip

\noindent \textit{Communication failures and fairness.~~}
%
%We consider solutions that are oriented towards asynchronous message-passing systems and thus they are oblivious to the time in which the packets arrive and departure. 
%
We assume that the communication channels are prone to packet failures, such as omission, duplication, reordering. However, we assume that if $p_i$ sends a message infinitely often to $p_j$, node $p_j$ receives that message infinitely often. The latter is called the \emph{fair communication} assumption. This paper uses a refinement of the fair communication assumption in which if $p_i$ sends a message of type $X$ and with an index number of at least $x$ infinitely often to $p_j$, node $p_j$ receives that message of type $X$ and with an index number of at least $x$  infinitely often.

%For example, the proposed algorithm sends infinitely often $\mathrm{GOSSIP}$ messages from any processor to any other. Despite the possible loss of messages, the communication fairness assumption implies that every processor receives infinitely often $\mathrm{GOSSIP}$ messages from any non-failing processor. 

\smallskip

\noindent \textbf{Arbitrary transient faults.~~}
We consider any violation of the assumptions according to which the system was designed to operate. We refer to these violations and deviations as \emph{arbitrary transient faults} and assume that they can corrupt the system state arbitrarily (while keeping the program code intact). The occurrence of an arbitrary transient fault is rare. Thus, as in~\cite{DBLP:journals/cacm/Dijkstra74,DBLP:books/mit/Dolev2000,DBLP:series/synthesis/2019Altisen}, our model assumes that the last arbitrary transient fault occurs before the system execution starts. Moreover, it leaves the system to start in an arbitrary state. 

%\subsection{Dijkstra's self-stabilization criterion}
%
%\label{sec:Dijkstra}
%
An algorithm is \textit{self-stabilizing} with respect to the task of $LE$, when every (unbounded) execution $R$ of the algorithm reaches within a finite period a suffix $R_{legal} \in LE$ that is legal. That is, Dijkstra~\cite{DBLP:journals/cacm/Dijkstra74} requires that $\forall R:\exists R': R=R' \circ R_{legal} \land R_{legal} \in LE \land |R'| \in \N$, where the operator $\circ$ denotes that $R=R' \circ R''$ concatenates $R'$ with $R''$.

%\subsection{Execution fairness vs. wait-free guarantees}
%
%\label{sec:Dijkstra}
%
%We say that a system execution is \emph{fair} when every step that is applicable infinitely often is executed infinitely often and fair communication is kept. Self-stabilizing algorithms often assume that their executions are fair. Wait-free algorithms guarantee that operations (that were invoked by non-failing nodes) always terminate in the presence of asynchrony and any number of node failures. Note that fair executions do not consider node failures (that were not detected by the system which then excluded these failing nodes from reentering the system, as in~\cite{DBLP:conf/netys/DolevGMS17}). This work assumes execution fairness during the period in which the system recovers from the occurrence of the last arbitrary transient fault. In other words, the system is wait-free only during legal executions, which are absent from arbitrary transient faults. Moreover, the system recovery from arbitrary transient faults is not wait-free, but this bounded recovery period occurs only once throughout the system execution.

\subsection{Complexity Measures}
\label{sec:timeComplexity}
The main complexity measure of self-stabilizing algorithms, called \emph{stabilization time}, is the time it takes the system to recover after the occurrence of the last transient fault. 

%This is referred to as the \emph{recovery time} (or stabilization time).  

\smallskip
\noindent
\textbf{Message round-trips and iterations of self-stabilizing algorithms.~~}
\label{sec:messageRoundtrips}
The correctness proof depends on the nodes' ability to exchange messages during the periods of recovery from transient faults. The proposed solution considers communications that follow the pattern of request-reply, \ie $\mathsf{MESSAGE}\text{-}\mathsf{TYPE}$ and $\mathsf{MESSAGE}\text{-}\mathsf{TYPEack}$ messages, as well as $\mathsf{GOSSIP}$ messages for which the algorithm does not send replies. The definitions of our complexity measures use the notion of a message round-trip for the cases of request-reply messages and the term algorithm iteration.
 
We give a detailed definition of \emph{round-trips} as follows. Let $p_i \in \sP$ and $p_j \in \sP \setminus \{p_i\}$. Suppose that immediately after system state $c$, node $p_i$ sends a message $m$ to $p_j$, for which $p_i$ awaits a reply. At system state $c'$, that follows $c$, node $p_j$ receives message $m$ and sends a reply message $r_m$ to $p_i$. Then, at system state $c''$, that follows $c'$, node $p_i$ receives $p_j$'s response, $r_m$. In this case, we say that $p_i$ has completed with $p_j$ a round-trip of message $m$. 

It is well-known that self-stabilizing algorithms cannot terminate their execution and stop sending messages~\cite[Chapter 2.3]{DBLP:books/mit/Dolev2000}. Moreover, their code includes a do-forever loop. Thus, we define a \emph{complete iteration} of a self-stabilizing algorithm. 
Suppose that node $p_i$ sends a gossip message infinitely often to $p_j \in \sP \setminus \{p_i\}$ (regardless of the message payload). Suppose that immediately after the state $c_{begin}$, node $p_i$ takes a step that includes the execution of the first line of the do-forever loop, and immediately after system state $c_{end}$, it holds that: (i) $p_i$ has completed the iteration it has started immediately after $c_{begin}$ (regardless of whether it enters branches), (ii) every request-reply message $m$ that $p_i$ has sent to any non-failing node $p_j \in \sP$ during the iteration (that has started immediately after $c_{begin}$) has completed its round trip, and (iii) it includes the arrival of at least one gossip message from $p_i$ to any non-failing $p_j \in \sP \setminus \{p_i\}$. In this case, we say that $p_i$'s complete iteration (with round-trips) starts at $c_{begin}$ and ends at $c_{end}$. The term \emph{non-failing} refers to a node that was recently working correctly and connected to the network. The term \emph{non-faulty} refers to a node that always works correctly and connected to the network.

\smallskip
\noindent
\textbf{Cost measures: asynchronous cycles and the happened-before relation.~~}
\label{ss:asynchronousCycles}
We say that a system execution is \emph{fair} when every step that is applicable infinitely often is executed infinitely often and fair communication is kept. Since asynchronous systems do not consider the notion of time, we use the term (asynchronous) cycles as an alternative way to measure the period between two system states in a fair execution. The first (asynchronous) cycle (with round-trips) of a fair execution $R=R' \circ R''$ is the shortest prefix $R'$ of $R$, such that each non-failing node executes at least one complete iteration in $R'$, where $\circ$ is the concatenation operator. The second cycle in execution $R$ is the first cycle in execution $R''$, and so on.

\begin{remark}
	\label{ss:first asynchronous cycles}
	For the sake of simple presentation of the correctness proof, when considering fair executions, we assume that any message that arrives in $R$ without being transmitted in $R$ does so within $\bigO(1)$ asynchronous cycles in $R$. 
\end{remark}

Lamport~\cite{DBLP:journals/cacm/Lamport78} defined the happened-before relation as the least strict partial order on events for which: (i) If steps $a, b \in R$ are taken by processor $p_i \in \sP$, $a \rightarrow b$ if $a$ appears in $R$ before $b$. (ii) If step $a$ includes sending a message $m$ that step $b$ receives, then $a \rightarrow b$. Using the happened-before definition, one can create a directed acyclic (possibly infinite) graph $G_R:(V_R,E_R)$, where the set of nodes, $V_R$, represents the set of system states in $R$. Moreover, the set of edges, $E_R$, is given by the happened-before relation. In this paper, we assume that the weight of an edge that is due to cases (i) and (ii) are zero and one, respectively. When there is no guarantee that execution $R$ is fair, we consider the weight of the heaviest directed path between two system state $c,c' \in R$ as the cost measure between $c$ and $c'$.      

\begin{algorithm*}[t!]
\begin{\VCalgSize}	
	\noindent \textbf{Definitions of $\preceq$:\label{ln:preceq}}
For integers $t$ and $t'$: $(\bullet,t) \preceq (\bullet,t') \iff t \leq t'$; For arrays $tab$ and $tab'$ of $(\bullet, integer)$: $tab \preceq tab' \iff \forall p_k \in \sP : tab[k] \preceq tab'[k]$; Also, $a \prec b \equiv a \preceq b \land a \neq b$\;

\smallskip

\noindent \textbf{local variables initialization:\label{ln:varStart}}
	
		$ssn := 0; ts := 0$\tcc*{snapshot, resp., write operation indices} 
$reg := [\bot, \ldots ,\bot]$\tcc*{shared registers ($\bot$ is smaller than any possibly written value)\label{ln:var}} 

	\smallskip

\noindent \textbf{macro}  $\mathrm{merge}(Rec)$ \label{ln:0merge} {
\lFor{$p_k \in \sP$}{$reg[k]\gets \max (\{reg[k]\} \cup \{r[k] \mid r \in Rec\})$\label{ln:0regGetsMaxRegCupMid}}
}

	\smallskip
	
\noindent \textbf{operation}  $\sWrite(v)$ \label{ln:0operationWriteV} \Begin{
	$ts \gets ts+1; reg[i] \gets (v, ts)$; \textbf{let} $lReg:=reg$\label{ln:0tsPlusOne}\;
	\lRepeat{$\mathrm{WRITEack}(\mathnormal{regJ} \succeq lReg)$ received from a majority\label{ln:0waitUntilWRITEackReg}}{$\mathsf{broadcast~} \mathrm{WRITE}(lReg)$\label{ln:0broadcastWRITEreg};}
	
	$\mathrm{merge}(Rec)$ \textbf{where} $Rec$ is the set of $reg$ arrays received at line~\ref{ln:0waitUntilWRITEackReg}\label{ln:0mergeRecWrite}\;
	$\Return()$\label{ln:0writeReturn}\;
}

\noindent \textbf{operation}  $\sSnapshot()$ \label{ln:0operationSnapshot} \Begin{
	\Repeat{$prev = reg$}{
		\textbf{let} $prev := reg$; $ssn \gets ssn + 1$\label{ln:0prevSsnGetsRegSsnPlusOne}\;
		\lRepeat{$\mathrm{SNAPSHOTack}(\bullet, \mathnormal{ssnJ} =ssn)$ received from a majority\label{ln:0waitUntilSNAPSHOTackReg}}{$\mathsf{broadcast} ~\mathrm{SNAPSHOT}(reg, ssn)$\label{ln:0ssnPlusOne};}
		$\mathrm{merge}(Rec)$ \textbf{where} $Rec$ is the set of $reg$ arrays received at line~\ref{ln:0waitUntilSNAPSHOTackReg}\label{ln:0mergeRecSnapshot}\;
	}\label{ln:0prevEqualsReg}
	$\Return(reg)$\label{ln:0snapshotReturn}\;
}

	\smallskip
	
\textbf{upon} message $\mathrm{WRITE}(\mathnormal{regJ})$  \textbf{arrival} \textbf{from} $p_j$ \label{ln:0arrivalWRITE} \Begin{
	\lFor{$p_k \in \sP$}{$reg[k] \gets  \max_{\preceq}(reg[k], \mathnormal{regJ}[k])$\label{ln:0kDotRegGetsMaxPreceqRegJWrite}}
	\textbf{send} $\mathrm{WRITEack}(reg)$ \textbf{to} $p_j$\label{ln:0sendWRITEackReg}\;
}

\smallskip

\textbf{upon} message $\mathrm{SNAPSHOT}(\mathnormal{regJ}, ssn)$ \textbf{arrival} \textbf{from} $p_j$ \label{ln:0arrivalSNAPSHOT} \Begin{
	\lFor{$p_k \in \sP$}{$reg[k] \gets \max_{\preceq}\{ reg[k], \mathnormal{regJ}[k] \}$\label{ln:0kDotRegGetsMaxPreceqRegJSnapshoot}}
	\textbf{send} $\mathrm{SNAPSHOTack}(reg, ssn)$ to $p_j$\label{ln:0sendSNAPSHOTackRegSsn}\;
}
\end{\VCalgSize}

\caption{\label{alg:0disCongif}The non-self-stabilizing and non-blocking algorithm by Delporte-Gallet \etal~\cite{DBLP:journals/tpds/Delporte-Gallet18} that emulates a snapshot object; code for $p_i$}

% in $\mathcal{CAMP}_{n,t}[t < n/2]$; code for processor $p_i$	
	
\end{algorithm*}

\section{Background}
\label{sec:back}
For the sake of completeness, we review the solutions of Delporte-Gallet \etal~\cite{DBLP:journals/tpds/Delporte-Gallet18}. 
%
%The non-blocking liveness criterion~\cite{DBLP:journals/toplas/HerlihyW90} requires that, at any time, and regardless of how many operations are applied to the snapshot object, at least one operation terminates. 

\subsection{The non-blocking algorithm by Delporte-Gallet \etal}
\label{sec:nonBlock}
The non-blocking solution to snapshot object emulation by Delporte-Gallet \etal~\cite[Algorithm~1]{DBLP:journals/tpds/Delporte-Gallet18} allows all write operations to terminate regardless of the invocation patterns of the other write or snapshot operations (as long as the invoking processors do not fail during the operation). However, for the case of snapshot operations, termination is guaranteed only if eventually the system execution reaches a period in which there are no concurrent write operations. 
Algorithm~\ref{alg:0disCongif} presents Delporte-Gallet \etal~\cite[Algorithm~1]{DBLP:journals/tpds/Delporte-Gallet18}. That is, we have changed some of the notation of Delporte-Gallet to fit the presentation style of this paper. 

%\EMS{[[@@ Fix me. @@ Moreover, we use the $\mathsf{broadcast}$ primitive according to its definition in Section~\ref{sec:spec}.]]}

\smallskip

\noindent \emph{Local variables.~~}
The node state appears in lines~\ref{ln:varStart} to~\ref{ln:var} and automatic variables (which are allocated and deallocated automatically when program flow enters and leaves the variable's scope) are defined using the let keyword, \eg the variable $prev$ (line~\ref{ln:0prevSsnGetsRegSsnPlusOne}). Also, when a message arrives, we use the parameter name $\mathrm{xJ}$ to refer to the arriving value for the message field $x$. Moreover, we denote variable $X$'s value at $p_i$ by $X_i$.

Processor $p_i$ stores an array $reg$ of $|\sP|$ elements (line~\ref{ln:var}), such that the $k$-th entry stores the most recent information about processor $p_k$'s object value and $reg[i]$ stores $p_i$'s actual object value. Every entry is a pair of the form $(v,ts)$, where the field $v$ is a $\nu$-bits object value and $ts$ is an unbounded integer that stores the object timestamp. The values of $ts$ serve as the index of $p_i$'s write operations. Similarly, $p_i$ maintains an index for the snapshot operations, $ssn$ (sequence number). Algorithm~\ref{alg:0disCongif} defines also the relation $\preceq$ that compares $(v,ts)$ and $(v',ts')$ according to the write operation indices (line~\ref{ln:preceq}).

\smallskip

\noindent \textbf{The $\sWrite()$ operation.~~}
Algorithm~\ref{alg:0disCongif}'s $\sWrite()$ operation appears in lines~\ref{ln:0operationWriteV} to~\ref{ln:0writeReturn} (client-side) and lines~\ref{ln:0arrivalWRITE} to~\ref{ln:0sendWRITEackReg} (server-side).
The client-side operation $\sWrite(v)$ stores the pair $(v,ts)$ in $reg[i]$ (line~\ref{ln:0tsPlusOne}), where $p_i$ is the invoking processor and $ts$ is a unique operation index. The primitive $\mathsf{broadcast}$ sends to all the processors in $\sP$ the message $\mathrm{WRITE}$ about $p_i$'s local perception of $reg$'s value. 

Upon the arrival of a $\mathrm{WRITE}$ message to $p_i$ from $p_j$ (line~\ref{ln:0arrivalWRITE}), the server-side code is run. Processor $p_i$ updates $reg$ according to the timestamps of the arriving values (line~\ref{ln:0kDotRegGetsMaxPreceqRegJWrite}). Then, $p_i$ replies to $p_j$ with the message $\mathrm{WRITEack}$ (line~\ref{ln:0sendSNAPSHOTackRegSsn}), which includes $p_i$'s local perception of the system shared registers.

Getting back to the client-side,  $p_i$ repeatedly broadcasts the message $\mathrm{WRITE}$ to all processors in $\sP$ until it receives replies from a majority of them (line~\ref{ln:0broadcastWRITEreg}). Once that happens, it uses the arriving values for keeping $reg$ up-to-date (line~\ref{ln:0mergeRecWrite}).

\smallskip

\noindent \textbf{The $\sSnapshot()$ operation.~~}
Algorithm~\ref{alg:0disCongif}'s $\sSnapshot()$ operation appears in
lines~\ref{ln:0operationSnapshot} to~\ref{ln:0snapshotReturn} (client-side) and lines~\ref{ln:0arrivalSNAPSHOT} to~\ref{ln:0sendSNAPSHOTackRegSsn} (server-side).
Recall that Delporte-Gallet \etal~\cite[Algorithm~1]{DBLP:journals/tpds/Delporte-Gallet18} is non-blocking with respect to the snapshot operations as long as there are no concurrent write operations. Thus, the client-side is written in the form of a repeat-until loop. Processor $p_i$ tries to query the system for the most recent value of the shared registers. The success of such attempts, \ie satisfying the if-statement condition of line~\ref{ln:0prevEqualsReg}, depends on the above assumption. Therefore, before each such broadcast, $p_i$ copies $reg$'s value to $prev$ (line~\ref{ln:0prevSsnGetsRegSsnPlusOne}) and exits the repeat-until loop only when the updated value of $reg$ indicates that there are no concurrent write operations. We note that the absence of concurrent writes implies the success of the atomic snapshot since it considers all previous write operations.

\begin{figure*}[t!]
	\begin{center}
		\includegraphics[page=1,scale=\figScale]{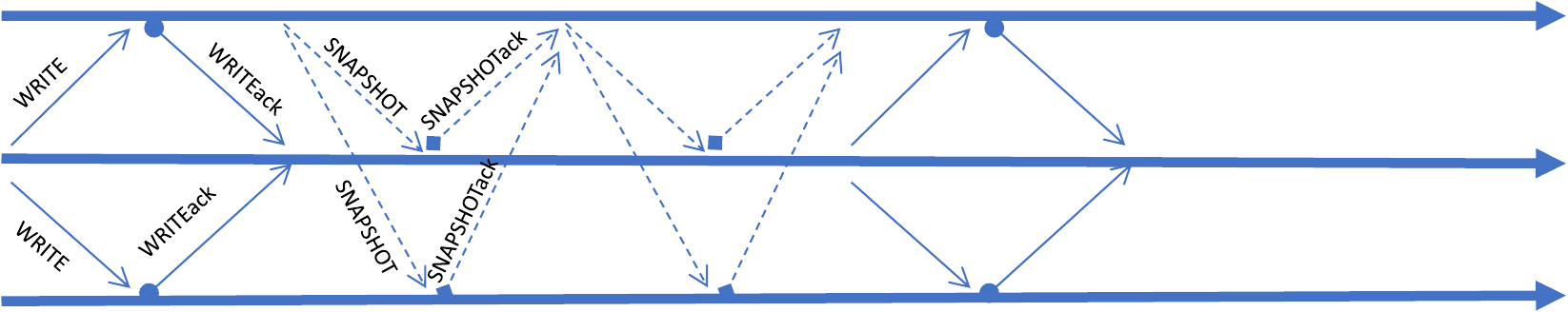}
		~~\\
		~~\\
		\includegraphics[page=3,scale=\figScale]{DrawingsCrop.pdf}
	\end{center}
	\caption{\label{fig:0nonBlockExec}Examples of Algorithm~\ref{alg:0disCongif}'s executions. The upper drawing illustrates a case of a terminating snapshot operation (dashed line arrows) that occurs between two write operations (solid line arrows). The acknowledgments of these messages are arrows that start with circles and squares, respectively. The lower drawing illustrates a case in which every execution of line~\ref{ln:0ssnPlusOne} occurs concurrently with write operations (regardless of whether the algorithm is self-stabilizing or not). Thus, snapshot operations cannot terminate.}
\end{figure*}

Figure~\ref{fig:0nonBlockExec} depicts two examples of Algorithm~\ref{alg:0disCongif}'s execution. The upper drawing illustrates a write operation that is followed by a snapshot operation and then a second write operation. We use this example when comparing with algorithms~\ref{alg:9disCongif},~\ref{alg:disCongif} and~\ref{alg:terminating}. The lower drawing illustrates a case of an unbounded sequence of write operations that disrupts a snapshot operation, which does not terminate for an unbounded period.

\begin{algorithm*}[t!]
\begin{\VCalgSize}

\noindent \textbf{local variables initialization:}
		$ssn,sns, ts := 0$\tcc*{snapshot and write operation indices}
$reg := [\bot, \ldots ,\bot]$\tcc*{shared registers ($\bot$ is smaller than any possibly written value)}
$writePending \gets \bot$\tcc*{stores $p_i$'s write task}
$\textbf{foreach } k,s : \repSnap[k, s] := \bot$\tcc*{stores $p_k$'s snapshot task result for index $s$}

%$ts \gets 0; sns \gets 0; ssn\gets 0; reg  \gets [\bot, \ldots ,\bot]; \textbf{foreach } x, y : \repSnap[x, y] = \bot$\;

\smallskip

\noindent \textbf{macro}  $\mathrm{merge}(Rec)$ \label{ln:9merge} {
\lFor{$p_k \in \sP$}{$reg[k]\gets \max (\{reg[k]\} \cup \{r[k] \mid r \in Rec\})$\label{ln:9regGetsMaxRegCupMid}}
}

\smallskip

\textbf{do forever} \Begin{\label{ln:9doForever}
\lIf{$(writePending \neq \bot)$}{$\mathrm{baseWrite}(writePending);writePending \gets \bot$\label{ln:backGroundWrite}}
\If{\emph{(there are messages $\mathrm{SNAP}()$ received and not yet processed)\label{ln:unboundedBuffer}}}{
\textbf{let} $\mathrm{SNAP}(source, sn)$ \textbf{be} {the oldest of these messages}\;
$\baseSnapshot(source, sn)$\;
\textbf{wait} \textbf{until} $(\repSnap[source, sn] \neq \bot)$\label{ln:waitUntilReadSnap};
}
}

\smallskip

\noindent \textbf{operation} $\sWrite(v)$ \Begin{
$writePending \gets v$;
\textbf{wait} \textbf{until} $(writePending = \bot); \Return()$\label{ln:preWrite};
}
 
\smallskip

\noindent \textbf{operation} $\sSnapshot()$ \Begin{
$sns \gets sns + 1; \label{ln:9snspp}\mathsf{reliableBroadcast} ~\mathrm{SNAP}(i, sns)$\label{ln:reliableBroadcastout}\;
\textbf{wait} \textbf{until} $(\repSnap[i, sns] \neq \bot); \Return(\repSnap[i, sns])$\label{ln:9repSnapIsns}\;
}

\smallskip

\noindent \textbf{function}  $\mathrm{baseWrite}(v)$  \label{ln:9baseWrite} \Begin{
$ts \gets ts + 1; reg[i] \gets (ts,v)$; \textbf{let} $lReg:=reg$\;
\lRepeat{$\mathrm{WRITEack}(\mathnormal{regJ} \succeq lReg)$ received from a majority\label{ln:9waitUntilWRITEack}}{$\mathsf{broadcast~} \mathrm{WRITE}(lReg)$\label{ln:9broadcastWRITEreg2};}
	$\mathrm{merge}(Rec)$ \textbf{where} $Rec$ is the set of $reg$ arrays received at line~\ref{ln:9waitUntilWRITEack}\label{ln:9mergeRecWrite2}\;
}

\noindent \textbf{function}  $\baseSnapshot(s, t)$ \label{ln:9baseSnapshot}\Begin{
\While{$\repSnap[s, t] = \bot$\label{ln:whileStart}}{
\textbf{let} $prev := reg$; $ssn \gets ssn + 1$\label{ln:prevRegSsnPlusOne}\; 
		\Repeat{$\mathrm{SNAPSHOTack}(\mathnormal{sJ}=s, \mathnormal{tJ}=t, \bullet, \mathnormal{ssnJ}=ssn)$ \emph{received from a majority)}\label{ln:9waitUntilSNAPSHOTackReg2}}{$\mathsf{broadcast}~\mathrm{SNAPSHOT}(s, t, reg, ssn)$\label{ln:9waitUntilSNAPSHOTack};}
$\mathrm{merge}(Rec)$ \textbf{where} $Rec$ is the set of $reg$ arrays received at line~\ref{ln:9waitUntilSNAPSHOTackReg2}\;
\lIf{$prev = reg$}{$\mathsf{reliableBroadcast} ~\mathrm{END}(s, t, prev)$\label{ln:prevReg}}
}
%$\Return()$\;
}

\smallskip

\textbf{upon} message $\mathrm{WRITE}(\mathnormal{regJ})$  \textbf{arrival} \textbf{from} $p_j$ \Begin{
\lFor{$p_k \in \sP$}{$reg[k] \gets  \max_{\prec_{sn}}(reg[k], \mathnormal{regJ}[k])$}
\textbf{send} $\mathrm{WRITEack}(reg)$ \textbf{to} $p_j$\;
}

\smallskip

\textbf{upon} message $\mathrm{SNAPSHOT}(s, t, \mathnormal{regJ}, \mathnormal{ssnJ})$ \textbf{arrival} \textbf{from} $p_j$ \Begin{
\lFor{$p_k \in \sP$}{$reg[k] \gets \max_{\prec_{sn}} (reg[k], \mathnormal{regJ}[k])$}
\textbf{send} $\mathrm{SNAPSHOTack}(s, t, reg, \mathnormal{ssnJ})$ \textbf{to} $p_j$\;
}

\smallskip

\textbf{upon} message $\mathrm{END}(s, t, val)$ \textbf{arrival} \textbf{from} $p_j$ \textbf{do} {
$\repSnap[s, t] \gets val$\label{ln:uponEND}\;
}	
	
\caption{\label{alg:9disCongif}The non-self-stabilizing and always-terminating algorithm by Delporte-Gallet \etal~\cite{DBLP:journals/tpds/Delporte-Gallet18} that emulates a snapshot object; code for $p_i$}	
	
\end{\VCalgSize}
	
\end{algorithm*}

\subsection{The always-terminating algorithm by Delporte-Gallet \etal}
\label{sec:alwaysTerm}
Delporte-Gallet \etal~\cite[Algorithm~2]{DBLP:journals/tpds/Delporte-Gallet18} guarantees termination for any invocation pattern of write and snapshot operations, as long as the invoking processors do not fail during these operations. Its advantage over Delporte-Gallet \etal~\cite[Algorithm~1]{DBLP:journals/tpds/Delporte-Gallet18} is that it can deal with an infinite number of concurrent write operations. This is because it guarantees the non-blocking progress criterion for the snapshot operations. We present~\cite[Algorithm~2]{DBLP:journals/tpds/Delporte-Gallet18} in Algorithm~\ref{alg:9disCongif} using the presentation style of this paper. We review Algorithm~\ref{alg:9disCongif} while pointing out some key challenges that exist when considering self-stabilization.

\smallskip

\noindent \textbf{High-level overview.~~}
Delporte-Gallet \etal~\cite[Algorithm~2]{DBLP:journals/tpds/Delporte-Gallet18} uses a job-stealing scheme for allowing rapid termination of snapshot operations. Processor $p_i \in \sP$ starts its $\sSnapshot$ operation by queueing this new task at all processors $p_j \in \sP$. Once $p_j$ receives $p_i$'s new task and when that task reaches the queue front, $p_j$ starts the $\baseSnapshot(s, t)$ procedure, which is similar to Algorithm~\ref{alg:0disCongif}'s $\sSnapshot()$ operation. This joint participation in all snapshot operations makes sure that all processors are aware of all on-going snapshot operations.

This joint awareness allows the system processors to make sure that no write operation can stand in the way of on-going snapshot operations. To that end, the processors wait until the oldest snapshot operation terminates before proceeding with later operations. Specifically, they defer write operations that run concurrently with snapshot operations. This guarantees termination of snapshot operations via the interleaving and synchronization of snapshot and write operations.

% (but can indefinitely defer write operations for some invocation patterns of snapshot operations that include an infinite number of them).

\smallskip

\noindent \textbf{Detailed description.~~}
Algorithm~\ref{alg:9disCongif} extends Algorithm~\ref{alg:0disCongif} in the sense that it uses all of Algorithm~\ref{alg:0disCongif}'s variables and two additional ones. These are the second operation index, $sns$, and an array $\repSnap$, which $\sSnapshot()$ operations use. The entry $\repSnap[x,y]$ holds the outcome of $p_x$'s $y$-th snapshot operation, where no explicit bound on the number of invocations of snapshot operations is given. 

In the context of self-stabilization, the use of such unbounded variables is not possible. The reasons are that real-world systems have bounded size memory as well as the fact that a single transient fault can bring any counter to its near overflow value and fill up any finite capacity buffer. We discuss the way around this challenge in Section~\ref{sec:bounded}.       

\smallskip

\noindent \emph{The $\sWrite()$ operation and the $\mathrm{baseWrite}()$ function.~~}
Since $\sWrite(v)$ operations are preemptible, $p_i$ cannot always start immediately to write. Instead, $p_i$ stores $v$ in $writePending_i$ (line~\ref{ln:preWrite}). The algorithm then runs the write operation as a background task (line~\ref{ln:backGroundWrite}) using the $\mathrm{baseWrite}()$ function (lines~\ref{ln:9baseWrite} to~\ref{ln:9mergeRecWrite2}).

\smallskip

\noindent \emph{The $\sSnapshot()$ operation.~~}
A call to $\sSnapshot()$ (line~\ref{ln:9snspp}) causes $p_i$ to reliably broadcast, via the primitive $\mathsf{reliableBroadcast}$, a new $sns$ index in a $\mathrm{SNAP}$ message to all processors in $\sP$. Processor $p_i$ then waits for the task’s completion by placing it as a background task (line~\ref{ln:9repSnapIsns}). We note that 
for our proposed solutions %the system settings of the proposed solutions 
we do not assume access to a reliable broadcast mechanism such as $\mathsf{reliableBroadcast}$; see Section~\ref{sec:aussata} for details and an alternative approach that uses safe registers instead of the (more involved) $\mathsf{reliableBroadcast}$ primitive. 

\smallskip

\noindent \emph{The $\baseSnapshot()$ function.~~}
This function essentially follows the $\sSnapshot()$ operation of Algorithm~\ref{alg:0disCongif}. That is, Algorithm~\ref{alg:9disCongif}'s snapshot repeat-until loops iterates until the retrieved $reg$ vector equals to the one that was known prior to the last repeat-until iteration. 
Algorithm~\ref{alg:9disCongif}'s $\mathrm{baseSnapshot()}$ procedure returns after at least one snapshot process has terminated. In detail, processor $p_i$ stores in $\repSnap[s,t]$, via a reliable broadcast of the $\mathrm{END}$ message, the result of the snapshot process (line~\ref{ln:prevReg} and~\ref{ln:uponEND}).

\smallskip

\noindent \emph{Synchronization between the $\mathrm{baseWrite}()$ and $\baseSnapshot()$ functions.~~}
Algorithm~\ref{alg:9disCongif} interleaves the background tasks in a do-forever loop (lines~\ref{ln:backGroundWrite} to~\ref{ln:waitUntilReadSnap}). As long as there is an awaiting write task, processor $p_i$ runs the $\mathrm{baseWrite}()$ function (line~\ref{ln:backGroundWrite}). Also, if there is an awaiting snapshot task, processor $p_i$ selects the oldest task, $(source, sn)$, and uses the $\baseSnapshot(source, sn)$ function. Here, Algorithm~\ref{alg:9disCongif} blocks until $\repSnap[source, sn]$ contains the result of that snapshot task. 

Note that line~\ref{ln:unboundedBuffer} implies that Algorithm~\ref{alg:9disCongif} does not explicitly assume that processor $p_i$ has bounded space for storing $\mathrm{SNAP}$ messages. In the context of self-stabilization, there must be an explicit bound on the size of all memory in use. We discuss how to overcome this challenge in Section~\ref{sec:bounded}.

\begin{figure*}[t!]
	\begin{center}
		\includegraphics[page=4,scale=\figScale]{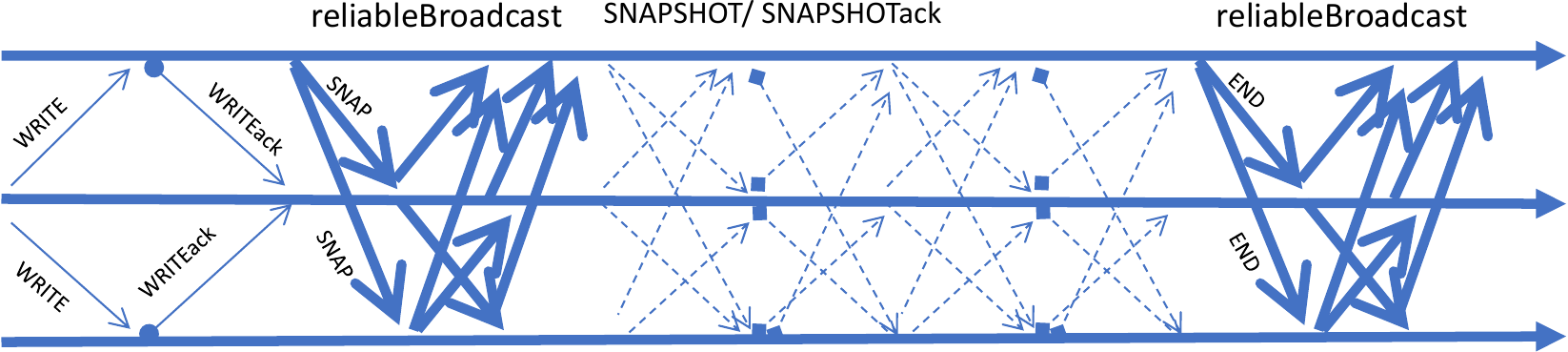}
	\end{center}	
	\caption{\label{fig:9nonBlockExec}Algorithm~\ref{alg:9disCongif}'s execution for the case depicted by the upper drawing of Figure~\ref{fig:0nonBlockExec}. The drawing illustrates a case of a terminating snapshot operation (dashed line arrows) that occurs after (a non-concurrent) write operation (thin solid line arrows). The acknowledgments of these messages are arrows that start with circles and squares, respectively.}
\end{figure*}

Figure~\ref{fig:9nonBlockExec} depicts an example of Algorithm~\ref{alg:9disCongif}'s execution where a write operation is followed by a snapshot operation. Note that each snapshot operation is handled separately and the communication costs of each such operation requires $\bigO(n^2)$ messages.

%OLD VERION \EMS{Figure~\ref{fig:9nonBlockExec} depicts an example of Algorithm~\ref{alg:9disCongif}'s execution in which a write operation that is followed by a snapshot operation. Each snapshot is handled separately and the communications of each such operation require $\bigO(n^2)$ messages.}  

%NEW VERSION \EMS{Figure~\ref{fig:9nonBlockExec} depicts an example of Algorithm~\ref{alg:9disCongif}'s execution where a write operation is followed by a snapshot operation. Each snapshot is handled separately and the communications of each such operation requires $\bigO(n^2)$ messages.}  

\section{An Unbounded Self-stabilizing Non-blocking Algorithm}
\label{sec:aussnba}
We propose Algorithm~\ref{alg:disCongif} as an elegant extension of Delporte-Gallet \etal~\cite[Algorithm~1]{DBLP:journals/tpds/Delporte-Gallet18}; 
we have simply added 
%the boxed code lines 
the boxed code lines to Algorithm~\ref{alg:disCongif}. Algorithms~\ref{alg:0disCongif} and~\ref{alg:disCongif} differ in their ability to deal with stale information, which can appear in the system after the occurrence of transient faults. 

%when starting in an arbitrary state. Note that we model the appearance of stale information as the result of  and assume that they occur only before the system starts running. 

\subsection{Algorithm description}
\label{sec:def}
Algorithm~\ref{alg:disCongif} considers the case in which any of $p_i$'s operation indices, $ssn_i$ and $ts_i$, is smaller than some other $ssn$ or $ts$ value, say, $ssn_m$, $reg_{i}[i].ts$, $reg_{j}[i].ts$ or $reg_m[i].ts$, where $X_m$ appears in the $X$ field of some in-transit message. 

For the case of corrupted $ssn$ values, $p_i$'s client-side simply ignores arriving message with $ssn$ values that do not match $ssn_i$ (line~\ref{ln:waitUntilSNAPSHOTackReg}). For the sake of clarity of our proposal, we also periodically remove any stored snapshot replies that their $ssn$ fields are not equal to $ssn_i$ (line~\ref{ln:deleteSNAPSHOTack}). 

For the case of corrupted $ts$ values, $p_i$'s do-forever loop makes sure that $ts_i$ is not smaller than $reg_i[i].ts$ (line~\ref{ln:tsGetsMaxTsRegTs}) before gossiping to every processor $p_j \in \sP$ its local copy of $p_j$'s shared register (line~\ref{ln:sendGossip}). Also, upon the arrival of such gossip messages, Algorithm~\ref{alg:disCongif} merges the arriving information with the local one (line~\ref{ln:GOSSIPupdate}). Moreover, when replies from write or snapshot messages arrive to $p_i$, it merges the arriving $ts$ value with the one in $ts_i$ (line~\ref{ln:tsRegMaxTsRegTs}).

%On the presentation side, we clarify that the code lines~\ref{ln:0broadcastWRITEreg} to~\ref{ln:0mergeRecWrite} and lines~\ref{ln:0ssnPlusOne} to~\ref{ln:0mergeRecSnapshot} are equivalent to lines~$3$ to~$5$ and lines~$10$ and~$12$ of \cite[Algorithm~1]{DBLP:journals/tpds/Delporte-Gallet18}, respectively, because it is merely a more detailed description of the code described in~\cite[Algorithm~1]{DBLP:journals/tpds/Delporte-Gallet18}.\OL{This paragraph should not be in this section about algo3, since the paragraph here is about algo1.}

Figure~\ref{fig:disCongifExec} depicts an example of Algorithm~\ref{alg:disCongif}'s execution in which a write operation is followed by a snapshot operation. Note that gossip messages do not interfere with write and snapshot operations.  

\begin{figure*}[t!]
	\begin{center}
		\includegraphics[page=2,scale=\figScale]{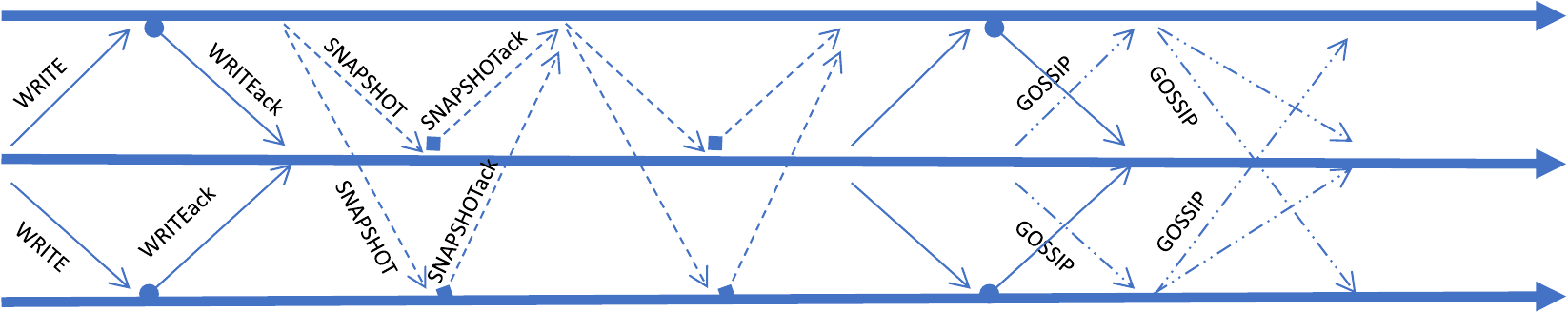}
	\end{center}
	\caption{\label{fig:disCongifExec}{Algorithm~\ref{alg:disCongif}'s execution for the case depicted in the upper drawing of Figure~\ref{fig:0nonBlockExec}. The drawing illustrates a case of a terminating snapshot operation (dashed line arrows) that occurs between two write operations (solid line arrows). The acknowledgments of these messages are arrows that start with circles and squares, respectively. The gossip messages appear in the dashed-double-dot lines.}}
\end{figure*}

\begin{algorithm*}[t!]
\begin{\VCalgSize}	
	\noindent \textbf{Definitions of $\preceq$:}
For integers $t$ and $t'$: $(\bullet,t) \preceq (\bullet,t') \iff t \leq t'$; For arrays $tab$ and $tab'$ of $(\bullet, integer)$: $tab \preceq tab' \iff \forall p_k \in \sP : tab[k] \preceq tab'[k]$; Also, $a \prec b \equiv a \preceq b \land a \neq b$\label{ln:def3}\;   

\smallskip

\noindent \textbf{local variables initialization (optional in the context of self-stabilization):}

	$ssn := 0; ts := 0$\tcc*{snapshout, resp., write operation indices} 
	 ~$reg := [\bot, \ldots ,\bot]$\tcc*{shared registers ($\bot$ is smaller than any possibly written value)}

%$ssn := 0; ts := 0; reg[n] := [\bot, \ldots ,\bot]$ ($\bot$ is smaller than any value can be possibly written)\;

	\smallskip

\noindent \textbf{macro}  $\mathrm{merge}(Rec)$ \label{ln:merge} \Begin{
\fbox{$ts\gets \max (\{ts, reg[i].ts\} \cup \{r[i].ts \mid r \in Rec\})$\label{ln:tsRegMaxTsRegTs};}

\lFor{$p_k \in \sP$}{$reg[k]\gets \max (\{reg[k]\} \cup \{r[k] \mid r \in Rec\})$\label{ln:regGetsMaxRegCupMid}}
}

	\smallskip

\noindent \textbf{do forever} \Begin{\label{ln:disCongifDoForever}
	\fbox{\lForEach{$ssn'\neq ssn$}{\textbf{delete} $\mathrm{SNAPSHOTack}(-, ssn')$\label{ln:deleteSNAPSHOTack}}}
	
	\fbox{$ts\gets \max \{ts, reg[i].ts\}$\label{ln:tsGetsMaxTsRegTs};}
	
	\fbox{\lFor{$p_k \in \sP:k\neq i$}{$\mathbf{send}~ \mathrm{GOSSIP}(reg[k])~ \mathsf{to}~ p_k$\label{ln:sendGossip}}}
}		
	\smallskip
	
\noindent \textbf{operation}  $\sWrite(v)$ \label{ln:operationWriteV} \Begin{
	$ts \gets ts+1; reg[i] \gets (v, ts)$; \textbf{let} $lReg:=reg$\label{ln:tsPlusOne}\;
	\lRepeat{$\mathrm{WRITEack}(\mathnormal{regJ} \succeq lReg)$ received from a majority\label{ln:waitUntilWRITEackReg}}{$\mathsf{broadcast~} \mathrm{WRITE}(lReg)$\label{ln:broadcastWRITEreg};}
	
	$\mathrm{merge}(Rec)$ \textbf{where} $Rec$ is the set of $reg$ arrays received at line~\ref{ln:waitUntilWRITEackReg}\label{ln:mergeRecWrite}\;
	$\Return()$\label{ln:writeReturn}\;
}

\noindent \textbf{operation}  $\sSnapshot()$ \label{ln:operationSnapshot} \Begin{
	\Repeat{$prev = reg$}{
		\textbf{let} $prev := reg$; $ssn \gets ssn + 1$\label{ln:prevSsnGetsRegSsnPlusOne}\;
		\lRepeat{$\mathrm{SNAPSHOTack}(\bullet, \mathnormal{ssnJ} =ssn)$ received from a majority\label{ln:waitUntilSNAPSHOTackReg}}{$\mathsf{broadcast} ~\mathrm{SNAPSHOT}(reg, ssn)$\label{ln:ssnPlusOne};}
		$\mathrm{merge}(Rec)$ \textbf{where} $Rec$ is the set of $reg$ arrays received at line~\ref{ln:waitUntilSNAPSHOTackReg}\label{ln:mergeRecSnapshot}\;
	}
	$\Return(reg)$\label{ln:snapshotReturn}\;
}

\textbf{upon} message $\mathrm{GOSSIP}(\mathnormal{regJ})$ \textbf{arrival} \textbf{from} $p_j$ \Begin{ 
\fbox{$reg[i]\gets \max \{reg[i],\mathnormal{regJ}\};ts\gets \max \{ts, reg[i].ts\}$\label{ln:GOSSIPupdate};}}

	\smallskip
	
\textbf{upon} message $\mathrm{WRITE}(\mathnormal{regJ})$  \textbf{arrival} \textbf{from} $p_j$ \label{ln:arrivalWRITE} \Begin{
	\lFor{$p_k \in \sP$}{$reg[k] \gets  \max_{\preceq}(reg[k], \mathnormal{regJ}[k])$\label{ln:kDotRegGetsMaxPreceqRegJWrite}}
	\textbf{send} $\mathrm{WRITEack}(reg)$ \textbf{to} $p_j$\label{ln:sendWRITEackReg}\;
}

\textbf{upon} message $\mathrm{SNAPSHOT}(\mathnormal{regJ}, ssn)$ \textbf{arrival} \textbf{from} $p_j$ \label{ln:arrivalSNAPSHOT} \Begin{
	\lFor{$p_k \in \sP$}{$reg[k] \gets \max_{\preceq}\{ reg[k], \mathnormal{regJ}[k] \}$\label{ln:kDotRegGetsMaxPreceqRegJSnapshoot}}
	\textbf{send} $\mathrm{SNAPSHOTack}(reg, ssn)$ to $p_j$\label{ln:sendSNAPSHOTackRegSsn}\;
}
\end{\VCalgSize}

\caption{\label{alg:disCongif}Self-stabilizing algorithm for non-blocking snapshot object; code for $p_i$. The boxed code lines mark the added code to Algorithm~\ref{alg:0disCongif}.}	
\end{algorithm*}

\subsection{Correctness}
\label{sec:cor1}
Although the extension performed to Algorithm~\ref{alg:0disCongif} for obtaining Algorithm~\ref{alg:disCongif} includes only few changes, proving convergence and closure for Algorithm~\ref{alg:disCongif} is not straightforward. We proceed with the details.

\smallskip

\noindent \textbf{Notation and definitions}
Definition~\ref{def:safeConfig} refers to $p_i$'s timestamps and snapshot sequence numbers, where $p_i \in \sP$. The set of $p_i$'s timestamps includes $ts_{i}$, $reg_{i}[i].ts$, $reg_{j}[i].ts$ and the value of $reg_m[i].ts$ in the payload of any message $m$ that is in transient in the system. The set of $p_i$'s snapshot sequence numbers includes $ssn_{i}$ and the value of $ssn_m$ in the payload of any $\mathrm{SNAPSHOT}$ and $\mathrm{SNAPSHOTack}$ message $m$ from, and respectively, to $p_i$ that is in transient in the system.

\begin{definition}[Algorithm~\ref{alg:disCongif}'s consistent operation indices]
\label{def:safeConfig}
(i) Let $c$ be a system state in which $ts_i$ is greater than or equal to any $p_i$'s timestamp values in the variables and fields related to $ts$. We say that the $ts$'s timestamps are consistent in $c$.
(ii) Let $c$ be a system state in which $ssn_i$ is greater than or equal to any $p_i$'s snapshot sequence numbers in the variables and fields related to $ssn$. We say that the $ssn$'s snapshot sequence numbers are consistent in $c$.	
\end{definition}

Theorems~\ref{thm:mainConvergence} and~\ref{thm:closure1} show the properties required by the self-stabilization design criteria.  

\begin{theorem}[\textbf{Algorithm~\ref{alg:disCongif}'s convergence}]
	\label{thm:mainConvergence}
	Let $R$ be a fair and unbounded execution of Algorithm~\ref{alg:disCongif}. Within $\bigO(1)$ asynchronous cycles in $R$, the system reaches a state $c \in R$ in which $ts$' timestamps and $ssn$'s snapshot sequence numbers are consistent in $c$.
\end{theorem} 
\begin{proof} The proof of the theorem follows by Lemmas~\ref{thm:eventuallyAscendingSent} and~\ref{thm:sequenceNumberRecovery}. The operation $\sWrite(v)$ (line~\ref{ln:operationWriteV}) starts by changing the state of the invoking processor, \ie incrementing $ts$ (line~\ref{ln:tsPlusOne}). Lemma~\ref{thm:eventuallyAscendingSent} shows the conditions that holds immediately before that increment occurs.   
	
\begin{lemma}[\textbf{Timestamp convergence}]
	\label{thm:eventuallyAscendingSent}
   Let $R$ be an unbounded fair execution of Algorithm~\ref{alg:disCongif}. Within $\bigO(1)$ asynchronous cycles in $R$, the system reaches a state $c \in R$ in which the value of $ts_i$ is greater than or equal to any $p_i$'s timestamp value. Moreover, suppose that node $p_i$ takes a step immediately after $c$ that includes the execution of line~\ref{ln:tsPlusOne}. Then in $c$, it holds that $ts_{i}=reg_{i}[i].ts=reg_{j}[i].ts$ as well as for every messages $m\in \mathit{channel}_{i,j},\mathit{channel}_{j,i}$ that is in transit from $p_i$ to $p_j$ or $p_j$ to $p_i$ it holds that $m.reg[i].ts=ts_{i}$.
\end{lemma} 

\begin{proof}
	Claims~\ref{thm:neverRemove},~\ref{thm:neverRemoveAgain},~\ref{thm:GOSSIPregJ} and~\ref{thm:tsititTR} prove the lemma. Claim~\ref{thm:neverRemove} denotes by $X_{i,\ell}$ the $\ell$-th value stored in $X_i$ during $R$, where $\ell \in \mathbb{N}$.
	\begin{claim}
		\label{thm:neverRemove}
		The sequences $ts_{i,\ell}$, $reg_{i,\ell}[i].ts$, $reg_{j,\ell}[i].ts$, $reg_{i,\ell}[i]$ and $reg_{j,\ell}[i]$ are non-decreasing.
	\end{claim}
	\begin{claimProof}
	We note that Algorithm~\ref{alg:disCongif} does only the following actions on $ts$ and $reg$ fields: increment (line~\ref{ln:tsPlusOne}) and merge using the $\max$ function (lines~\ref{ln:tsRegMaxTsRegTs},~\ref{ln:regGetsMaxRegCupMid},~\ref{ln:tsGetsMaxTsRegTs},~\ref{ln:mergeRecSnapshot},~\ref{ln:GOSSIPupdate},~\ref{ln:kDotRegGetsMaxPreceqRegJWrite} and~\ref{ln:kDotRegGetsMaxPreceqRegJSnapshoot}). That is, there are no assignments. Thus, the claim is true, because the value of these fields is never decremented during $R$.
	\end{claimProof}
		
	\begin{claim}
		\label{thm:neverRemoveAgain}
	 Within $\bigO(1)$ asynchronous cycles in $R$, $ts_{i} \geq reg_{i}[i].ts$.
	\end{claim}
	\begin{claimProof}
		Since $R$ is unbounded, it holds that node $p_j \in \sP$ calls lines~\ref{ln:tsGetsMaxTsRegTs},~\ref{ln:sendGossip} and~\ref{ln:GOSSIPupdate} for an unbounded number of times during $R$. Recall the line numbers that may change the value of $ts_{i}$ and $reg_{i}[i].ts$, cf. the proof of Claim~\ref{thm:neverRemove}. Note that only line~\ref{ln:tsPlusOne} change the value of $ts_{i}$, via an increment (thus the claim states an inequality rather than an equality) whereas lines~\ref{ln:tsRegMaxTsRegTs},~\ref{ln:tsGetsMaxTsRegTs} and~\ref{ln:GOSSIPupdate} update $ts_i$ and $reg_{i}[i].ts$ by taking the maximum value of $ts_i$ and $reg_{i}[i].ts$. The rest of the proof is implied by Claim~\ref{thm:neverRemove}, and the fact that $p_i$ executes lines~\ref{ln:tsGetsMaxTsRegTs},~\ref{ln:sendGossip} and~\ref{ln:GOSSIPupdate} at least once in every $\bigO(1)$ asynchronous cycles.
	\end{claimProof}
	
	Algorithm~\ref{alg:disCongif} sends $\mathrm{GOSSIP}$ messages in line~\ref{ln:sendGossip}, request messages in lines~\ref{ln:broadcastWRITEreg} and~\ref{ln:ssnPlusOne} as well as replies in lines~\ref{ln:sendWRITEackReg} and~\ref{ln:sendSNAPSHOTackRegSsn}. Claim~\ref{thm:GOSSIPregJ}'s proof considers lines~\ref{ln:broadcastWRITEreg} and~\ref{ln:ssnPlusOne} in which $p_i$ sends a request message to $p_j$, whereas Claim~\ref{thm:tsititTR}'s proof considers lines~\ref{ln:sendGossip},~\ref{ln:sendWRITEackReg} and~\ref{ln:sendSNAPSHOTackRegSsn} in which $p_j$ replies or gossips to $p_i$.

\begin{claim}
		\label{thm:GOSSIPregJ}
	 Let $m\in \mathit{channel}_{i,j}$ be a message on transit from $p_i$ to $p_j$ (during the first asynchronous cycles of $R$) and $reg_m$ the value of the $reg$ filed in $m$, where $p_i,p_j \in \sP$ are non-failing nodes. Within $\bigO(1)$ asynchronous cycles in $R$, $reg_{i}[i].ts \geq reg_{m}[i].ts$ and $reg_{i}[i].ts \geq  \mathnormal{regJ}[i].ts$ whenever $p_j$ raises the events $\mathrm{GOSSIP}(\mathnormal{regJ})$, $\mathrm{WRITE}(\mathnormal{regJ})$ or $\mathrm{SNAPSHOT}(\mathnormal{regJ},\bullet)$. 
	 
	 %$\mathrm{WRITE}(\mathnormal{regJ})$ or $\mathrm{SNAPSHOT}(\mathnormal{regJ}, ssn)$.
	\end{claim}
	\begin{claimProof}
		Suppose during the first asynchronous cycles of $R$, node $p_i$ indeed sends message $m$, \ie $m$ does not appear in $R$'s starting system state. Let $a_{k} \in R$ be the first step in $R$ in which $p_i$ calls line~\ref{ln:sendGossip},~\ref{ln:broadcastWRITEreg} or~\ref{ln:ssnPlusOne} and for which there is a step $a_{depart,k} \in R$, which appears in $R$ after $a_{k}$ and in which message $m$ is sent (in a packet by the end-to-end or quorum protocol). Note that the value of $reg_{m}[i].ts$ in the message payload is defined by the value of $reg_{i}[i].ts$ in the system state that immediately precedes $a_{k}$. The rest of the proof relies on the fact that until $m$ arrives to $p_j$, the invariant $reg_{i}[i].ts \geq reg_{m}[i].ts$ holds (due to Claim~\ref{thm:neverRemove}).
		
		%\sout{, if there is any such step}
		
		 Let $a_{arrival,k} \in R$ be the first step that appears after $a_{depart,k}$ in $R$ in which the node at $p_j$ delivers the packet that $a_{depart,k}$ transmits the message $m$ (if there are several such packets, consider the last to arrive). By the assumption that $R$ is unbounded and fair as well as the definition of asynchronous cycles, step $a_{arrival,k}$ indeed appears in $R$ within $\bigO(1)$ asynchronous cycles. During $a_{arrival,k}$, node $p_j$ raises the message delivery event $\mathrm{GOSSIP}(\mathnormal{regJ})$ (when $a_{k}$ considers line~\ref{ln:GOSSIPupdate}), $\mathrm{WRITE}(\mathnormal{regJ})$ (when $a_{k}$ considers line~\ref{ln:arrivalWRITE}) or 		 
		 $\mathrm{SNAPSHOT}(\mathnormal{regJ}, ssn)$ (when $a_{k}$ considers line~\ref{ln:arrivalSNAPSHOT}), such that $reg_{i}[i].ts \geq reg_{m}[i].ts = \mathnormal{regJ}[i].ts$.
		 
		 Suppose that step $a_{k}$ does not appear in $R$, \ie $m$ appears in $R$'s starting system state. By the definition of asynchronous cycles with round-trips (Remark~\ref{ss:first asynchronous cycles}), within $\bigO(1)$ asynchronous cycles, all messages in transit to $p_j$ arrive (or leave the communication channel). Thus, within $\bigO(1)$ asynchronous cycles, the system start a suffix, $R'$, of $R$ in which all delivered messages during $R'$ were indeed sent during $R$. Thus, the above proof holds with respect to messages received during $R'$ that were sent during $R$. 
	\end{claimProof}
	
	%\sout{Immediately after that, the system starts an execution in which this claim holds trivially.}
	
	\begin{claim}
	\label{thm:tsititTR}
	Let $m \in \mathit{channel}_{j,k}$ be a message on transit from $p_j$ to $p_k$ (during the first asynchronous cycles of $R$) and $reg_{m}$ the value of the $reg$ filed in $m$, where $p_i,p_j,p_k \in \sP$ are non-failing nodes and $i=k$ may or may not hold. Within $\bigO(1)$ asynchronous cycles, $reg_{j}[i].ts \geq reg_{m}[i].ts$ and $reg_{i}[i].ts \geq  \mathnormal{regJ}[i].ts$ whenever $p_k$ raises the events $\mathrm{GOSSIP}(\mathnormal{regJ})$, $\mathrm{WRITEack}(\mathnormal{regJ})$ or $\mathrm{SNAPSHOTack}(\mathnormal{regJ},\bullet)$. 
	\end{claim}
	\begin{claimProof}
	Suppose during the first asynchronous cycles of $R$, node $p_k$ indeed sends message $m$, \ie $m$ does not appear in $R$'s starting system state. Let $a_{k} \in R$ be the first step in $R$ in which $p_k$ calls line~\ref{ln:sendGossip},~\ref{ln:sendWRITEackReg} or~\ref{ln:sendSNAPSHOTackRegSsn} and for which there is a step $a_{depart,k} \in R$, which appears in $R$ after $a_{k}$. Note that the value of $reg_{m}[i].ts$ in the message payload is defined by the value of $reg_{i}[i].ts$ in the system state that immediately precedes $a_{k}$. The rest of the proof relies on the fact that until $m$ arrives to $p_j$, the invariant $reg_{i}[i].ts \geq reg_{m}[i].ts$ holds (due to Claim~\ref{thm:neverRemove}). 

	Let $a_{arrival,k} \in R$ be the first step that appears after $a_{depart,k}$ in $R$, if there is any such step, in which the node at $p_j$ delivers the packet that $a_{depart,k}$ transmits the message $m$ (if there are several such packets, consider that the last to arrive). By the assumption that $R$ is fair and the definition of asynchronous cycles, step $a_{arrival,k}$ appears in $R$ within $\bigO(1)$ asynchronous cycles. During $a_{arrival,k}$, node $p_j$ raises the message delivery event $\mathrm{GOSSIP}(\mathnormal{regJ})$ (when $a_{k}$ considers line~\ref{ln:sendGossip}) 
	 $\mathrm{WRITEack}(\mathnormal{regJ})$ (when $a_{k}$ considers line~\ref{ln:sendWRITEackReg}) or $\mathrm{SNAPSHOTack}(\mathnormal{regJ},\bullet)$ (when $a_{k}$ considers line~\ref{ln:sendSNAPSHOTackRegSsn}), such that $reg_{i}[i].ts \geq reg_{m}[i].ts = \mathnormal{regJ}[i].ts$.		
	
For the case in which step $a_{k}$ does not appear in $R$, the proof follows the same arguments that appear in the proof of Claim~\ref{thm:GOSSIPregJ}.
\end{claimProof}

\noindent This completes the proof of the lemma.
\end{proof}

\begin{lemma}[\textbf{Sequence number convergence}]
	\label{thm:sequenceNumberRecovery}
	Let $R$ be a fair and unbounded execution of Algorithm~\ref{alg:disCongif}. Within $\bigO(1)$ asynchronous cycles in $R$, the system reaches a state $c_x \in R$ in which the value of $ssn_i$ is greater than or equal to any $p_i$'s snapshot sequence number.
\end{lemma} 

\begin{proof}
		Claims~\ref{thm:neverRemoveS},~\ref{thm:GOSSIPregJS} and~\ref{thm:tsititTRS} prove the lemma.	
		\begin{claim}
		\label{thm:neverRemoveS}
		The sequence $ssn_{i,\ell}$ is non-decreasing.
	\end{claim}
	\begin{claimProof}
		Algorithm~\ref{alg:disCongif} only increments (line~\ref{ln:prevSsnGetsRegSsnPlusOne}), and assigns (lines~\ref{ln:ssnPlusOne} and~\ref{ln:sendSNAPSHOTackRegSsn}) $ssn$ values. Thus, the claim is true, because the value of this field is never decremented during $R$.
	\end{claimProof}
	
	\smallskip
	
	\noindent The proofs of Claims~\ref{thm:GOSSIPregJS} and~\ref{thm:tsititTRS} are followed by similar arguments to the ones that appear in the proofs of Claims~\ref{thm:GOSSIPregJ} and~\ref{thm:tsititTR}.  
	
		\begin{claim}
		\label{thm:GOSSIPregJS}
		Let $m\in \mathit{channel}_{i,j}$ be a $\mathrm{SNAPSHOT}$ message on transit from $p_i$ to $p_j$ (during the first asynchronous cycles of $R$) that includes the filed $ssn$ with the value of $ssn_m$. Within $\bigO(1)$ asynchronous cycles, $ssn_{i} \geq ssn_{m}$ including when $p_j$ raises the event $\mathrm{SNAPSHOT}(\mathnormal{regJ}, ssn_{m})$.
	\end{claim}

	\begin{claim}
	\label{thm:tsititTRS}
	Let $m \in \mathit{channel}_{j,i}$ be a $\mathrm{SNAPSHOTack}$ message on transit from $p_j$ to $p_i$ (during the first asynchronous cycles of $R$) and $ssn_{m}$ the value of the $reg$ filed in $m$. Within $\bigO(1)$ asynchronous cycles, $ssn_{i} \geq ssn_{m}$ including when $p_j$ raises the event $\mathrm{SNAPSHOTack}(\mathnormal{regJ}, ssn_m)$. 
	\end{claim}
%\end{proof}	
\noindent This completes the proof of the lemma, which completes the proof of the theorem.
\end{proof}
\end{proof}

\smallskip

Theorem~\ref{thm:closure1} shows Algorithm~\ref{alg:disCongif}'s termination and linearization properties. It also bounds the cost of $\sWrite()$ and $\sSnapshot()$ operations. Since termination is not always guaranteed, the latter bound considers the case in which the execution of Algorithm~\ref{alg:disCongif} includes no $\sWrite()$ operations.

\begin{theorem}[\textbf{Algorithm~\ref{alg:disCongif}'s termination and linearization}]
	\label{thm:closure1}
	Let  $R$ be an execution of Algorithm~\ref{alg:disCongif} that starts in system state $c$, in which the timestamps and snapshot sequence numbers are consistent (Definition~\ref{def:safeConfig}). Execution $R$ is legal with respect to the task of emulating snapshot objects. Moreover, every $\sWrite()$ operation has a cost measure of two (Section~\ref{ss:asynchronousCycles}). Furthermore, after the occurrence of the last $\sWrite()$ operation, the cost measure of a $\sSnapshot()$ operation is four. 
\end{theorem} 

\begin{proof}
We start the proof by observing the differences between Algorithms~\ref{alg:0disCongif} and~\ref{alg:disCongif}.
%
% and its counterpart in~\cite[Algorithm~1]{DBLP:journals/tpds/Delporte-Gallet18}.
%
Note that Algorithms~\ref{alg:0disCongif} and~\ref{alg:disCongif} use the same variables. Any message that Algorithm~\ref{alg:0disCongif} sends, also Algorithm~\ref{alg:disCongif} sends. The only exception are gossip messages: Algorithm~\ref{alg:disCongif} sends gossip messages, while Algorithm~\ref{alg:0disCongif} does not.
The two algorithms differ in line~\ref{ln:tsRegMaxTsRegTs}, lines~\ref{ln:deleteSNAPSHOTack} to~\ref{ln:sendGossip} and line~\ref{ln:GOSSIPupdate}. 

The next step in the proof is to show that during $R$, any step that includes the execution of line~\ref{ln:GOSSIPupdate} does not change the state of the calling processor. This is due to the fact the every timestamp uniquely couples an object value (line~\ref{ln:tsGetsMaxTsRegTs}) and that timestamps are consistent in every system state throughout $R$ (Lemma~\ref{thm:eventuallyAscendingSent}).

The rest of the proof considers $Alg_{noGOSSIP}$ that is obtained from the code of Algorithm~\ref{alg:disCongif} by the removal of lines~\ref{ln:sendGossip} and~\ref{ln:GOSSIPupdate}, in which the gossip messages are sent and received, respectively. We use this definition to show that $Alg_{noGOSSIP}$ simulates  Algorithm~\ref{alg:0disCongif}. This means that from the perspective of its external behavior (\ie its requests, replies and failure events), any trace of $Alg_{noGOSSIP}$ has a trace of Algorithm~\ref{alg:0disCongif} (as long as indeed the starting system state, $c$, encodes consistent timestamps and snapshot sequence numbers). Since Algorithm~\ref{alg:0disCongif} satisfies the task of emulating snapshot objects, it holds that $Alg_{noGOSSIP}$ also satisfies the task. This implies that Algorithm~\ref{alg:disCongif} satisfies the task as well.

Recall the fact the every timestamp uniquely couples an object value (line~\ref{ln:tsGetsMaxTsRegTs}) as well as that timestamps and snapshot sequence numbers are consistent in every system state throughout $R$ (Lemma~\ref{thm:eventuallyAscendingSent}). These facts imply that also line~\ref{ln:tsRegMaxTsRegTs}, and lines~\ref{ln:deleteSNAPSHOTack} to~\ref{ln:tsGetsMaxTsRegTs} do not change the state of the calling node.

Observe that each $\sWrite()$ requires a single round of quorum communication. Thus, every $\sWrite()$ operation has a cost measure of two. Furthermore, after the last $\sWrite()$ operations, it takes at most two rounds of quorum communication for any $\sSnapshot()$ operation to terminate, \ie has a cost measure of four. However, before the occurrence last $\sWrite()$ operations, there is no bound on the communication costs.
\end{proof}

\begin{algorithm*}[t!]
	%\noindent \textbf{constants:}
	%$dfltRep :=(0,\bot,\bot)$\;
	\begin{\VCalgSizeSmall}

		\noindent \textbf{input:} $\delta$ a number of observed concurrent writes after which write operations block temporarily; 
		
		\noindent \textbf{local variables initialization (optional in the context of self-stabilization):} $ts := 0$ is $p_i$'s write operation index; $ssn,sns := 0$ are $p_i$'s snapshot operation indices; $reg[n]:=  [\bot, \ldots, $ $ \bot]$ buffers all shared registers; $writePending \gets \bot$ stores $p_i$'s write task; $\pndSnap[n] := [(0,\bot,\bot),\ldots,(0,\bot,\bot)]$ control variables of snapshot operations; each entry form is $(sns, vc, \mathit{fnl})$, where $sns$ is an index, $vc$ is a vector clock that timestamps the snapshot operation $sns$, and $\mathit{fnl}$ is the operation's returned value\label{ln:alg4Variables}; 
		
		%\smallskip
		
		\noindent \textbf{macro} $\mathrm{VC}:=[ts_k]_{p_k\in\sP}$ \textbf{where} $ts_k:=0$ \textbf{when} $reg[k]=\bot$ \textbf{otherwise} $reg[k]=(\bullet,ts_k)$\label{ln:vc}\; 
		
		%\smallskip
		
		\noindent \textbf{macro} $exceeds(k) := (\delta \leq \sum_{\ell\in\{1,\ldots,n\}}\mathrm{VC}[\ell]-\pndSnap[k].vc[\ell])$;
		
		\noindent \textbf{macro} $\exceed(k) := \{(k,\pndSnap[k].sns,\pndSnap[k].vc):\pndSnap[k].\mathit{fnl} = \bot \land ((\delta=0\land \pndSnap[k].sns>0)\lor (\pndSnap[k].vc\neq \bot \land exceeds(k)))\}$\label{ln:exceedDeltaK};
		
		\noindent \textbf{macro} $\exceed := (\bigcup_{p_k \in \sP} \exceed(k) )\cup \{(i,\pndSnap[i].sns,\pndSnap[i].vc):\pndSnap[i].sns> 0 \land \pndSnap[i].\mathit{fnl}= \bot\}$\label{ln:exceedDelta};
		
		%\smallskip
		
		\noindent \textbf{macro}  $\sStore(A)$ \lRepeat{\emph{majority of $\mathrm{SAVEack}(\mathnormal{AJ}$$=$$\{(k,s)$$:$$(k,$$s,$$\bullet)$$\in$$A\})$ arrived}}{$\mathsf{broadcast}$ $\mathsf{SAVE}(A)$\label{ln:safeStoreSend}} 
		
		%\smallskip
		
		\noindent \textbf{macro}  $\mathrm{merge}(Rec)$ \label{ln:merge2} \{$ts\gets \max (\{ts, reg[i].ts\} \cup \{r[i].ts \mid r \in Rec\})$\label{ln:tsRegMaxTsRegTs2};
		\lFor{$p_k \in \sP$}{$reg[k]\gets \max (\{reg[k]\} \cup \{r[k] \mid r \in Rec\})$\label{ln:regGetsMaxRegCupMid2}\}}
		
		%\smallskip
		
		\textbf{do forever} \Begin{\label{ln:terminatingDoForever}
			
			\lForEach{$ssn'\neq ssn$}{\textbf{delete} $\mathrm{SNAPSHOTack}(-, ssn')$\label{ln:deleteSNAPSHOTack2}}
			
			$(ts,sns)\gets (\max \{ts, reg[i].ts\},\max \{sns, \pndSnap[i].sns\})$\label{ln:tsGetsMaxTsRegTs2}\;
			
			\lFor{$k \in \{1,\ldots,n\}: \pndSnap[k].vc \not \preceq \mathrm{VC}$, \textbf{where} line~\ref{ln:def3} defines the relation $\preceq$}{$\pndSnap[k].vc \gets \bot$\label{ln:vcFix}}
			
			\lIf{$sns \neq \pndSnap[i].sns$}{$\pndSnap[i]\gets (sns,\bot,\bot)$\label{ln:snsNegRepSnap}}
			
			\lFor{$p_k \in \sP:k\neq i$}{$\mathbf{send}~ \mathrm{GOSSIP}(reg[k],\pndSnap[k].sns)~ \mathsf{to}~ p_k$\label{ln:sendGossip2}}
			
			\lIf{$writePending \neq \bot$}{$\{\mathrm{baseWrite}(writePending);writePending \gets \bot;\}$\label{ln:writePendingExceedDelta}}
			\lIf{$\exceed \neq \emptyset$}{$\baseSnapshot(\exceed)$\label{ln:letSexceedDeltaCall}}
		}
		
		%\smallskip
		
		\noindent \textbf{operation} $\sWrite(v)$ \{\label{ln:tsEqtsPlusOne}$writePending \gets v$;
		\textbf{wait} \textbf{until} $(writePending = \bot); \Return()$;\label{ln:waitUntilwritePendingBotsWritets}\}
	
		%\smallskip
		
		\noindent \textbf{operation} $\sSnapshot()$ \label{ln:operationSnapshot2}\Begin
		{
			$(sns, \pndSnap[i]) \gets (sns + 1,(sns,\bot,\bot))$\label{ln:snsPLUSonepndSnap}; 
			\textbf{wait} \textbf{until} $(\pndSnap[i].\mathit{fnl} \neq \bot);\label{ln:terminatingWaitUntilSnapshot} \Return(\pndSnap[i].\mathit{fnl})$\;
		}
		
		%\smallskip
		
		\noindent \textbf{function}  $\mathrm{baseWrite}(v)$ \{$ts \gets ts + 1; reg[i] \gets (ts,v)$;\label{ln:tsIncrementWrite} \textbf{let} $lReg:=reg$;
		\lRepeat{$\mathrm{WRITEack}(\mathnormal{regJ} \succeq lReg)$ received from a majority\label{ln:waitUntilWRITEackReg2}}{$\mathsf{broadcast~} \mathrm{WRITE}(lReg)$\label{ln:broadcastWRITEreg2}; $\mathrm{merge}(Rec)$ \textbf{where} $Rec$ is the received $reg$ arrays\label{ln:mergeRecWrite2}\}}
				
		%\smallskip
		
		\noindent \textbf{function}  $\baseSnapshot(S)$ \Begin{
			\Repeat{$(S \cap \Delta)=\emptyset \lor ((S \cap \Delta) =(i, \bullet) \land \pndSnap[i].sns> 0 \land \pndSnap[i].\mathit{fnl}= \bot \land \neg exceeds(i))$\label{ln:outerStop}}{
				$ssn \gets ssn + 1$\label{ln:prevSsnGetsRegSsnPlusOne2};
				\textbf{let} $prev:=reg$\label{ln:letSprime}; \Repeat{$(S \cap \Delta)=\emptyset$ \emph{or majority of} $(\mathrm{SNAPSHOTack}(\bullet, \mathnormal{ssnJ} =ssn)$ \emph{arrived)}\label{ln:waitUntilSNAPSHOTackReg2}}{
					$\mathsf{broadcast} ~\mathrm{SNAPSHOT}((S \cap \Delta), reg, ssn)$\label{ln:ssnPlusOne2}\;
				}
				$\mathrm{merge}(Rec)$ \textbf{where} $Rec$ is the set of $reg$ arrays received at line~\ref{ln:waitUntilSNAPSHOTackReg2}\label{ln:waitUntilSNAPSHOTackReg2a}\;
				
				\lIf{$prev = reg \land (S \cap \Delta) \neq \emptyset$\label{ln:prevEqReg}}{$\sStore(\{(k, \pndSnap[k].sns, prev):(k, s, \bullet) \in (S \cap \Delta)\})$\label{ln:call2sStore}}
				
				\lElseIf{$((i, \bullet) \in (S \cap \Delta)) \land (\pndSnap[i].vc=\bot)$\label{ln:vcUpdateCond}}{$\pndSnap[i].vc\gets \mathrm{VC}$\label{ln:vcUpdate}}
			}
		}
		
		%\smallskip
		
		\textbf{upon} message $\mathrm{SAVE}(\mathnormal{AJ})$ \textbf{arrival} \textbf{from} $p_j$\label{ln:safeStoreArrival} 
		\Begin{
			\ForEach{$(k, s, r) \in \mathnormal{AJ}$}{			
			\lIf{$\pndSnap[k]=(s, \bullet,\bot)$}{$\pndSnap[k].\mathit{fnl}\gets r$\label{ln:safeStoreUpon}}
			\lElseIf{$\pndSnap[k].sns<s $}{$\pndSnap[k]\gets (s,\bot,r)$\label{ln:safeStoreUponMore}}
		}	
						
			\textbf{send} $\mathrm{SAVEack}(\{(k, s) :(k, s, \bullet) \in \mathnormal{AJ}\})$ \textbf{to} $p_j$\label{ln:safeStoreSendAck}\;
		}
		
		%\smallskip	
		
		\textbf{upon} message $\mathrm{GOSSIP}(\mathnormal{regJ}, \mathnormal{snsJ})$ \textbf{arrival} \textbf{from} $p_j$ \label{ln:GOSSIParrival}\Begin{ 
			$reg[i]\gets \max \{reg[i],\mathnormal{regJ}\};(ts, sns)\gets (\max \{ts, reg[i].ts\}, \max \{sns, \mathnormal{snsJ}\})$\label{ln:GOSSIPupdate2};}
		
		%\smallskip
		\textbf{upon} message $\mathrm{WRITE}(\mathnormal{regJ})$  \textbf{arrival} \textbf{from} $p_j$ \Begin{
			\lFor{$p_k \in \sP$}{$reg[k] \gets  \max_{\prec_{sn}}(reg[k], \mathnormal{regJ}[k])$}
			\textbf{send} $\mathrm{WRITEack}(reg)$ \textbf{to} $p_j$ (*concurent $\baseSnapshot()$ calls need piggybacking with line~\ref{ln:ssnPlusOne2}'s message*)\label{ln:concurentPiggybacking}\;
		}
		
		\textbf{upon} message $\mathrm{SNAPSHOT}(\mathnormal{SJ}, \mathnormal{regJ}, \mathnormal{ssnJ})$ \textbf{arrival} \textbf{from} $p_j$ \label{ln:SNAPSHOTarrival} \Begin{
			\lFor{$p_k \in \sP$}{$reg[k] \gets \max_{\prec_{sn}} (reg[k], \mathnormal{regJ}[k])$\label{ln:SNAPSHOTarrivalEND}}
			
			\lForEach{$(s,sn,vc) \in \mathnormal{SJ}: \pndSnap[s].sns<sn\lor \pndSnap[s]=(sn,\bot,\bot)$}{$\pndSnap[s]\gets (sn,vc,\bot)$\label{ln:vcRemoteUpdate}}
			
			\textbf{let} $A:=\{(k,\pndSnap[k].sns,\pndSnap[k].\mathit{fnl}):(k,s,\bullet) \in \mathnormal{SJ} : ( \pndSnap[k].\mathit{fnl} \neq \bot\lor s<\pndSnap[k].sns ) \}$\label{ln:prepareREsponse}\;
			
			\textbf{send} $\mathrm{SNAPSHOTack}(reg, \mathnormal{ssnJ})$ \textbf{to} $p_j$\label{ln:SNAPSHOTackSend}; \lIf{$A\neq \emptyset$}{\textbf{send} $\mathsf{SAVE}(A)$ \textbf{to} $p_j$ \label{ln:sendSAVEwithA} (* piggyback these messages *)}}

	\end{\VCalgSizeSmall}
	
	\caption{\label{alg:terminating}{Self-stabilizing always-terminating snapshot objects; code for $p_i$}}
	
\end{algorithm*}

\section{An Unbounded Self-stabilizing Always Terminating Algorithm}
\label{sec:aussata}
We propose Algorithm~\ref{alg:terminating} as a variation of Delporte-Gallet \etal~\cite[Algorithm~2]{DBLP:journals/tpds/Delporte-Gallet18}. Algorithms~\ref{alg:9disCongif} and~\ref{alg:terminating} differ mainly in their ability to recover from transient faults. This implies some constraints. For example, Algorithm~\ref{alg:terminating} must have a clear bound on the number of pending snapshot tasks as well as on the number of stored results from snapshot tasks that have already terminated (see Section~\ref{sec:alwaysTerm} for details). For sake of simple presentation, Algorithm~\ref{alg:terminating} assumes that the system needs, for each processor, to cater for at most one pending snapshot task. It turns out that this assumption allows us to avoid the use of a self-stabilizing mechanism for reliable broadcast, as an extension of the non-self-stabilizing reliable broadcast that Delporte-Gallet \etal~\cite[Algorithm~2]{DBLP:journals/tpds/Delporte-Gallet18} use. Instead,  Algorithm~\ref{alg:terminating} uses a simpler mechanism for safe registers. 

The above opens up another opportunity: Algorithm~\ref{alg:terminating} can defer pending snapshot tasks until either (i) at least one processor was able to observe at least $\delta$ concurrent write operations, where $\delta$ is an input parameter, or (ii) no $\delta$ concurrent write operations were observed, \ie $\exceed=\emptyset$ (line~\ref {ln:exceedDelta}). Our intention here is to have $\delta$ as a tunable parameter that balances the latency (with respect to snapshot operations) vs. communication costs. That is, for the case of $\delta$ being a very high (finite) value, Algorithm~\ref{alg:terminating} guarantees termination in a way that resembles~\cite[Algorithm~1]{DBLP:journals/tpds/Delporte-Gallet18}, which uses $\bigO(n)$ messages per snapshot operation, and for the case of $\delta=0$, Algorithm~\ref{alg:terminating} behaves in a way that resembles~\cite[Algorithm~2]{DBLP:journals/tpds/Delporte-Gallet18}, which uses $\bigO(n^2)$ messages per snapshot operation.

\subsection{High-level description}
Algorithms~\ref{alg:9disCongif} uses reliable broadcasts for informing all non-failing processors about new snapshot tasks (line~\ref{ln:reliableBroadcastout}) as well as the results of snapshot tasks that have terminated (line~\ref{ln:prevReg}). Since we assume that each processor can have at most one pending snapshot task, we can avoid the need of using a self-stabilizing mechanism for reliable broadcast. Indeed, Algorithm~\ref{alg:terminating} simply lets every processor disseminate its (at most one) pending snapshot task and use a safe register for facilitating the delivery of the task result to its initiator. That is, once a processor finishes a snapshot task, it broadcasts the result to all processors and waits for replies from a majority of processors, which may possibly include the initiator of the snapshot task (using the macro $\sStore()$, line~\ref{ln:safeStoreSend}). This way, if processor $p_j$ notices that it has the result of an ongoing snapshot task, it sends that result to the requesting processor.

\subsection{Algorithm details}
We review Algorithm~\ref{alg:terminating}'s do-forever loop (lines~\ref{ln:deleteSNAPSHOTack2} to~\ref{ln:letSexceedDeltaCall}), the $\baseSnapshot()$ function together with the dealing of message $\mathrm{SNAPSHOT}$ (lines~\ref{ln:SNAPSHOTarrival} to~\ref{ln:SNAPSHOTarrivalEND}), as well as the macro $\sStore(s,r)$ (line~\ref{ln:safeStoreSend}) together with the dealing of message $\mathrm{SAVE}$ (lines~\ref{ln:safeStoreArrival} to~\ref{ln:safeStoreSendAck}). 

\smallskip

\noindent
\textbf{The do-forever loop.~~}
Algorithm~\ref{alg:terminating}'s do-forever loop (lines~\ref{ln:deleteSNAPSHOTack2} to~\ref{ln:letSexceedDeltaCall}), includes a number of lines for cleaning stale information, such as out-of-synch $\mathrm{SNAPSHOTack}$ messages (line~\ref{ln:deleteSNAPSHOTack2}), out-dated operation indices (line~\ref{ln:tsGetsMaxTsRegTs2}), illogical vector-clocks (line~\ref{ln:vcFix}) or corrupted $\pndSnap$ entries (line~\ref{ln:snsNegRepSnap}). The gossiping of operation indices (lines~\ref{ln:sendGossip2} and~\ref{ln:GOSSIParrival}) also helps to remove stale information (as in Algorithm~\ref{alg:disCongif} but only with the addition of $sns$ values). 

The synchronization between write and snapshot operations (lines~\ref{ln:writePendingExceedDelta} and~\ref{ln:letSexceedDeltaCall}) starts with a write, if there is any such pending task (line~\ref{ln:writePendingExceedDelta}), before running its own snapshot task, if there is any such pending, as well as any snapshot task (initiated by others) for which $p_i$ observed that at least $\delta$ write operations occur concurrently with it (line~\ref{ln:letSexceedDeltaCall}). 

\smallskip

\noindent \textbf{The $\sWrite()$ operation and the $\mathrm{baseWrite}()$ function.~~}
As in Algorithm~\ref{alg:9disCongif}, $p_i$ does not start immediately a write operation. Node $p_i$ permits concurrent write operations by storing $v$ and a unique index in $writePending_i$ (line~\ref{ln:tsEqtsPlusOne}). The algorithm then runs the write operation as a background task (line~\ref{ln:writePendingExceedDelta}) using the $\mathrm{baseWrite}()$ function (line~\ref{ln:waitUntilWRITEackReg2}).

\smallskip

\noindent
\textbf{The $\baseSnapshot()$ function and the $\mathrm{SNAPSHOT}$ message.~~}
Algorithm~\ref{alg:terminating} maintains the state of every snapshot task in the array $\pndSnap$. The entry $\pndSnap_i[k]=(sns,vc,\mathit{fnl})$ includes: (i) the index $sns$ of the most recent snapshot operation that $p_k \in \sP$ has initiated and $p_i$ is aware of, (ii) the vector clock representation of $reg_k$ (\ie just the timestamps of $reg_k$, cf. line~\ref{ln:vc}) and (iii) the final result $\mathit{fnl}$ of the snapshot operation (or $\bot$, in case it is still running). 

The $\baseSnapshot()$ function includes an outer loop part (lines~\ref{ln:prevSsnGetsRegSsnPlusOne2} and~\ref{ln:outerStop}), an inner loop part (lines~\ref{ln:letSprime} to~\ref{ln:waitUntilSNAPSHOTackReg2a}), and a result update part (lines~\ref{ln:prevEqReg} to~\ref{ln:vcUpdate}). The outer loop increments the snapshot index, $ssn$ (line~\ref{ln:prevSsnGetsRegSsnPlusOne2}), so that it can consider a new query attempt by the inner loop. The outer loop ends when (i) there are no more pending snapshot tasks that this call to $\baseSnapshot()$ needs to handle, or (ii) the only pending snapshot task for the current invocation of $\baseSnapshot()$ is the one of $p_i$ and $p_i$ has not observed at least $\delta$ concurrent writes. The inner loop broadcasts $\mathrm{SNAPSHOT}$ messages, which includes all the pending $(S\cap \exceed)$ that are relevant to this call to $\baseSnapshot()$ together with the local current value of $reg$ and the snapshot query index $ssn$. The inner loop ends when acknowledgments are received from a majority of processors and the received values are merged (line~\ref{ln:waitUntilSNAPSHOTackReg2a}). The results are updated by writing to an emulated safe shared register (line~\ref{ln:prevEqReg}) whenever $prev=reg$. In case the results do not allow $p_i$ to terminate its snapshot task (line~\ref{ln:vcUpdate}), Algorithm~\ref{alg:terminating} uses the query results for storing the timestamps in the field $vs$. This allows to balance a trade-off between snapshot operation latency and communication costs, as we explain next.

\smallskip

\indent
\emph{The use of the input parameter $\delta$ for balancing the trade-off between snapshot operation latency and communication costs.~~}
For the case of $\delta=0$, the set $\exceed$ (line~\ref{ln:exceedDelta}) includes all the nodes for which there is no stored result, \ie $\pndSnap[k].\mathit{fnl}=\bot$. Thus, no snapshot tasks are ever deferred, as in Delporte-Gallet \etal~\cite[Algorithm~2]{DBLP:journals/tpds/Delporte-Gallet18}. The case of $\delta>0$ uses the fact that 
Algorithm~\ref{alg:terminating} samples the vector clock value of $reg_i$ and stores it in $\pndSnap[i].vc$ (line~\ref{ln:vcUpdate}) once it had completed at least one iteration of the repeat-until loop (line~\ref{ln:waitUntilSNAPSHOTackReg2} and~\ref{ln:waitUntilSNAPSHOTackReg2a}). This way, we can be sure that the sampling of the vector clock is an event that occurred not before the start of $p_i$'s snapshot operation that has the index of $\pndSnap[i].sns$. 

\smallskip

\indent
\emph{Many-jobs-stealing scheme for reduced blocking periods.~~}
We note that $p_k$'s task is considered active as long as $\pndSnap[k].\mathit{fnl} \neq \bot$. For helping all currently actives snapshot tasks, $p_i$ samples the set of currently pending task $(S_i\cap \exceed_i)$ (line~\ref{ln:letSprime}) before starting the inner repeat-until loop (lines~\ref{ln:letSprime} to~\ref{ln:waitUntilSNAPSHOTackReg2a}). Processor $p_i$ broadcasts from the client-side the $\mathrm{SNAPSHOT}$ message, which includes the most recent snapshot task information, to all processors. The reception of this $\mathrm{SNAPSHOT}$ message on the server-side (lines~\ref{ln:SNAPSHOTarrival} to~\ref{ln:SNAPSHOTarrivalEND}), updates the local information (line~\ref{ln:vcRemoteUpdate}) and prepares the response information (line~\ref{ln:prepareREsponse}) before sending the reply to the client-side (line~\ref{ln:SNAPSHOTackSend}). Note that if the receiver notices that it has the result of an ongoing snapshot task, it sends that result to the requesting processor (line~\ref{ln:sendSAVEwithA}).

\smallskip

\noindent
\textbf{The $\sStore()$ function and the $\mathrm{SAVE}$ message.~~}
The $\sStore()$ function considers a snapshot task that was initiated by processor $p_k \in \sP$. This function is responsible for storing the result $r$ of this snapshot task in a safe register. It does so by broadcasting the client-side message $\mathsf{SAVE}$ to all processors in the system (line~\ref{ln:safeStoreSend}). Upon the arrival of the $\mathsf{SAVE}$ message to the server-side, the receiver stores the arriving information, as long as the arriving information is more recent than the local one. Then, the server-side replies with a $\mathrm{SAVEack}$ message to the client-side, who is waiting for a majority of such replies (line~\ref{ln:safeStoreSend}).   

\begin{figure*}[t!]
	\begin{center}
		\includegraphics[page=5,scale=\figScale]{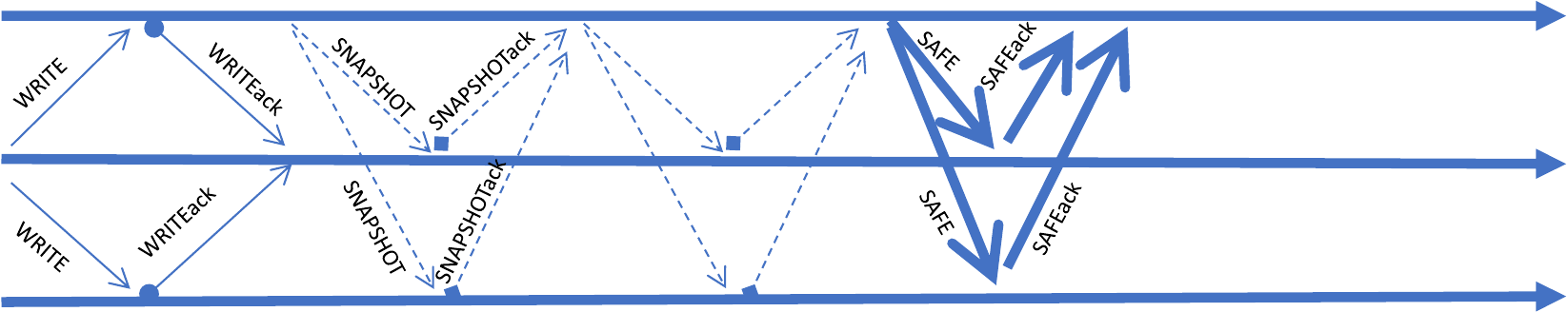}
	\end{center}	
	\begin{center}
		\includegraphics[page=6,scale=\figScale]{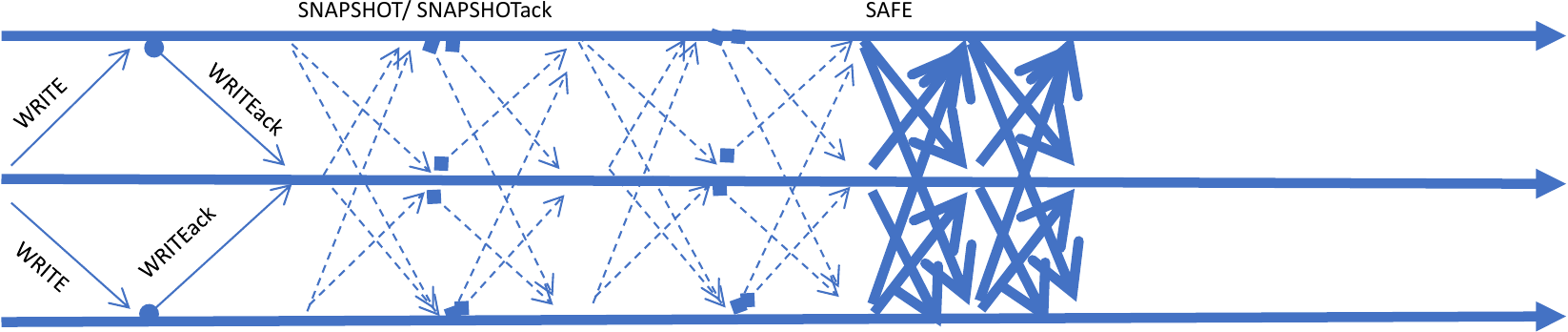}
	\end{center}
	\caption{\label{fig:terminatingExec}The upper drawing depicts an example of Algorithm~\ref{alg:terminating}'s execution for a case that is equivalent to the one depicted in the upper drawing of Figure~\ref{fig:9nonBlockExec}, i.e., only one snapshot operation. The lower drawing illustrates the case of concurrent invocations of snapshot operations by all nodes.}
\end{figure*}

Figure~\ref{fig:terminatingExec} depicts two examples of Algorithm~\ref{alg:terminating}'s execution. In the upper drawing, a write operation is followed by a snapshot operation. Note that fewer messages are considered when comparing to Figure~\ref{fig:9nonBlockExec}'s example. The lower drawing illustrates the case of concurrent invocations of snapshot operations by all nodes. Observe the potential improvement with respect to number of messages (in the upper drawing) and throughput (in the lower drawing) since Algorithm~\ref{alg:9disCongif} uses $\bigO(n^2)$ messages for each snapshot task and handles only one snapshot task at a time.  

\subsection{Correctness}
\label{sec:cor2}
We now prove the convergence (recovery), termination and linearization of Algorithm~\ref{alg:terminating}.

\begin{definition}[Algorithm~\ref{alg:terminating}'s consistent system states and executions]
\label{def:safeConfigTerminating}
(i) Let $c$ be a system state in which $ts_i$ is greater than or equal to any $p_i$'s timestamp values in the variables and fields related to $ts$. We say that the $ts$' timestamps are consistent in $c$.
(ii) Let $c$ be a system state in which $ssn_i$ is greater than or equal to any $p_i$'s snapshot sequence numbers in the variables and fields related to $ssn$. We say that the $ssn$'s snapshot sequence numbers are consistent in $c$.	
(iii) Let $c$ be a system state in which $sns_i$ is greater than or equal to any $p_i$'s snapshot operation index in the variables and fields related to $sns$. Moreover, $\forall p_i \in \sP: sns_i = \pndSnap_i[i].sns$ and $\forall p_i,p_j \in \sP: \pndSnap_j[i].sns \leq \pndSnap_i[i].sns$. We say that the $sns$'s snapshot sequence numbers are consistent in $c$.	
(iv) Let $c$ be a system state in which $\forall p_i,p_k \in \sP:\pndSnap_i[k].vc\preceq \mathrm{VC}_i$ holds, where $\mathrm{VC}_i$ is the returned value from a macro defined in line~\ref{ln:vc} when executed by processor $p_i$. We say that the vector clock values are consistent in $c$.	
We say that system state $c$ is consistent if it is consistent with respect to invariants (i) to (iv). Let $R$ be an execution of Algorithm~\ref{alg:terminating} that all of its system states are consistent and $R'$ be a suffix of $R$. We say that execution $R'$ is consistent (with respect to $R$) if any message arriving in $R'$ was indeed sent in $R$ and any reply arriving in $R'$ has a matching request in $R$.  
\end{definition}

\begin{theorem}[\textbf{Algorithm~\ref{alg:terminating}'s convergence}]
\label{thm:terminatingAlgConvergence}
   Let $R$ be a fair and unbounded execution of Algorithm~\ref{alg:terminating}. Within $\bigO(1)$ asynchronous cycles in $R$, the system reaches a consistent state $c \in R$ (Definition~\ref{def:safeConfigTerminating}).
Within $\bigO(1)$ asynchronous cycles after $c$, the system starts a consistent execution~$R'$.
\end{theorem}
\begin{proof}
Note that Lemmas~\ref{thm:eventuallyAscendingSent} and~\ref{thm:sequenceNumberRecovery} imply invariants (i), and respectively, (ii) of  Definition~\ref{def:safeConfigTerminating} also for the case of Algorithm~\ref{alg:terminating}, because they use the similar code lines for asserting these invariants. 

We now consider the proof of invariant (iii) of Definition~\ref{def:safeConfigTerminating}. 
Note that the variables and fields of $sns$ and the data structure $\pndSnap$ in Algorithm~\ref{alg:terminating} follow the same patterns of information as the variables and fields of $ts$ and the data structure $reg$ in Algorithm~\ref{alg:disCongif}. Moreover, within one asynchronous cycle, every processor $p_i \in \sP$ executes line~\ref{ln:snsNegRepSnap} at least once. Therefore, the proof of invariant (iii) can follow similar arguments to the ones appearing in the proof of Lemma~\ref{thm:eventuallyAscendingSent}. Specifically, $\forall p_i,p_j \in \sP: \pndSnap_j[i].sns \leq \pndSnap_i[i].sns$ holds due to arguments that appear in the proof of Claim~\ref{thm:GOSSIPregJ} with respect to the variables and the fields of $ts$ and the structure $reg$.

The proof of invariant (iv) is implied by the fact that within one asynchronous cycle, every processor $p_i \in \sP$ executes line~\ref{ln:vcFix} at least once and the fact that $\mathrm{VC}_i$ is assigned to $\pndSnap_i[k].vc$ in line~\ref{ln:vcUpdate}. Note that these are the only lines of code that assign values to $\pndSnap_i[k].vc$ and that value of every entry in $\mathrm{VC}_i$ is not decreasing (Claim~\ref{thm:neverRemove}).

By the definition of asynchronous cycles (Section~\ref{ss:asynchronousCycles}), within one asynchronous cycle, $R$ reaches a suffix $R'$, such that every received message during $R'$  was sent during $R$. By repeating the previous argument, it holds that within $\bigO(1)$ asynchronous cycles, $R$ reaches a suffix $R'$ in which for every received reply message, we have that its associated request message was sent during $R$. Thus, $R'$ is consistent.   
\end{proof} 

\smallskip

The proof of Theorem~\ref{thm:livenessAll} considers both complete and not complete $\sSnapshot()$ operations. We say that a $\sSnapshot()$ operation is \emph{complete} if it starts due to a step $a_i$ in which $p_i$ calls the $\sSnapshot()$ operation (line~\ref{ln:operationSnapshot2}) and its operation index, $s$, is greater than any of $p_i$'s snapshot indices in the system state that appears immediately before $a_i$, where $s$ is the value of $sns_i$ stored in $\pndSnap_i[i].sns$ (line~\ref{ln:snsPLUSonepndSnap}). Otherwise, was say that it is not complete.

\begin{theorem}[\textbf{Algorithm~\ref{alg:terminating}'s termination and linearization}]
	\label{thm:livenessAll}
Let $R$ be a consistent execution (as defined by Definition~\ref{def:safeConfigTerminating}) with respect to some execution of Algorithm~\ref{alg:terminating}. Suppose that there exists $p_i \in \sP$, such that in $R$'s second system state (which immediately follows $R$'s first step that may include a call to the $\sSnapshot()$ operation in line~\ref{ln:operationSnapshot2}) it holds that $\pndSnap_i[i]=(s,\bullet,\bot)$ and $s>0$.
%
%Within $\bigO(\delta)$ asynchronous cycles, 
%
Eventually the system reaches a state $c \in R$ in which $\pndSnap_i[i]=(s,\bullet,x):x\neq\bot$. 	Moreover, Algorithm~\ref{alg:terminating}'s cost measures are in $\bigO(n)$ per $\sWrite()$ operation and in $\bigO(n+\delta)$ per $\sSnapshot()$ operation.
\end{theorem}

\begin{proof}
Lemmas~\ref{thm:livenessI},~\ref{thm:livenessII},~\ref{thm:linearizationII} and~\ref{thm:allBounds} prove the theorem. Lemmas~\ref{thm:livenessI} and~\ref{thm:livenessII} use the function $\sS_i()$, which we define next, as well as Definition~\ref{def:dissemination}.  Whenever $p_i$'s program counter is outside of the function $\baseSnapshot()$, the $\sS_i()$ function returns the value of $\exceed_i$. Otherwise, the function returns the value of $(S_i \cap \Delta_i)$. The proof of lemmas~\ref{thm:livenessI} and~\ref{thm:livenessII} is by convergence stairs~\cite[Chapter 2.9]{DBLP:books/mit/Dolev2000}, which Definition~\ref{def:dissemination} facilitates.

\begin{definition}[Eventual (i) diffusion, (ii) triggering, and (iii) non-interruption]
	\label{def:dissemination}
	Let $R$ be an execution of Algorithm~\ref{alg:terminating} and $p_i \in \sP$ a non-faulty processor that invokes snapshot task, $T(i,s)$, such that $\pndSnap_i[i]=(s,\bullet,\bot)$ and $s>0$ in the starting system state of $R$. The definition refers to a \emph{diffusion set} of processors $D \subseteq \sP$ that includes at least one majority set  $M \subseteq \sP: |M|> |\sP|/2$ and any processor that executes an unbounded (possibly infinite) number of write operations concurrently with $T(i,s)$. 
	
	(i) Suppose that there is a diffusion set, $D$, such that eventually, the system reaches a state $c \in R$ in which $(i, \bullet) \in \sS_i()$ and $\pndSnap_j[i]=(s,\bullet,\bot)$, for any $p_j \in D$. In this case, we say that $R$ has diffused task $T(i,s)$ by $c$. 
	
	(ii) Suppose that there is a diffusion set, $D$, such that eventually, the system reaches a state $c' \in R$ in which for $p_j$, it holds that $(i, \bullet) \in \sS_i()$, $\pndSnap_j[i]=(s,y,\bullet):y\neq\bot$ and $y \preceq VC_j$, for any $p_j \in D$. In this case, we say that $R$ triggers the helping-scheme for task $T(i,s)$ in $c'$. 
	
	(iii) Suppose that there is a diffusion set, $D$, such that eventually, the system reaches a state $c'' \in R$ in which for $p_j$, it holds that $(i, \bullet) \in \sS_i()$ and $\pndSnap_j[i] \in \exceed_j$, for any $p_j \in D$, we say that $R$ does not interrupt task $T(i,s)$ in $c''$.
\end{definition}

\begin{lemma}[\textbf{Algorithm~\ref{alg:terminating}'s termination --- part I}]
\label{thm:livenessI}
Let  $R$ be a consistent execution (Definition~\ref{def:safeConfigTerminating}) with respect to some execution of Algorithm~\ref{alg:terminating}. Suppose that there exists $p_i \in \sP$, such that in $R$'s second system state (which immediately follows $R$'s first step that may include a call to the $\sSnapshot()$ operation in line~\ref{ln:operationSnapshot2}) it holds that $\pndSnap_i[i]=(s,\bullet,\bot)$ and $s>0$.
%
%Within $\bigO(\delta)$ asynchronous cycles, 
%
Eventually the system reaches a state $c \in R$ in which either: (i) $R$ does not interrupts task $T(i,s)$ in $c$ (Definition~\ref{def:dissemination}),
%
%for any non-failing processor $p_j \in \sP$ it holds that $(i, \bullet) \in \sS_j()$ (line~\ref{ln:exceedDelta}) and $\pndSnap_j[i]=(s,\bullet,\bot)$, 
%
(ii) any majority $M \subseteq \sP:|M|>|\sP|/2$ includes at least one $p_j \in M$, such that $\pndSnap_j[i]=(s,\bullet,x):x\neq\bot$ or (iii) $\pndSnap_i[i]=(s,\bullet,x):x\neq\bot$. Moreover, whenever $R$ is a fair execution, the above occurs within $\bigO(\delta)$ asynchronous cycles.
\end{lemma}
\begin{proof}
Towards a proof in the way of contradiction, suppose that the lemma is false. That is, for any finite prefix $R'$ of $R$,
%
%that includes $\bigO(\delta)$ asynchronous cycles, such that 
%
none of the lemma invariants hold.
%
% during $R'$. 
%
The proof uses claims~\ref{thm:livenessInoStep} and~\ref{thm:livenessIcv} for demonstrating a  contradiction with the above assumption in Claim~\ref{thm:livenessIatLeastDelta}.

\begin{claim}
\label{thm:livenessInoStep}
$R'$ does not include a step in which processor $p_i$ evaluates the if-statement condition in line~\ref{ln:prevEqReg} to be true (or at least one of the lemma invariants holds). 
\end{claim}
\begin{claimProof}
Arguments (1), (2) and (3) show that during $a_i \in R'$ processor $p_i$ calls the function $\sStore(\{(k, \pndSnap[k].sns, prev):(k, s, \bullet) \in S\})$. Argument (4) shows that this implies that invariant (ii) holds. Thus, we reached a contradiction with the assumption in the lemma proof.

\smallskip

\noindent Argument (1): \emph{a call to $\mathrm{baseWrite}()$ ends eventually, \ie within a finite time.~~}

%  within $\bigO(1)$ asynchronous cycles

A call to $\mathrm{baseWrite}(v)$ starts with $p_i$ incrementing $ts_i$ and storing it in $reg_i[i]$ (line~\ref{ln:tsIncrementWrite}) to a value that is unique (in the system state that immediately follows) with respect to $ts$'s variables and fields that are associated with $p_i$ (Theorem~\ref{thm:terminatingAlgConvergence}). We note that the repeat-until loop in line~\ref{ln:waitUntilWRITEackReg2} terminates within a finite time due to the theorem assumption that $R$ is consistent as well as the system settings assumptions about communication fairness and that the system always includes a majority of processors that never fail.

\smallskip

\noindent Argument (2): \emph{the repeat-until loop in lines~\ref{ln:letSprime} to~\ref{ln:waitUntilSNAPSHOTackReg2} ends eventually.~~} %within $\bigO(1)$ asynchronous cycles.~~}

The call to $\baseSnapshot(S_i): (i, \bullet) \in S_i$ starts with $p_i$ incrementing $ssn_i$ (line~\ref{ln:prevSsnGetsRegSsnPlusOne2}) to a value that is unique (in the system state that immediately follows) with respect to $ssn$'s variables and fields that are associated with $p_i$ (Theorem~\ref{thm:terminatingAlgConvergence}). We note that the repeat-until loop in line~\ref{ln:waitUntilSNAPSHOTackReg2} terminates within a finite time due to the system settings assumptions about communication fairness and that the system always includes a majority of processors that never fail as well as the theorem assumption that $R$ is consistent and the uniqueness of $ssn_i$'s value (or the fact that $\sS_i()=\emptyset$, which implies the lemma since then invariant (iii) holds).

\smallskip

Argument (3) considers a call that $p_i$ performs to $\baseSnapshot()$ with the parameter $S_i$.

\noindent Argument (3): \emph{showing that eventually, $p_i \in \sP$ executes $\baseSnapshot_i(S_i): (i, \bullet) \in S_i$ where it takes a step $a_i$ that includes the execution of the if-statement in line~\ref{ln:prevEqReg}.~~}

\indent  The assumption that invariant (iii) does not hold in $R'$ implies that $(i, \bullet) \in \sS_i$ whenever processor $p_i$ takes a step that includes the execution of $\baseSnapshot_i(S_i)$ or line~\ref{ln:letSexceedDeltaCall}, which is part of Algorithm~\ref{alg:terminating}'s do-forever loop. The latter occurs within finite time (due to Argument (1) of this claim) and it includes the call to $\baseSnapshot_i()$ (line~\ref{ln:letSexceedDeltaCall}). Thus, the execution of line~\ref{ln:prevEqReg} is implied by the fact that the repeat-until loop in lines~\ref{ln:letSprime} to~\ref{ln:waitUntilSNAPSHOTackReg2a} eventually ends due to Argument (2) of this claim.

\smallskip

\noindent Argument (4): \emph{showing that invariant (ii) holds.~~} 

\indent The function $\sStore_i()$, which $p_i$ calls in line~\ref{ln:prevEqReg}, repeatedly sends to all processors the message $\mathrm{SAVE}$ until $p_i$ receives matching $\mathrm{SAVEack}$ messages from a majority of processors. Theorem~\ref{thm:terminatingAlgConvergence} and the assumption that $R'$ is consistent imply that every received $\mathrm{SAVEack}$ message can be associated with a matching $\mathrm{SAVE}$ message that was indeed sent during $R$. Thus, the rest of the proof shows that the existence of this majority of acknowledgments from processors $p_j \in \sP$ implies that invariant (ii) holds (due to the intersection property of majority groups). According to lines~\ref{ln:safeStoreArrival} to~\ref{ln:safeStoreSendAck}, the arrival of the message $\mathrm{SAVE}$ to $p_j \in \sP$ assures that $\pndSnap_j[i].\mathit{fnl} \neq \bot$ before sending the message $\mathrm{SAVEack}$ back to $p_i$. This is due to Theorem~\ref{thm:terminatingAlgConvergence} and the assumption that $R$ is consistent.  
\end{claimProof}

\begin{claim}
\label{thm:livenessIcv}
%
%Within $\bigO(1)$ asynchronous cycles,
%
The property of eventual helping-scheme triggering holds (Definition~\ref{def:dissemination}) or at least one of the lemma invariants holds. 
\end{claim}
\begin{claimProof}
We first show that the claim holds for a majority of the processors by considering the case of $j=i$ before considering the case of $j \neq i$. Then, we show that the claim holds for any processor that performs an unbounded (possibly infinite) number of write operations concurrently with $p_i$'s snapshot operation.

\smallskip

\textbf{The $\mathbf{j=i}$ case.~~}
%
%Within $\bigO(1)$ asynchronous cycles
%
Eventually, $p_i$ calls $\baseSnapshot(S_i): (i, \bullet) \in S_i$ (line~\ref{ln:letSexceedDeltaCall}) due to Argument (3) in Claim~\ref{thm:livenessInoStep}. 
This, Argument (2) in Claim~\ref{thm:livenessInoStep} and Claim~\ref{thm:livenessInoStep} imply the execution of line~\ref{ln:vcUpdateCond} in every call for $\baseSnapshot(S_i)$. Hence, the claim for the case of $j=i$.  

\smallskip

\textbf{The $\mathbf{j\neq i}$ case.~~}
By the arguments of the case of $j=i$, eventually, processor $p_i$ executes  
lines~\ref{ln:letSprime} and~\ref{ln:ssnPlusOne2} in which $p_i$ broadcasts the record $(i,\pndSnap_i[i].sns$, $\pndSnap_i[i].vc) \in S'$ to all processors in the system via $\mathrm{SNAPSHOT}(S',\bullet)$ messages. This repeats until a majority of nodes acknowledge its reception, which we know to occur within a finite time due to the theorem assumption that $R$ is consistent as well as the system settings assumptions about communication fairness and that the system always includes a majority of processors that never fail. Note that $\pndSnap_i[i].vc\neq \bot$ holds by the above case of $j=i$. Moreover, once processor $p_j$ receives this $\mathrm{SNAPSHOT}$ message, $\pndSnap_j[i].vc\neq \bot$ holds (line~\ref{ln:vcRemoteUpdate}). Furthermore, $y \preceq VC_j$ holds since otherwise, $R$ is not a consistent execution.

\smallskip

\textbf{The case of any processor, $p_k \in \sP$, that performs an unbounded (possibly infinite) number of write operations concurrently with $p_i$'s snapshot operation.~~}
The above arguments for the case of $j\neq i$ can be repeated as long as invariant (iii) does not hold particularly for the case of $p_k$ since the acknowledgments to its $\mathrm{WRITE}$ messages are   piggybacked with $\mathrm{SNAPSHOT}$ message, cf. the comment of the code in line~\ref{ln:concurentPiggybacking}. Thus, the arrival of such a $\mathrm{SNAPSHOT}$ message to all $p_j \in \sP$ occurs eventually (or one of the lemma invariants holds).
\end{claimProof}

\begin{claim}
\label{thm:livenessIatLeastDelta}
Let $c' \in R'$ be the system state that appear in Definition~\ref{def:dissemination}, which Claim~\ref{thm:livenessIcv} showed to exist. Let $x$ be the (finite or infinite) number of iterations of Algorithm~\ref{alg:terminating}'s outer loop in $\baseSnapshot()$ function (lines~\ref{ln:prevSsnGetsRegSsnPlusOne2} and~\ref{ln:outerStop}) that processor $p_i$ takes between $c'$ and $c'' \in R'$, where $c''$ is a system state after which the system eventually reaches the state $c'''$ in which at least one of the lemma invariants holds. The value of $x$ is actually finite and $x\leq\delta$. 
\end{claim}

\begin{claimProof}
Arguments (1) to (3) show that $x\leq\delta$. Moreover, between $c''$ and $c'''$ there are a finite number of steps.

%$\bigO(\delta)$ asynchronous cycles.

\smallskip

\noindent Argument (1): \emph{as long as none of the proof invariants hold, whenever processor $p_i$ iterates over the outer loop in $\baseSnapshot()$ function (lines~\ref{ln:prevSsnGetsRegSsnPlusOne2} and~\ref{ln:outerStop}), $p_i$ takes a step in which it tests the if-statement condition at line~\ref{ln:prevEqReg} and that condition does \emph{not} hold.}

\indent  
%
%Within $\bigO{(1)}$ asynchronous cycle
%
Eventually, $p_i$ takes a step that includes a call to $\baseSnapshot(S_i): (i, \bullet) \in S_i$ (line~\ref{ln:letSexceedDeltaCall}) at least once (Argument (3) in Claim~\ref{thm:livenessInoStep}). By Claim~\ref{thm:livenessInoStep}, that call includes the execution of line~\ref{ln:prevEqReg} in which the if-statement condition does not hold (because then Argument (4) in Claim~\ref{thm:livenessInoStep} implies that invariant (ii) holds). 

\smallskip

%outer loop in $\baseSnapshot()$ function (lines~\ref{ln:prevSsnGetsRegSsnPlusOne2} and~\ref{ln:outerStop})

\noindent Argument (2):  \emph{suppose that there are at least $x$ consecutive and complete iterations of $p_i$'s outer loop in the $\baseSnapshot()$ function (lines~\ref{ln:prevSsnGetsRegSsnPlusOne2} and~\ref{ln:outerStop}) between $c'$ and $c''$ in which the if-statement condition at line~\ref{ln:prevEqReg} does not hold. There are at least $x$ write operations that run concurrently with the snapshot operation that has the index of $s$.} 

\indent  The only way that the if-statement condition in line~\ref{ln:prevEqReg} does not hold in a repeated manner is by repeated changes of $ts$ field values in $reg_i$ during the different executions of lines~\ref{ln:prevSsnGetsRegSsnPlusOne2} to~\ref{ln:waitUntilSNAPSHOTackReg2a}. Such changes can only happen due to increments of $ts_j:p_j \in \sP$ (line~\ref{ln:tsEqtsPlusOne}) at the start of $\sWrite()$ operations.

\smallskip

\noindent Argument (3): \emph{there exists $x'\leq \delta$ for which $(i, \bullet) \in \sS_i()$ (or at least one of the lemma invariants hold), where $x'$ is the number of consecutive and complete iterations of the outer loop in the $\baseSnapshot()$ function (lines~\ref{ln:prevSsnGetsRegSsnPlusOne2} and~\ref{ln:outerStop}) between $c'$ and $c''$ in which the if-statement condition at line~\ref{ln:prevEqReg} does not hold.} 

\indent  Argument (2) implies that the number of iterations continues to grow (as long as none of the lemma invariants holds). The proof of Argument (2) and Claim~\ref{thm:neverRemove} imply that during every such iteration there are increments of at least one of the summation $\sum_{\ell\in\{1,\ldots,n\}}\mathrm{VC}_i[\ell]-\pndSnap_i[i].vc[\ell]$ until that summation is at least $\delta$. Recall that $\pndSnap[i].vc\neq\bot$ (Claim~\ref{thm:livenessIcv}) and $\pndSnap_i[i].\mathit{fnl} = \bot$ (the assumption that none of the lemma invariants hold in $R'$). Thus, $(i, \bullet) \in \sS_i()$ holds (line~\ref{ln:exceedDelta}, for the case of $k=i$).

\smallskip

\noindent Argument (4): \emph{suppose that $p_i$ has taken at least $x'$ iterations of the outer loop in $\baseSnapshot()$ function (lines~\ref{ln:prevSsnGetsRegSsnPlusOne2} and~\ref{ln:outerStop}) after system state $c'$ (which is defined in Claim~\ref{thm:livenessIcv}). After these $x'$ iterations, suppose that the system has reached a state $c''$ in which $(i, \bullet) \in \sS_i()$, as in Argument (3). 
%	
%	Within $\bigO(1)$ asynchronous cycles 
%	
Eventually after $c''$, one of the lemma invariants holds in system state $c'''$.}

To the end of showing that there is a diffusion set, $D$, that does not interrupt $(i, \bullet) \in \sS_i()$ in $c''' \in R$, we first show that a majority set $M \subseteq \sP: |M|> |\sP|/2$ exists before considering any processor  $p_j$ that executes an unbounded (possibly infinite) number of write operations concurrently with task $(i, \bullet) \in \sS_i()$.

Eventually after $c''$ (which Argument (3) defines), it holds that $reg_j$'s $ts$ fields are not smaller than the ones of $reg_i$'s $ts$ fields in $c''$. This is because in every iteration of the outer loop in $\baseSnapshot()$ function (lines~\ref{ln:prevSsnGetsRegSsnPlusOne2} and~\ref{ln:outerStop}), processor $p_i$ broadcasts $reg_i$ to all processors (line~\ref{ln:ssnPlusOne2}). These $\mathrm{SNAPSHOT}$ messages arrive eventually at least a majority, $M$, of non-failing processors $p_j\in M$ and upon their arrival $p_j$ updates $reg_j$ (lines~\ref{ln:SNAPSHOTarrival} to~\ref{ln:SNAPSHOTarrivalEND}). The rest of the proof shows that $(i, \bullet) \in \sS_j()$ holds (line~\ref{ln:exceedDelta} for the case of $k=i$); the reasons for that are similar to the ones that appear in the proof of Argument (3).

For the case of processor $p_j$ that executes an unbounded (possibly infinite) number of write operations concurrently with task $(i, \bullet) \in \sS_i()$, we note that similar to the above arguments can repeat due to the acknowledgments of $\mathrm{WRITE}$ messages are piggybacked with $\mathrm{SNAPSHOT}$ message, cf. the comment of the code in line~\ref{ln:concurentPiggybacking}.  
\end{claimProof}

To the end of completing the proof of the lemma, we note that, other than in Claim~\ref{thm:livenessIatLeastDelta}, all interaction between $p_i$ and the other diffusion set members $p_j \in D$ are based on a constant number of round-trip interactions, which can be performed within $\bigO(1)$ asynchronous cycles. For the case of Claim~\ref{thm:livenessIatLeastDelta}, $\bigO(\delta)$ asynchronous cycles are required.
\end{proof}

\smallskip

We note that invariants (i) and (i) of lemmas~\ref{thm:livenessI} and~\ref{thm:livenessII} match and invariant (iii) of Lemma~\ref{thm:livenessI} implies that the theorem holds.

\begin{lemma}[\textbf{Algorithm~\ref{alg:terminating}'s termination --- part II}]
	\label{thm:livenessII}
	Let  $R$ be a consistent execution of Algorithm~\ref{alg:terminating} (Definition~\ref{def:safeConfigTerminating}) and $p_i \in \sP$. Moreover, suppose that 
	%
	%either (i) 
	%
	in any system state of $R$, it holds that $\pndSnap_i[i]=(s,\bullet,\bot)$, $s>0$.
	% as well as for any non-failing processor $p_j \in \sP$ it holds that $(i,\bullet)\in \sS_j()$ (line~\ref{ln:exceedDelta}) and $\pndSnap_j[i]=(s,\bullet,\bot)$, or (ii) in any system state of $R$, it holds that $\pndSnap_i[i]=(s,\bullet,\bot)$, $s>0$ as well as any majority $M \subseteq \sP:|M|>|\sP|/2$ includes at least one $p_j \in M$, such that $\pndSnap_j[i]=(s,\bullet,x):x\neq\bot$. Within $\bigO(1)$ asynchronous cycles, 
	Eventually, the system reaches a state $c \in R$ in which $\pndSnap_i[i]=(s,\bullet,x):x\neq\bot$. Moreover, whenever $R$ is a fair execution, the above occurs within $\bigO(\delta)$ asynchronous cycles. 
\end{lemma}
\begin{proof}
	The proof is implied by Claims~\ref{thm:livenessIImajority} and~\ref{thm:livenessIIexceedDelta}.
	
	\begin{claim}
		\label{thm:livenessIImajority}
		Suppose that $\pndSnap_i[i].sns> 0$ holds in any system state of $R$ and that for any majority $M \subseteq \sP:|M|>|\sP|/2$ includes at least one $p_j \in M$, such that $\pndSnap_j[i]=(s,\bullet,x):x\neq\bot$. Eventually, the system reaches a state $c \in R$ in which $\pndSnap_i[i]=(s,\bullet,x):x\neq\bot$.
	\end{claim}
	\begin{claimProof}
		Towards a proof in the way of contradiction, suppose that the claim is false. That is, $R$ has no suffix $R'$, such that $\pndSnap_i[i]=(s,\bullet,x):x \neq\bot$ holds in the starting system state of $R'$. Arguments (1) to (2) show the needed contradiction. Recall that by Argument (3) in Claim~\ref{thm:livenessInoStep} it holds that every iteration of the do-forever loop during $R'$ includes a call to $\baseSnapshot(S_i):(i, \bullet) \in S_i$ (line~\ref{ln:letSexceedDeltaCall}).
		
\smallskip

 \noindent Argument (1): \emph{eventually, a majority of nodes acknowledge the message $\mathrm{SNAPSHOT}(\bullet,reg_i, ssn_i)$, such that for at least one $\mathrm{SNAPSHOTack}(\mathnormal{AJ} ,\bullet,\mathnormal{ssnJ})$ acknowledgment, it holds that $\mathnormal{ssnJ} =ssn_i$, $(s, \pndSnap_i[s].sns,\bullet, x) \in \mathnormal{AJ}$ and $\bot \neq x = \pndSnap_j[s].\mathit{fnl}$.~~}

\indent We show that eventually, for at least one $\mathrm{SNAPSHOTack}(\mathnormal{AJ} ,\bullet,\mathnormal{ssnJ})$ message, say the one from $p_j$, it holds that $\mathnormal{ssnJ} =ssn_i$, $(s, \pndSnap_i[s].sns, x) \in \mathnormal{AJ}$ and $\bot \neq x = \pndSnap_j[s].\mathit{fnl}$. This is followed from the fact that line~\ref{ln:ssnPlusOne2} broadcasts repeatedly the $\mathrm{SNAPSHOT}(\bullet, ssn_i)$ message until at least a majority receives it and acknowledges it. By Argument (2) in Claim~\ref{thm:livenessInoStep}, the repeat-until loop in lines~\ref{ln:letSprime} to~\ref{ln:waitUntilSNAPSHOTackReg2} ends eventually. Moreover, by the proof of Argument (2) in Claim~\ref{thm:livenessInoStep}, the received acknowledgments indeed refer to these messages, and at least one of these acknowledgments includes $(s, \pndSnap_i[s].sns,\bullet, x) \in \mathnormal{AJ}:\bot \neq x = \pndSnap_j[s].\mathit{fnl}$ due to the claim assumption about $M$. 
		
\smallskip

\noindent Argument (2): \emph{eventually, $\pndSnap_i[i].\mathit{fnl}\neq \bot$ holds.~~}

\indent By the proof of Argument (1), a majority $M$ of processors successfully acknowledge the message $\mathrm{SNAPSHOT}$ (line~\ref{ln:SNAPSHOTackSend}). By the claim assumption and the intersection property of majority sets, for at least one of the acknowledging nodes, say, $p_k$, it holds that $\pndSnap_k[i]=(s,\bullet,x):x\neq\bot$. Since line~\ref{ln:SNAPSHOTackSend} piggyback the $\mathrm{SNAPSHOTack}(reg, \mathnormal{ssnJ})$ and $\mathsf{SAVE}(A)$ messages, the message $\mathsf{SAVE}(A)$ successfully arrives from $p_k$ to $p_i$. The proof is done by the fact that once the message $\mathsf{SAVE}$ arrives to $p_i$, line~\ref{ln:safeStoreUpon} updates $\pndSnap_i[i]=(s,\bullet,x):x\neq\bot$.
\end{claimProof}

\begin{claim}
	\label{thm:livenessIIexceedDelta}
	Suppose that $R$ does not interrupt task $(i, \bullet) \in \sS_i()$ in $R$'s starting system state (Definition~\ref{def:dissemination}). Eventually, the system reaches a state $c \in R$ in which $\pndSnap_i[i]=(s,\bullet,x):x\neq\bot$.
\end{claim}
\begin{claimProof}
The proof is by a sequence of statements, \ie arguments (1) to (3). 
			
	\smallskip

	\noindent Argument (1): \emph{a call to $\sStore_k()$ during $R$ implies the claim eventually.~~}

	\indent Suppose that the if-statement condition in line~\ref{ln:prevEqReg} holds after the execution of lines~\ref{ln:ssnPlusOne2} to~\ref{ln:waitUntilSNAPSHOTackReg2a}. For the case of $i=k$, the call to $\sStore_i(\{(\bullet,r_i):r_i\neq \bot\})$ causes $p_i$ to send the message  $\mathsf{SAVE}(\{(\bullet,r_i):r_i\neq \bot\})$ to itself and the reception of this message assigns $r_i\neq \bot$ to $\pndSnap_i[i].\mathit{fnl}$ (line~\ref{ln:safeStoreUpon}). For the case of $k\neq i$, the proof of Argument (4) of Claim~\ref{thm:livenessInoStep} implies that the call to $\sStore_k()$ creates a set, $M \subseteq \sP$, of majority processors, $p_j$, that store in $\pndSnap_j[i].\mathit{fnl}$ the result of $(i, \bullet) \in \sS_j()$. Moreover, the proof of Argument (2) of Claim~\ref{thm:livenessIImajority} implies the claim.

	\smallskip
	
	\noindent Argument (2): \emph{eventually, the system reaches a state in which} there are no active write operations {or neither $(i, \bullet) \in \sS_j()$ nor $(i, \bullet) \in \sS_i()$ hold.~~}
	
	\indent  Note that during $R$, any processor that executes the $\sWrite()$ function, returns from this call to $\sWrite()$ eventually (Argument (1) in Claim~\ref{thm:livenessInoStep} and the assumption that $R$ is non-interruptive). Specifically, by Definition~\ref{def:dissemination}, for any processor $p_j$ that executes an unbounded (possibly infinite) number of write operations concurrently with task $(i, \bullet) \in \sS_j()$, it holds that $\pndSnap_j[i] \in \exceed_j$ eventually. Then, $p_j$ calls $\baseSnapshot_j(S_j):S_j=\exceed_j$. Note that by the exit condition of the outer repeat-until loop (line~\ref{ln:outerStop}) implies that $p_j$ does not return from that call to $\baseSnapshot_j(S_j):S_j=\exceed_j$ as long as $(i, \bullet) \in \sS_j()$ holds. Note that if the latter does not hold, then neither $(i, \bullet) \in \sS_i()$ holds (due to claims~\ref{thm:livenessIcv} and~\ref{thm:livenessIatLeastDelta}).

	\smallskip

	\noindent Argument (3): \emph{eventually, $p_i$'s snapshot operation, $(i, \bullet) \in \sS_i()$, terminates.~~}
	
	\indent Towards a proof in the way of contradiction, suppose that the statement is false. That is, let $R'$ be a suffix of $R$, in which $p_i$'s snapshot operation does not terminate. In other words, in every system state of $R'$, it holds that $\pndSnap_i[i]=(s,\bullet,\bot)$.

	By Argument (2) of Claim~\ref{thm:livenessInoStep}, the repeat-until loop in lines~\ref{ln:ssnPlusOne2} to~\ref{ln:waitUntilSNAPSHOTackReg2a}  terminates. By line~\ref{ln:waitUntilSNAPSHOTackReg2}, this happens only when (i) $(S_i \cap \Delta_i) \neq \emptyset$ or (ii) majority of matching $\mathrm{SNAPSHOTack}$ messages arrived. The former case implies that the proof is done, because the only way in which $(i,\bullet)$ leaves the set $(S_i \cap \Delta_i)$ is by having $\pndSnap_i[i]=(s,\bullet,x):x\neq\bot$ (due to the execution of line~\ref{ln:safeStoreUpon}). For the latter case, we note that Argument (2) implies that the if-statement condition in line~\ref{ln:prevEqReg} holds, and thus by Argument (2), which showed the needed contradiction, completes the proof of the claim.	
\end{claimProof}

To the end of completing the proof of the lemma, we note that all interaction between $p_i$ and the other diffusion set members $p_j \in D$ (that are considered directly in the proof of the lemma) are based on a constant number of round-trip interactions, which can be performed within $\bigO(1)$ asynchronous cycles. Since the proof considers also Lemma~\ref{thm:livenessI}, then (indirectly) $\bigO(\delta)$ asynchronous cycles are needed.
\end{proof}

\begin{lemma}[\textbf{Algorithm~\ref{alg:terminating}'s linearization}]
	\label{thm:linearizationII}
	Algorithm~\ref{alg:terminating} respects the sequential specification of the snapshot object.
\end{lemma}
\begin{proof} We note that the $\mathrm{baseWrite}()$ functions in Algorithms~\ref{alg:9disCongif} and~\ref{alg:terminating} are identical. Moreover, Algorithm~\ref{alg:9disCongif}'s lines~\ref{ln:prevRegSsnPlusOne} to~\ref{ln:9waitUntilSNAPSHOTack} are similar to Algorithm~\ref{alg:terminating}'s lines~\ref{ln:prevSsnGetsRegSsnPlusOne2} to~\ref{ln:waitUntilSNAPSHOTackReg2a}, but differ in the following manner: (i) the dissemination of the operation tasks is done outside of Algorithm~\ref{alg:9disCongif}'s lines~\ref{ln:prevRegSsnPlusOne} to~\ref{ln:9waitUntilSNAPSHOTack} but inside of Algorithm~\ref{alg:terminating}'s lines~\ref{ln:prevSsnGetsRegSsnPlusOne2}, and (ii) Algorithm~\ref{alg:9disCongif} considers one snapshot operation at a time whereas Algorithm~\ref{alg:terminating} considers many snapshot operations.

%\EMS{Chryssis, the detailed proof here needs to rewrite~\cite[Lemma~7]{DBLP:journals/tpds/Delporte-Gallet18} for Algorithm~\ref{alg:terminating}.}
%\cgr{If we assume that the proof of Lemma 7 is acceptable, then what you write below I believe it is almost sufficient. For us it is more important to
%	show convergence, and once stale information is cleaned, we should expect the algorithm to behave as in the non-self-stabilizing
%	version. Moving along this line of thought, then I think this boils down in arguing that the main difference of these two algorithms
%	(in the absence of stale information) is essentially the use of safe registers than reliable broadcast. We need to argue that 
%	the use of safe registers in Algorithm 4 provides what reliable broadcast does in Algorithm 2 (without caring performance differences). Makes
%	sense?} 

The proof is based on observing that the definition of linearizability (Section~\ref{sec:snapshotTask}) allows concurrent snapshot operations to have the same result (as long as they each individually respect all the other constraints that appear in the definition of linearizability). Moreover, by the same definition, the  linearizability property does not depend on the way in which the snapshot tasks (and their results) are disseminated.  
(Indeed, the linearizability proof of Delporte-Gallet \etal~\cite[Lemma~7]{DBLP:journals/tpds/Delporte-Gallet18} does not consider the way in which the snapshot tasks, and their results, are disseminated when selecting linearization points. These linearization points are selected according to some partition, defined in~\cite[Lemma~7]{DBLP:journals/tpds/Delporte-Gallet18}. The proof there explicitly allows the same partition to include more than one snapshot result.)
\end{proof}

\begin{lemma} \label{thm:allBounds}
	Algorithm~\ref{alg:terminating}'s costs are $4n+18$ per $\sWrite()$ and $8n+2\delta+34$ per $\sSnapshot()$.
\end{lemma}

\begin{proof}
	The lemma's proof is given in claims~\ref{thm:latencyBoundWrite} and~\ref{thm:latencyBoundSnapshot}. They consider a legal execution of Algorithm~\ref{alg:terminating} and use claim~\ref{thm:specificInBaseSnapshot} for bounding Algorithm~\ref{alg:terminating}'s cost measures. Their proofs consider the term message round trip exchange (RTE), which refers to the period that starts in a system state that appears immediately before a step in which a request message is sent from processor $p_i$ to $p_j$ and ends immediately after the step in which $p_i$ receives $p_j$'s response to its request. We note that each RTE has the cost measure of two (Section~\ref{sec:timeComplexity}).
    
    \begin{claim} \label{thm:specificInBaseSnapshot}
        Let $c_{\mathit{calls}} \in R$ be the system state that appears immediately before the step where node $p_i$ calls $baseSnapshot_i(S_i)$. Within $(2n+8)$ RTEs, the system reaches the state $c_{\mathit{returns}} \in R$ that appears immediately after the step where $p_i$ returns from this call to $baseSnapshot_i(S_i)$.
    \end{claim}

    \begin{claimProof}
        The proof begins with a presentation of the longest possible happened-before relation between the system states $c_{\mathit{calls}}$ and $c_{\mathit{returns}}$. Then, the proof sums up the RTEs found in this happened-before relation.

        \begin{enumerate}
            \item \label{item:allDelta} \textit{Node $p_i$ observes the tasks in $\sdelta$ to be concurrent with at least $\delta$ writes.~~} Suppose that $p_i$ has observed all tasks in $\sdelta$ to be concurrent with at least $\delta$ writes. To the end of not losing generality, at the end of the proof (item~\ref{item:double}), we consider the complementary case. Thus, $p_i$ does not return until all tasks in $\sdelta$ have a result due to the exit condition of the outer repeat-until loop (line~\ref{ln:outerStop}). 
            
            \item \label{item:freshtasks} \textit{Only fresh tasks are stored in $\pndSnap[]$.~~} Until item~\ref{item:safeRegEVENMORE} in the proof, we assume that no node stores in $\pndSnap[]$ a result for any task in $\sdelta$. In item~\ref{item:safeRegEVENMORE}, we explain why the proof does not lose generality due to this assumption.
            
            \item \label{item:majority} \textit{A majority observe the tasks in $\sdelta$ to be concurrent with at least $\delta$ writes.~~} Since node $p_i$ called $\baseSnapshot_i(S_i)$, it accesses the quorum system via broadcasting of $\mathrm{SNAPSHOT}(\sdelta, reg_i, \bullet)$ (line~\ref{ln:ssnPlusOne2}). Within one RTE from $c_{\mathit{calls}}$, node $p_i$'s first snapshot quorum access is completed and the system reaches a state, $c_{\sdelta \text{ majority}} \in R$, in which a majority of the nodes store the tasks in $\sdelta$ in their local $\Delta$-set. The reasons for this are as follows. Recall that we assume that $p_i$ has observed that all tasks in $S_i$ have been concurrent with at least $\delta$ writes (item~\ref{item:allDelta}). Also note that incoming $reg$ values are always merged with the local $reg$ variable using the max function (line~\ref{ln:SNAPSHOTarrivalEND}). Therefore, all nodes in this majority merge the incoming $reg$ values with their local ones and thus they must include the tasks in $\sdelta$ in their local $\Delta$-sets.

            \item \label{item:whenpossible} \textit{Node $p_i$ performs as many snapshot quorum accesses as possible.~~} After $c_{\mathit{calls}}$, $p_i$ performs snapshot quorum accesses (line~\ref{ln:ssnPlusOne2}). If $reg_i = prev_i$ (line~\ref{ln:prevEqReg}), $baseSnapshot_i(S_i)$ returns. Otherwise, a concurrent write operation was observed by $p_i$. Since the proof aims to find the longest happened-before relation, let us assume, whenever possible, that $reg_i \neq prev_i$ after every quorum access completion. Similarly, we assume that, whenever possible, no task in $\sdelta$ terminates, because that also allows $\baseSnapshot_i(S_i)$ to return.

            \item \label{item:oneMoreCalled} \textit{All nodes observe the tasks in $\sdelta$ to be concurrent with at least $\delta$ writes.~~} In order for $reg \neq prev$ to hold after the completion of every quorum access, at least one write has to be performed by another node concurrently with every quorum access of $p_i$ (line~\ref{ln:ssnPlusOne2}). Starting from $c_{\sdelta \text{ majority}}$, this repeated interruption by a concurrent write can happen at most $n-1$ times, because each time a node $p_j$ writes, at least one node in the majority that $p_j$ receives replies for its $\mathrm{WRITEack}$ message has that reply piggybacked with a $\mathrm{SNAPSHOT}(\sdelta, \bullet)$ message, cf. the comment of the code in line~\ref{ln:concurentPiggybacking}. 
            
            Thus, when $p_j$ returns from $\mathrm{baseWrite}_j()$, node $p_j$ immediately calls $\baseSnapshot_j(S_j) : (\sdelta) \subseteq S_j$. The reasons for this are as follows. Recall that in item~\ref{item:majority}, we showed that for a majority of the nodes, the property that the node has observed the tasks in $\sdelta$ to be concurrent with at least $\delta$ writes, holds. Due to that property, and because that property propagates due to the piggybacked $\mathrm{SNAPSHOT}(\sdelta, \bullet)$ message (by a similar reason as in item~\ref{item:majority}), $(\sdelta) \subseteq S_j$ holds. Moreover, the call to $\baseSnapshot_j(S_j)$ is indeed immediate, because the call appears in line~\ref{ln:letSexceedDeltaCall} which is listed right after line~\ref{ln:writePendingExceedDelta} where $\mathrm{baseWrite}()$ returns. Furthermore, recall the assumption from item~\ref{item:whenpossible} that, whenever possible, no task in $\sdelta$ terminates.

            \item \label{item:allCalled} \textit{All non-failing and writing nodes call $\baseSnapshot()$.~~} Let $c_\mathit{all~called}\in R$ denote the system state where all non-failing and writing nodes $p_j$ have called $\baseSnapshot_j(S_j) : (\sdelta) \subseteq S_j$. Recall from item~\ref{item:oneMoreCalled} that during each snapshot quorum access (line~\ref{ln:ssnPlusOne2}) by $p_i$, one other node might perform a write and then immediately calls $\baseSnapshot_j(S_j) : (\sdelta) \subseteq S_j$. Since each such snapshot quorum access takes one RTE and since by item~\ref{item:oneMoreCalled} there are at most $(n-1)$ such other nodes, the system reaches $c_\mathit{all~called}$ within $(n-1)$ RTEs from $c_{\sdelta \text{ majority}}$. It is significant that system state $c_\mathit{all~called}$ is reached, and not only $c_{\sdelta \text{ majority}}$, since the argument in item~\ref{item:safeReg} requires that all nodes and not only a majority of the nodes have called $\baseSnapshot_j(S_j) : (\sdelta) \subseteq S_j$.
            
            \item \label{item:safeReg} \textit{The first call to $\sStore()$.~~} Within one RTE from $c_\mathit{all~called}$, at least one node $p_k$ notices that $reg_k = prev_k$ (line~\ref{ln:prevEqReg}) after its snapshot quorum access (line~\ref{ln:ssnPlusOne2}) finishes. This is true because no node is performing a write quorum access at this point (they are only done in line~\ref{ln:broadcastWRITEreg2}), and thus no write can interrupt the snapshot. Since $reg_k = prev_k$, node $p_k$ calls $\sStore_k(S_k) : (\sdelta) \subseteq S_k$ (line~\ref{ln:call2sStore}) and within one more RTE, a majority of the nodes store the results of all tasks in $\sdelta$. Let $c_\mathit{majority~result} \in R$ denote this system state.
            
            We note that the arguments stated above can also consider the cases in which more than one node calls $\sStore()$ concurrently or that the calling node crashes during the execution of $\sStore()$. This is due to our assumption that the system includes a majority of non-faulty nodes and the fact that any non-$\bot$ task result stored in $\pndSnap[].\mathit{fnl}$ is legitimate.

            \item \label{item:piggyback} \textit{Node $p_i$ returns from $\baseSnapshot()$.~~} Within one RTE from $c_\mathit{majority~result}$, $p_i$ receives a $\mathrm{SAVE}$ message that is piggybacked with a $\mathrm{SNAPSHOTack}$ (cf. line~\ref{ln:sendSAVEwithA}) from a majority of the nodes when $p_i$ performs a snapshot quorum access (line~\ref{ln:waitUntilSNAPSHOTackReg2}). Thus, for each task $T \in \sdelta$, either $p_i$ gets the result for $T$ (the $\mathit{fnl}$ field stores a non-$\bot$ value) or finds that the $sns$ field is increased (meaning the node which created $T$ has finished it and then begun a new snapshot), see line~\ref{ln:prepareREsponse}. In both cases, $\sdelta$ becomes empty and $p_i$ returns from $baseSnapshot_i(S_i)$ immediately before system state $c_{\mathit{returns}}$.

\item \label{item:safeRegEVENMORE} \textit{Removing the assumption made in item \ref{item:freshtasks} about only fresh tasks are stored in $\pndSnap[]$.~~} We now turn to consider the case in which there is a non-failing node, $p_k\in \sP$, that stores in $\pndSnap_k[j]$ a result for any task in $\sdelta$ (without ever having the result stored in $\pndSnap[j]$ by the task initiator or a majority of processors). During legal executions, all transient faults are absent (Section~\ref{sec:sys}), and thus the only way that this can happen is when the processor that executes $\sStore()$ crashes. By the assumption that there is a majority of non-faulty processors and the last paragraph of item~\ref{item:safeReg}, it holds that there is a non-faulty processor that will complete the task correctly and within the bounds stated by the lemma. 

            \item \label{item:double} \textit{The complementary cases to item~\ref{item:allDelta}.~~} Recall that in item~\ref{item:allDelta}, we assume that all tasks in $S_i$ have been observed to run concurrently with at least $\delta$ writes. We now turn to consider the two complementary cases to that assumption, which are: (i) $p_i$ does not observe any task in $S_i$ to run concurrently with at least $\delta$ writes and (ii) $p_i$ observes some, but not all, tasks in $S_i$ to run concurrently with at least $\delta$ writes. 
            
            In case (i), the only task that can be in $S_i$ without having been observed to run concurrently with at least $\delta$ writes is $p_i$'s own task $T_i$ (line~\ref{ln:exceedDelta}). When $\baseSnapshot_i(S_i)$ is called, $p_i$ performs one quorum access. If $reg_i \neq prev_i$, node $p_i$, after this single quorum access, might observe that $T_i$ has been running concurrently with at least $\delta$ writes. In case it does, we are back to the case in item~\ref{item:allDelta} and the total time is one RTE plus all RTEs from item~\ref{item:allDelta} to item~\ref{item:safeRegEVENMORE}. In case it does not, $p_i$ returns from $\baseSnapshot_i(S_i)$ and the proof is complete, since only one RTE was spent. If instead $reg_i = prev_i$, node $p_i$ calls $\sStore_i()$ and then returns from $\baseSnapshot_i(S_i)$. In total, two RTEs were spent, and the proof is completed.

            In case (ii), $S_i$ contains $p_i$'s own task $T_i$ as well as some non-empty set of tasks, $E$, that $p_i$ has observed to be concurrent with at least $\delta$ writes, \ie $S_i = \{ T_i \} \cup E$. Ignoring $T_i$, the same execution as in item~\ref{item:allDelta} to item~\ref{item:safeRegEVENMORE} happens, and the tasks in $E$ terminate within the sum of the RTEs in those items. During this execution, $p_i$ also tries to get a result for $T_i$, but $p_i$ is the only node which tries to do that. When the tasks in $E$ get a result, $T_i$ also gets a result if it is $p_i$ that reads $reg_i = prev_i$. If another node gets a result for the tasks in $E$, task $T_i$ might not get a result.

            During the execution described in the paragraph above, it may happen at any time that some or several nodes observe that $T_i$ suddenly \textit{has} been concurrent with at least $\delta$ writes. For all nodes $p_j:i \neq j$, this does not necessarily cause $T_i$ to be added to $S_j \cap \Delta_j$, since we have to consider the case in which when $p_j$ calls $\baseSnapshot_j(S_j)$, the predicate $T_i \in S_j$ did not hold. For $p_i$, if this happens just after $p_i$ received results for all other tasks in $\sdelta$, \ie $E$, it means that $p_i$ essentially has to start all over again in $\baseSnapshot_i(S_i)$, since $S_i$ contains a task, $T_i$, that only $p_i$ has, but no other node, $p_j$, has called $\baseSnapshot_j(S_j) : T_i \in S_j$. 
            %
            %\ems{This is due to the fact that if $p_i$ does not observe that $T_i$ to be concurrent with at least $\delta$ writes, $p_i$ just returns from $\baseSnapshot_i()$, since $\sdelta$ becomes empty.}
            %
            Therefore, all RTEs spent in item~\ref{item:allDelta} to item~\ref{item:safeRegEVENMORE} have to be spent again, and thus the total time for this case is twice the RTEs from those items. We note that this ``start over'' behavior cannot happen more than once with respect to $T_i$, because no tasks are added to $\sdelta$ during the call to $\baseSnapshot_i(S_i)$, and $T_i$ is the only $p_i$'s task that can go from not having $\delta$ observed concurrent writes, to having $\delta$ observed concurrent writes.
        \end{enumerate}

        Item~\ref{item:majority} is one RTE, item~\ref{item:allCalled} is $n-1$ RTEs, item~\ref{item:safeReg} is two RTEs and item~\ref{item:piggyback} is one RTE. Item~\ref{item:double} tells us to multiply the sum by two. Therefore, the longest happened-before relation between the system states $c_{\mathit{calls}}$ and $c_{\mathit{returns}}$ has $2n+8$ RTEs.
    \end{claimProof}

	\begin{claim} \label{thm:latencyBoundWrite}  
		Algorithm~\ref{alg:terminating}'s cost measure is $4n+18$ per $\sWrite()$ operation.
	\end{claim}

	\begin{claimProof}
        Let $p_i$ be the node that invokes the client $\sWrite()$ operation and $a_i \in R$ the step that includes this invocation. The proof considers the cases in which (i) the case in which $a_i$ does not include the execution of any line in $\baseSnapshot()$ as well as (ii) the case in which it does.
        In case (i), $\mathrm{baseWrite}()$, and hence $\sWrite()$, return within one RTE due to similar arguments to the ones that appear in the end of the proof of Theorem~\ref{thm:closure1} about the number of rounds it takes to perform a $\sWrite()$ operation. In case (ii), $\sWrite()$ returns within $2n+9$ RTEs due to Claim~\ref{thm:specificInBaseSnapshot} and case (i).
	\end{claimProof}

    \begin{claim} \label{thm:latencyBoundSnapshot}
        Algorithm~\ref{alg:terminating}'s cost measure is $8n+2\delta+34$ per $\sSnapshot()$ operation.
    \end{claim}
    
    \begin{claimProof}
        Let $p_i$ be the node that invokes the client $\sSnapshot_i()$ operation, which creates task $T$. Let $c_{\mathit{calls}} \in R$ be the system state that appears immediately before the step, $a_i \in R$, that includes this innovation of $\sSnapshot_i()$ and let $c_{\mathit{returns}} \in R$ be the system state that appears immediately after the step in which $p_i$ returns from that $\sSnapshot_i()$ invocation.
        The proof begins with a presentation of the longest possible happened-before relation between the system states $c_{\mathit{calls}}$ and $c_{\mathit{returns}}$. Then, the proof  sums up the RTEs found in this happened-before relation.
        
        \begin{enumerate}
            \item \label{item:baseSnapshot1} \textit{Node $p_i$ returns from $\baseSnapshot()$.~~} Suppose that step $a_i$ includes the execution of some line in $\baseSnapshot_i(S_i) : T \notin S_i$. Note that the case of $T \in S_i$ reduces the proof to only item~\ref{item:baseSnapshot2}. Moreover, if $a_i$ does not include the execution of $\baseSnapshot_i(\dot)$, the proof reduces to only items~\ref{item:doForever} and~\ref{item:baseSnapshot2}.
            %Since the  proof only considers the longest happened-before relation, these cases are ignored. 
            Due to Claim \ref{thm:specificInBaseSnapshot}, $p_i$ returns from $\baseSnapshot_i(S_i)$ in $2n+8$ RTEs, and the system reaches state $c_\mathit{do~forever} \in R$.

            \item \label{item:doForever} \textit{Node $p_i$ observes $T$ to be concurrent with $\delta$ writes.~~} Starting from $c_\mathit{do~forever}$, node $p_i$ performs a number of iterations of its do-forever loop. 
            We follow similar assumptions as in the proof of Claim~\ref{thm:specificInBaseSnapshot}'s items $1$ to $3$. Namely, we assume that until $p_i$ has observed $T$ to be concurrent with at least $\delta$ writes, node $p_i$ calls $\baseSnapshot_i(S_i):S_i=\{T\}$ and that $reg_i \neq prev_i$ when executing lines~\ref{ln:call2sStore} and~\ref{ln:vcUpdate}, \ie whenever $p_i$ decides whether it can calculate the result of task $T$, and respectively, return from the call to $\baseSnapshot_i(S_i):S_i=\{T\}$. 
            We note that from state $c_\mathit{do~forever}$, there can be $\delta$ RTEs before $p_i$ has observed $T$ to be concurrent with at least $\delta$ writes. The reason is the following. Each RTE $p_i$ spends in the do-forever loop corresponds to one write that $T$ is concurrent with. One write-RTE-correspondence happens in $\mathrm{baseWrite}()$, where one write occurs and $p_i$ spends one RTE. One write-RTE-correspondence happens in $\baseSnapshot()$, since we assumed that $reg_i \neq prev_i$, meaning at least one write occurs during $p_i$'s call to $\baseSnapshot()$. 

            \item \label{item:baseSnapshot2} \textit{Node $p_i$ returns from $\baseSnapshot(S):T\in S$.~~} If state $c_\mathit{T \in \Delta}$ (item~\ref{item:doForever}) is reached after a call to $\baseSnapshot_i()$ returns, $p_i$ may call $\mathrm{baseWrite}_i()$, which takes one RTE. Then, $p_i$ has to call $\baseSnapshot_i(S_i): T \in S$ and returns within $2n+8$ due to Claim~\ref{thm:specificInBaseSnapshot} and by then $p_i$ has a result for $T$, meaning $snapshot_i()$ returns and state $c_{\mathit{returns}}$ is reached.
        \end{enumerate}

        Items~\ref{item:baseSnapshot1} to~\ref{item:baseSnapshot2} consider $2n+8$ RTEs, $\delta$ RTEs and $2n+9$ RTEs, respectively. Thus, the longest happened-before relation between the system states $c_{\mathit{calls}}$ and $c_{\mathit{returns}}$ has $4n+\delta+17$ RTEs.
    \end{claimProof}

Lemma~\ref{thm:allBounds}'s proof end.
\end{proof}

Theorem~\ref{thm:livenessAll}'s proof end.
\end{proof}

\section{Bounded Variations on Algorithms~\ref{alg:disCongif} and~\ref{alg:terminating}}
\label{sec:bounded}
In this section, we discuss how we can obtain bounded variations of our two unbounded self-stabilization algorithms.
Dolev \etal~\cite[Section 10]{DBLP:journals/corr/abs-1806-03498} present a solution to a similar transformation: They show how to take a self-stabilizing atomic MWMR register algorithm for message passing systems  that uses unbounded operation indices and transform it to an algorithm that uses bounded indices. We review the techniques that Dolev \etal use and explain how a similar transformation can also be used for Algorithms~\ref{alg:disCongif} and~\ref{alg:terminating} with respect to their different operation indices.   

The procedure by Dolev \etal~\cite[Section 10]{DBLP:journals/corr/abs-1806-03498} considers operation indices, which they call tags, whereas the proposed algorithms refer to operation indices as (i) $ts$ values in the variables and message fields, as well as (ii) $ssn$ and $sns$ values in the variables and message fields. That is, Dolev \etal consider just one type of operation index whereas we consider several types. For sack of simple presentation, when describing next the procedure by Dolev \etal, we refer to the case of many types of operation indices.

\begin{enumerate}
\item\label{ln:stop} Once node $p_i \in \sP$ stores an operation index that is at least $\mathrm{MAXINT}$, node $p_i$ disables the invocation of all operations (of all types) while allowing the completion of the existing ones (until all nodes agree on the highest index for each type of operation, cf. item~\ref{ln:restart}), where $\mathrm{MAXINT} \in \mathbb{Z}^+$ is a very large constant, say, $\mathrm{MAXINT} = 2^{64}-1$.
				
\item\label{ln:restart} While the invocation of new operations (of all types) is disabled (by item~\ref{ln:stop}), the gossip procedure keeps on propagating the maximal operation indices (and merge the arriving information with the local one). Eventually, all nodes share the same operation indices (for all types). At that point in time, the procedure for dealing with integer overflow events uses a consensus-based global reset procedure for replacing, per operation type,  the highest operation index with its initial value $0$, while keeping the values of all shared registers unchanged. 
\end{enumerate}

\noindent \textbf{Self-stabilizing global reset procedure.~~}
The implementation of the self-stabilizing procedure for global reset can be based on existing mechanisms, such as the one by Awerbuch \etal~\cite{DBLP:conf/wdag/AwerbuchPVD94}. We note that the system settings of Awerbuch \etal~\cite{DBLP:conf/wdag/AwerbuchPVD94} assume execution fairness. 
%
%This assumption is allowed by our system settings (Section~\ref{sec:seldomFair}). 
%
This is because we assume that reaching $\mathrm{MAXINT}$ can only occur due to a transient fault. Thus, execution fairness, which implies all nodes are eventually alive, is seldom required (only for recovering from transient faults).

\section{Evaluation of Self-Stabilizing Snapshots} 
\label{sec:eval}
%
%To the end of validating the proposed solutions and comparing their performances to the ones proposed in~\cite {DBLP:journals/tpds/Delporte-Gallet18}, we present an evaluation of algorithms~\ref{alg:0disCongif} to~\ref{alg:terminating}. We have implemented the four algorithms in a thread-based approach since some of the pseudocode in~\cite {DBLP:journals/tpds/Delporte-Gallet18} considers non-event based handling of the communications. The results of our experiments show that the overhead of Algorithm~\ref{alg:disCongif} compared to Algorithm~\ref{alg:0disCongif} is very small, that Algorithm~\ref{alg:terminating} is at least as fast as Algorithm~\ref{alg:9disCongif} and that snapshot operations are completed in constant time in Algorithm~\ref{alg:terminating} as opposed to in linear time, like Algorithm~\ref{alg:9disCongif}, when the number of snapshotters increases.
%
%\subsection{Focal research questions}
%
%In addition to the validation of the correctness proof of algorithms~\ref{alg:0disCongif} to~\ref{alg:terminating}, our evaluation focuses on the following research questions: (1) How well does the system scale with respect to the number of writers and snapshotters? (2) What is the overhead needed to pay when comparing algorithms~\ref{alg:0disCongif} and~\ref{alg:disCongif} as well as algorithms~\ref{alg:9disCongif} and~\ref{alg:terminating}? (3) What role does the trade-off parameter, $\delta$, play in the performance of Algorithm~\ref{alg:terminating}?
%
To the end of validating the proposed solutions and comparing their performances to the ones proposed in~\cite {DBLP:journals/tpds/Delporte-Gallet18}, we present an experimental evaluation of algorithms~\ref{alg:0disCongif} to~\ref{alg:terminating} that focuses on the following research questions. 
\begin{itemize}
	\item How well does the system scale with respect to the number of writers and snapshotters? 
	\item What is the overhead needed to pay when comparing algorithms~\ref{alg:0disCongif} and~\ref{alg:disCongif} as well as algorithms~\ref{alg:9disCongif} and~\ref{alg:terminating}? 
	\item What role does the trade-off parameter, $\delta$, play in the performance of Algorithm~\ref{alg:terminating}? 
\end{itemize}
Our evaluation criteria mainly consider the operation latency, which is the time between the client invocation and operation return. We also consider the number of communication rounds, \ie the communication costs. We measure the communication rounds as the number of quorum accesses (performed in \eg line \ref{ln:0waitUntilWRITEackReg}). While the throughput might also be of interest, we note that, in our case, it can be found by simply inverting the latency, since in our experiments, the nodes perform operations immediately one after the other. We have implemented the four algorithms in a thread-based approach since some of the pseudocode in~\cite {DBLP:journals/tpds/Delporte-Gallet18} considers non-event based handling of the communications. During our experiments we have observed that the overhead of Algorithm~\ref{alg:disCongif} compared to Algorithm~\ref{alg:0disCongif} is negligible, as well as that Algorithm~\ref{alg:terminating}'s snapshot operations are completed within a period that appears to be constant with respect to the number of pending snapshot operations, whereas Algorithm~\ref{alg:9disCongif}'s does depend on the number of pending snapshots. More detailed conclusions appear in Section~\ref{sec:discussion}.

\subsection{Experiment description}
Our list of experiments includes the following.
\begin{enumerate}
    \item \textbf{Scalability of write operations with respect to write operations.~~}
    %This experiment allows us to understand how write operations change, with respect to the the evaluation criteria, when the number of writers increases.
    This information about write operation scalability is a basic property of Attiya \etal~\cite{DBLP:journals/jacm/AttiyaBD95}, which we need as a control experiment for the evaluation of our implementations. The experiments are conducted by considering 0 snapshotters and a varying number of writers, \ie 1 to 7 writers.
    
    \item \textbf{Scalability of write operations with respect to snapshot operations.~~}
    %This experiment allows us to understand how write operations change, with respect to the evaluation criteria, when the number of snapshotters increases. 
    Since writes and snapshots compete for resources when they run concurrently, it is important to observe how everything works together. The experiments are conducted by considering 7 writers and a varying number of snapshotters, \ie 0 to 7 snapshotters.
    
    \item \textbf{Scalability of snapshot operations with respect to snapshot operations.~~}
    Since some of the studied algorithms can deal with more than one snapshot at a time, we use this experiment for understanding how snapshot operations change, with respect to the evaluation criteria, when the number of concurrent snapshots increases. The experiments are conducted by considering 0 writers and a varying number of snapshotters, \ie 1 to 7 snapshotters.
    
    \item \textbf{Scalability of snapshot operations with respect to write operations.~~}
    %This experiment allows us to understand how snapshot operations change, with respect to the evaluation criteria, when the number of writers increases. This information is important since snapshot operations need to see the same registers twice in a row before they can terminate (see lines X,Y,Z,Q), but writers change the registers, potentially starving them indefinitely. Algorithms~\ref{alg:9disCongif} and~\ref{alg:terminating} guarantee termination of snapshot operations while algorithms~\ref{alg:0disCongif} and~\ref{alg:disCongif} do not. 
    The aim here is to observe how concurrent snapshot operations are affected by various number of concurrent writes. The experiments are conducted by considering 7 snapshotters and a varying number of writers, \ie 0 to 7 writers.
\end{enumerate}

\subsection{System setup}
\label{sec:system_setup}
%
%\smallskip \noindent
%\textbf{Evaluation environment.~~} 
%
We run the experiments on PlanetLab~\cite{DBLP:journals/ccr/ChunCRBPWB03}. PlanetLab includes a set of servers distributed across the Internet, in various places of the world. PlanetLab fits well our objective to evaluate distributed systems since it provides more real-world scenarios than \eg computer simulated ones, because its experiments include real-data regarding packet loss, varying delays and varying node performances. Next, we describe our implementation, the experiment setup and the algorithmic test cases.

Remark~\ref{thm:planetLab} clarifies a natural difference between the system settings described in Section~\ref{sec:sys} and the ones that hold for PlanetLab~\cite{DBLP:journals/ccr/ChunCRBPWB03}. 

\begin{remark}
\label{thm:planetLab}
Our analytical system model (Section~\ref{sec:sys}) assumes that local processing time is always zero since it aims at allowing to demonstrate the algorithms' correctness, \ie the ability to recover from transient-faults, rather than accurately predicting the system performances (for a given implementation). However, in the PlanetLab testbed, the local processing time depends on the system load and can be observed in the experiment results.
\end{remark}

%Now we describe the evaluation environment, the implementation, the test setup and finally the algorithmic test cases.
%The nodes used for the evaluation need to have some basic capabilities, such as: (1) Responding to ping and ssh. (2) Executing simple commands, such as \texttt{ls /}, in a reasonable time (10 seconds). (3) Access to the Internet. (4) Possibility to open ports.

\smallskip \noindent
\textbf{Implementation.~~} The program was written in the Rust Programming Language~\cite{Matsakis:2014:RL:2692956.2663188} and is based on UDP, threads and condition variables. Each node uses a single UDP socket for communication. Furthermore, each node that runs the program uses the following one, two or three threads:

\begin{itemize}
    \item A thread that continuously listens for incoming packets on the UDP socket, by doing a blocking receive. Once a packet arrives (meaning an application level message arrives), the corresponding ``upon message arrival'' logic (\eg line \ref{ln:0arrivalWRITE}) is executed by this thread as well.
    \item A thread that continuously runs the do-forever loop if the current algorithm has a do-forever loop (\ie algorithms \ref{alg:9disCongif}, \ref{alg:disCongif} and \ref{alg:terminating}).
    \item A thread that performs client $write()$ or $snapshot()$ operations, if the node in question is a writer or a snapshotter, respectively.
\end{itemize}

Each wait-until statement (lines~\ref{ln:waitUntilReadSnap},~\ref{ln:preWrite},~\ref{ln:9repSnapIsns},~\ref{ln:waitUntilwritePendingBotsWritets} and~\ref{ln:terminatingWaitUntilSnapshot}) is implemented using a condition variable, which is notified in other parts of the program when the program notices the condition is met.
To the end of avoiding CPU overload, the do-forever loops are not running all the time. For Algorithm~\ref{alg:9disCongif}, the do-forever loop is only taking a single iteration whenever there is write that needs to be handled (line~\ref{ln:backGroundWrite}, $writePending$ becomes different from $\bot$) or there is a snapshot that needs to be handled (line~\ref{ln:unboundedBuffer}, a $SNAP()$ message is received and not yet processed). These conditions are handled using condition variables. For Algorithm~\ref{alg:disCongif}, the do-forever loop is run with a 1 second delay between every iteration, since the loop is only needed for recovery from arbitrary transient faults. We see that there is a trade-off one has to make between saving CPU time and having a fast recovery from arbitrary transient faults. For Algorithm~\ref{alg:terminating}, the do-forever loop is run with a $0.1~ms$ sleep. This ensures that the CPU is not overloaded while still adding a negligible delay compared to the RTT of the selected PlanetLab nodes, which on average is approximately $25~ms$. Furthermore, the do-forever loop only sends gossip messages once every second, since they are only needed for recovering after the occurrence of arbitrary transient faults. We note that Algorithm \ref{alg:terminating}'s do-forever loop could also be implemented using a condition variable just like Algorithm \ref{alg:9disCongif}'s was.

Quorum access is done as follows: (1) The sender constructs the message to be broadcast and sends it to all nodes. (2) The sender waits for the arrival of replies for a quorum of nodes (which is a simple majority in our case) before notifying the condition variable. If more than $100~ms$ passes, the sender re-broadcasts the message. (3) Upon receiving an acknowledgment message, the sender adds the messages's sender identity to a set according to the message's type (\ie one set for $\mathrm{WRITEack}$ messages, one set for $\mathrm{SNAPSHOTack}$ messages and one set for $\mathrm{SAVEack}$ messages). If a set contains a majority of all identities, the corresponding condition variable is notified.

The reliable broadcast used by Algorithm~\ref{alg:9disCongif} is implemented in the same way as the quorum access, but instead of waiting for a majority, all processors must acknowledge. Since no experiment contains crashed nodes, this implementation of a reliable broadcast works. However, we note that for real-world use, a fault-tolerant reliable broadcast is needed, for example uniform reliable broadcast~\cite{DBLP:books/sp/Raynal18}.

\smallskip \noindent
\textbf{Experiment setup.~~} Each physical PlanetLab node runs the program described above and they get a unique identifier. Writers are allocated from the highest node identifier and snapshotters from the lowest node identifier. Fifteen nodes are used for all experiments, and in any experiment, there are at most seven writers and snapshotters, respectively, meaning no node is both a writer and a snapshotter since it allows a clearer system behavior. The fifteen PlanetLab nodes were selected from the nodes the tool plcli~\cite{plcli} reported as healthy.

Each combination of experiment, number of (write or snapshot) clients, algorithm and $\delta$ is run for 60 seconds. The client operation latency is measured as $\frac{1}{|W|} \sum_{p_i \in W} \frac{60}{nwo_i}$ and $\frac{1}{|S|} \sum_{p_i \in S} \frac{60}{nso_i}$ respectively. The number of quorum accesses per write operation is measured as $\frac{1}{|W|} \sum_{p_i \in W} \frac{nwq_i}{nwo_i}$ and the number of quorum accesses per snapshot operation is measured as ${\sum_{p_i \in \sP} (nsq_i + nsq'_i + nrb_i)}/{\sum_{p_i \in S} nso_i}$, where $W \subseteq \sP$ denotes the set of writers, $S \subseteq \sP$ denotes the set of snapshotters, $nwo$ the number of write operations, $nso$ the number of snapshot operations, $nwq$ the number of write quorum accesses, $nsq$ the number of snapshot quorum accesses, $nsq'$ the number of quorum accesses for the  safe register and $nrb$ the number of reliable broadcasts. We note that since nodes that are not snapshotters can still perform snapshot quorum accesses (cf. Algorithm \ref{alg:9disCongif} and Algorithm \ref{alg:terminating}, for the helping schemes), the number of quorum accesses of snapshot operations is measured system-wide, as opposed to per-node, as for write operations. Each combination is run 10 times, with a random node identifier being assigned to each physical node each time. To the end of mitigating the effect of outliers, the highest and the lowest result values are discarded and the average of the rest is presented.

\smallskip \noindent
\textbf{Algorithmic test cases.~~} The experiments are run for algorithms \ref{alg:0disCongif} to \ref{alg:terminating}. For Algorithm \ref{alg:terminating}, the experiments are run with $\delta \in \{0, 1, 10, 100, 500 and 1,000,000\}$. Note that $\delta = 0$ means that as soon as a node sees a pending snapshot, it blocks its write operations and helps complete the snapshot. This means the snapshot latency is minimized at the cost of a longer write latency. The fact that $\delta \in \{1, 10, 100, 500\}$ means that at least 1, 10, 100, and respectively, 500 write operations need to be observed before a node blocks its write operations. The case of $\delta = 1,000,000$ means that, at least in our experiment setup, write operations never block in order to help snapshot operations. This implies that the write latency is minimized at the cost of potentially starving the snapshotters.

\subsection{Results}
We present the results of experiments 1 to 4.

\subsubsection*{Experiment 1: Scalability of write operations with respect to write operations}
\label{sec:exp1}

\begin{figure}
\centering
\begin{minipage}{.45\textwidth}
    \centering
    \includegraphics[width=\textwidth]{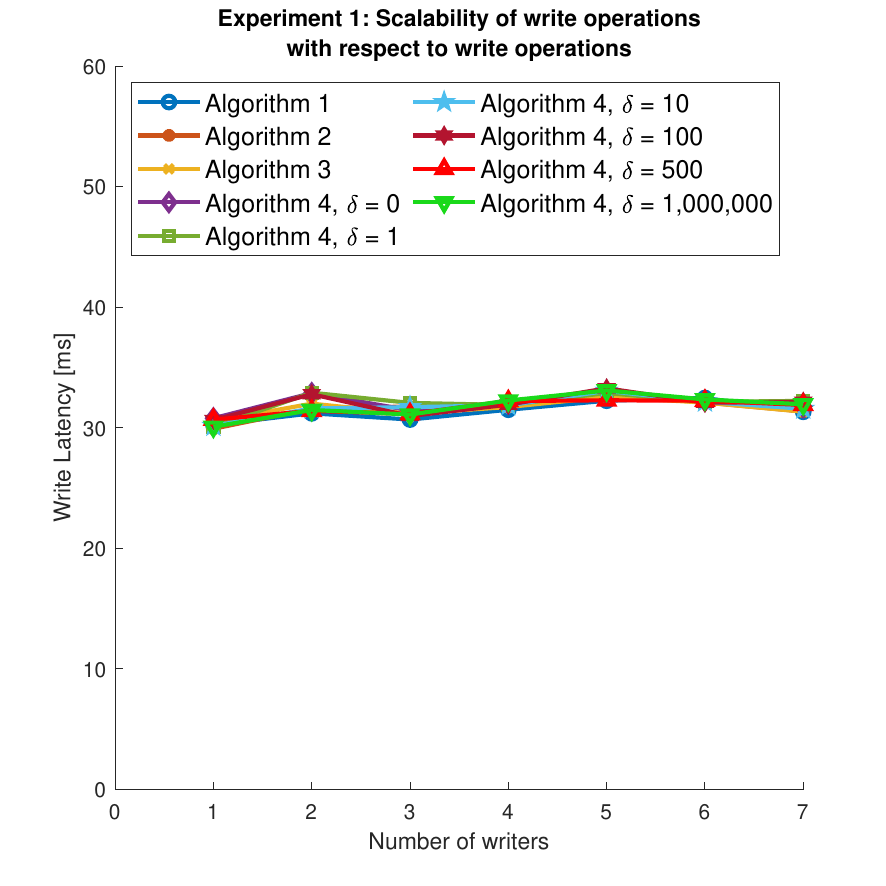}
    \captionof{figure}{The result of Experiment 1, showing an overview of the write latencies.}
    \label{fig:exp1_latency_overview}
\end{minipage}%
\hspace{0.5cm}
\begin{minipage}{.45\textwidth}
    \includegraphics[width=\textwidth]{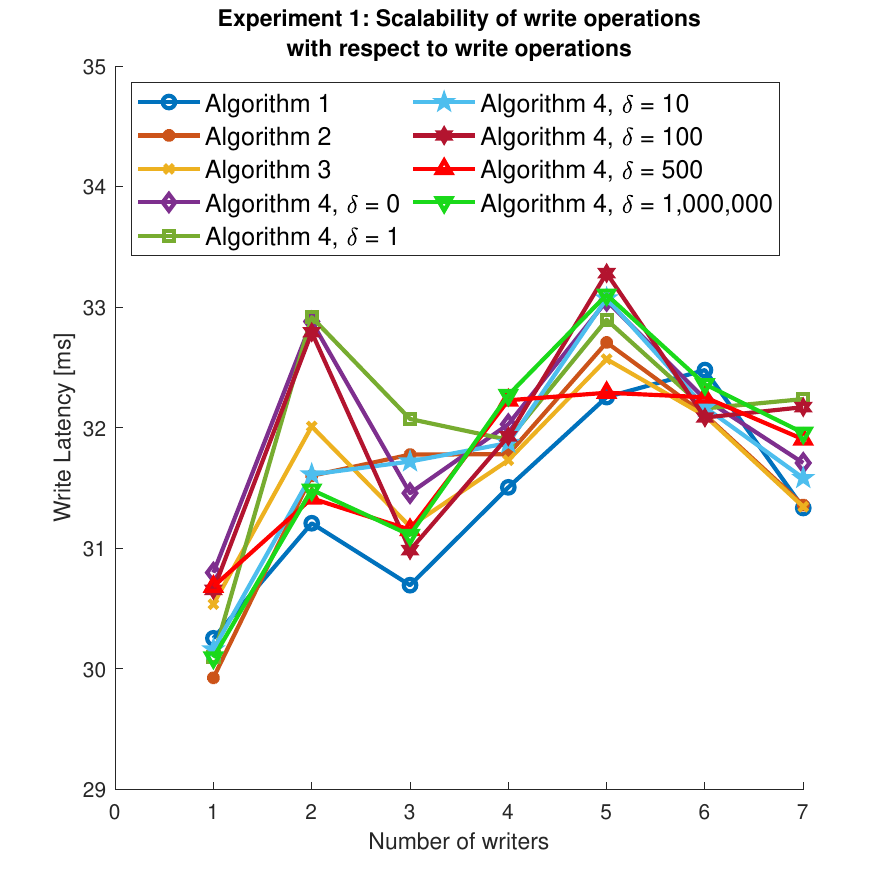}
    \centering
    \captionof{figure}{The result of Experiment 1, showing the write latencies zoomed in.}
    \label{fig:exp1_latency_zoomed}
\end{minipage}
\end{figure}

The result of Experiment 1 (latency) is shown in Figure~\ref{fig:exp1_latency_overview}. We observe that the write latency is between 29 and $34~ms$ for all algorithms and values of $\delta$. This is expected, since a write operation only needs a single quorum access and thus a single round trip. Recall that the average round trip time for the selected PlanetLab nodes is around $25~ms$ (see Section \ref{sec:system_setup}), which is matched by our results. The extra 4 to $9~ms$ comes from the local processing time which is needed in addition to the network delay (see Remark~\ref{thm:planetLab}). Figure~\ref{fig:exp1_latency_zoomed} is a zoomed in version of Figure~\ref{fig:exp1_latency_overview}. We see that there is a non-trivial pattern of latencies, which varies with the number of writers, but appears similar for all algorithms. Specifically, for a given number of writers, the difference between any two algorithms is at most $2~ms$. We suspect that this pattern is due to local processing time and CPU scheduling.

The number of write quorum accesses per write operation was, as expected, almost always one for all number of writers in Experiment 1. For all combinations of number of writers, algorithm and $\delta$, at most $0.002$ retransmissions per write operation were needed due to packet loss that is common in the Internet. Some retransmissions can also have occurred if a node started slightly before all other nodes. In that case the node might perform some retransmissions because it has started writing, but since no other node is started, the node does not get any replies, resulting in retransmissions.

\subsubsection*{Experiment 2: Scalability of write operations with respect to snapshot operations}
\label{sec:exp2}

\begin{figure}
\centering
\begin{minipage}{.45\textwidth}
    \centering
    \includegraphics[width=\textwidth]{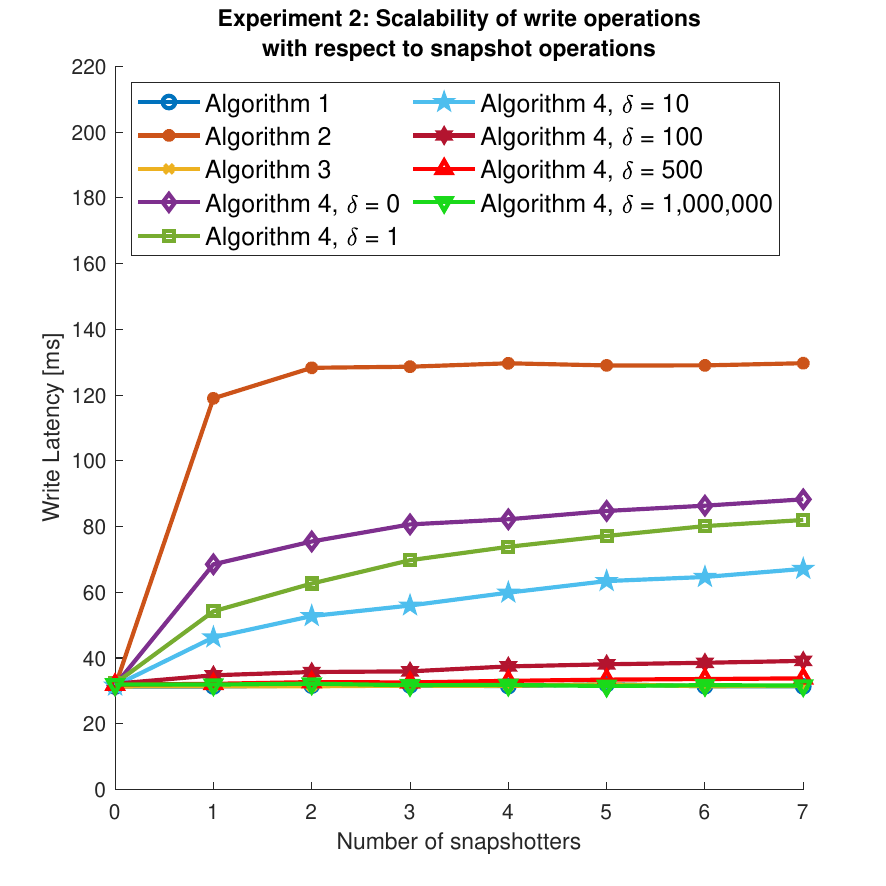}
    \captionof{figure}{The result of Experiment 2, showing an overview of the write latencies.}
    \label{fig:exp2_latency_overview}
\end{minipage}%
\hspace{0.5cm}
\begin{minipage}{.45\textwidth}
    \includegraphics[width=\textwidth]{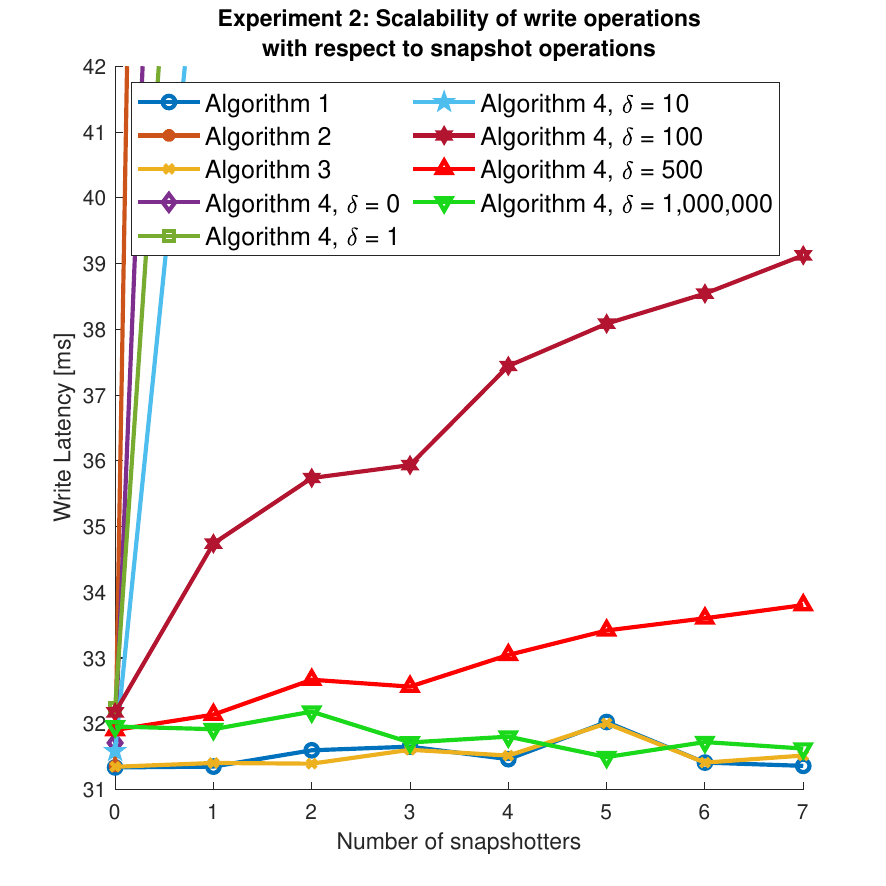}
    \centering
    \captionof{figure}{The result of Experiment 2, showing the write latencies zoomed in.}
    \label{fig:exp2_latency_zoomed}
\end{minipage}
\end{figure}

The result of Experiment 2 (latency) is shown in Figure~\ref{fig:exp2_latency_overview} (overview) and Figure~\ref{fig:exp2_latency_zoomed} (zoomed). There are several intersting observations. (1) When there are 0 snapshotters, the write latency is the same as in Experiment 1 for 7 writers. (2) The higher the $\delta$, the lower the write latency for Algorithm~\ref{alg:terminating}. In other words, this experiment confirms that higher $\delta$'s lead to lower write latencies since snapshot tasks are deferred for a longer time. (3) As the number of snapshotters increases, the write latency becomes longer, which is especially noticably for lower $\delta$ values. This is because writer nodes more often have to handle snapshot tasks, too, in $\baseSnapshot()$, because $\Delta$ becomes non-empty faster with more snapshotters. (4) The write latency of Algorithm~\ref{alg:9disCongif} is longer than Algorithm~\ref{alg:terminating} for all $\delta$ and number of snapshotters, meaning Algorithm~\ref{alg:terminating} is an improvement over Algorithm~\ref{alg:9disCongif} in this experiment. (5) There is no significant latency difference between Algorithm~\ref{alg:0disCongif} and \ref{alg:disCongif}.
As for number of write quorum accesses per write operation, we observed the same result as Experiment 1. Even if there are snapshotters, each write operation still requires only one quorum access.

\subsubsection*{Experiment 3: Scalability of snapshot operations with respect to snapshot operations}
\label{sec:exp3}

\begin{figure}
\centering
\begin{minipage}{.45\textwidth}
    \centering
    \includegraphics[width=\textwidth]{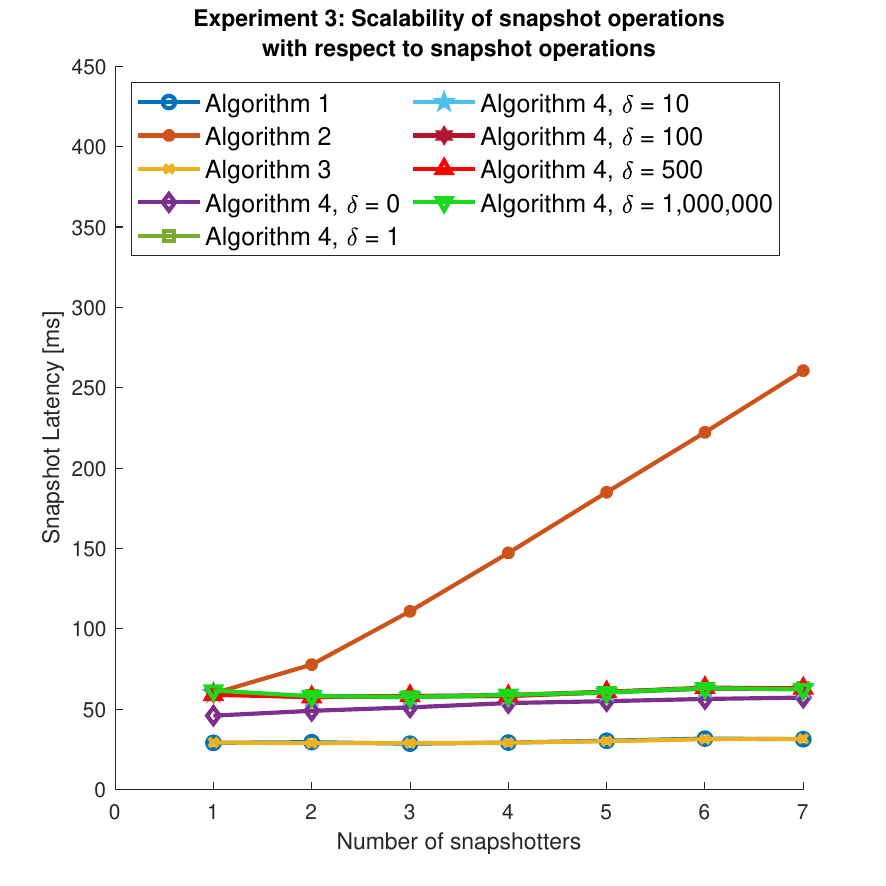}
    \captionof{figure}{The result of Experiment 3, showing an overview of the snapshot latencies.}
    \label{fig:exp3_latency_overview}
\end{minipage}%
\hspace{0.5cm}
\begin{minipage}{.45\textwidth}
    \includegraphics[width=\textwidth]{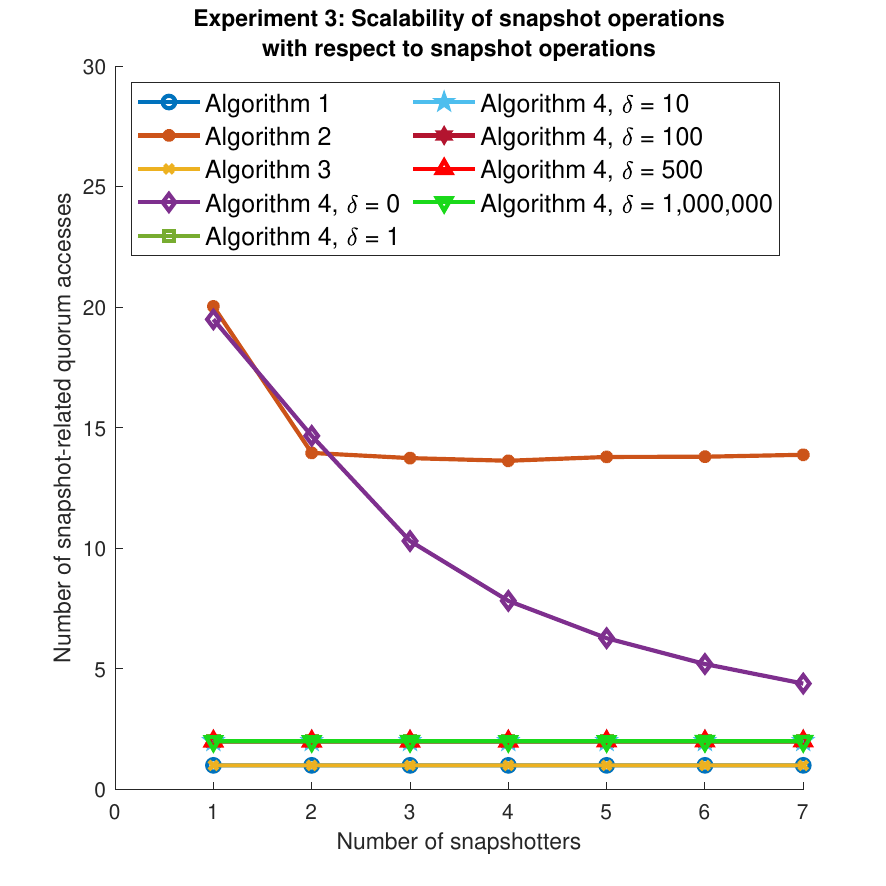}
    \centering
    \captionof{figure}{The result of Experiment 3, showing an overview of the number of snapshot-related quorum accesses.}
    \label{fig:exp3_quorum_overview}
\end{minipage}
\end{figure}

Figure~\ref{fig:exp3_latency_overview} shows the snapshot latency results for Experiment~3. For algorithms~\ref{alg:0disCongif} and \ref{alg:disCongif}, the snapshot latency is around $30~ms$, corresponding to one round trip time and one quorum access. This is expected since in the absence of writers, snapshots need only a single quorum access. 

For Algorithm~\ref{alg:terminating} and all $\delta$, the snapshot latency is around $60~ms$, which is corresponding to two round trip times and two quorum accesses. This is expected since one access is needed for the snapshot (line~\ref{ln:waitUntilSNAPSHOTackReg2}) and one for $\sStore()$ (line~\ref{ln:prevEqReg}). For the case of $\delta = 0$, we note that the latency is slightly lower than for the cases of $\delta > 0$. This is because for $\delta = 0$, node $p_j$ activates the helping scheme as soon as it becomes aware of any snapshot task. Thus, it can happen that $p_j$ calls $\sStore_j()$ with the task result, and the task creator, $p_i$, receives the corresponding $\mathrm{SAVE}$ message before $p_i$'s call to $\sStore_i()$. This means that $p_i$ only needs one quorum access for this particular client snapshot operation, instead of two. Our raw evaluation data confirms this hypothesis. Specifically, the largest difference appeared when there was a single snapshotting node that performed less calls to $\sStore()$ than invocations of $\sSnapshot()$. For $\delta > 0$, this cannot happen, since nodes do not activate the helping scheme before the occurrence of a concurrent write. However, Experiment~3 does not consider any write operations.

For Algorithm~\ref{alg:9disCongif}, the snapshot latency increases linearly as the number of snapshotters increases. This is because in Algorithm~\ref{alg:9disCongif}, all nodes must handle all snapshot operations, and they handle them sequentially. We see that also in Experiment~3, Algorithm~\ref{alg:terminating} offers improved latencies over Algorithm~\ref{alg:9disCongif}.

We now turn to Figure~\ref{fig:exp3_quorum_overview}, which shows the number of snapshot-related quorum accesses per snapshot operation. For algorithms~\ref{alg:0disCongif},~\ref{alg:disCongif} and~\ref{alg:terminating}, for the cases in which $\delta > 0$, the reason the results look like they do are for the same reasons given above for the snapshot latency results. For the case of $\delta = 0$, even if there is just a single snapshotter, all nodes help with that node's snapshot tasks immediately. Since Algorithm~\ref{alg:terminating} handles multiple snapshot tasks at once, as the number of snapshotters increases, the number of quorum accesses \emph{per snapshotter} decreases. We also see that for $\delta > 0$, Algorithm~\ref{alg:9disCongif} uses at least 6 times as many quorum accesses compared to Algorithm~\ref{alg:terminating}. For $\delta = 0$, Algorithm~\ref{alg:9disCongif} uses up to 3 times as many quorum accesses compared to Algorithm~\ref{alg:terminating}, for 7 snapshotters. In other words, Algorithm~\ref{alg:terminating} offers improved costs compared to Algorithm~\ref{alg:9disCongif}.

\subsubsection*{Experiment 4: Scalability of snapshot operations with respect to write operations}
\label{sec:exp4}

\begin{figure}
\centering
\begin{minipage}{.45\textwidth}
    \centering
    \includegraphics[width=\textwidth]{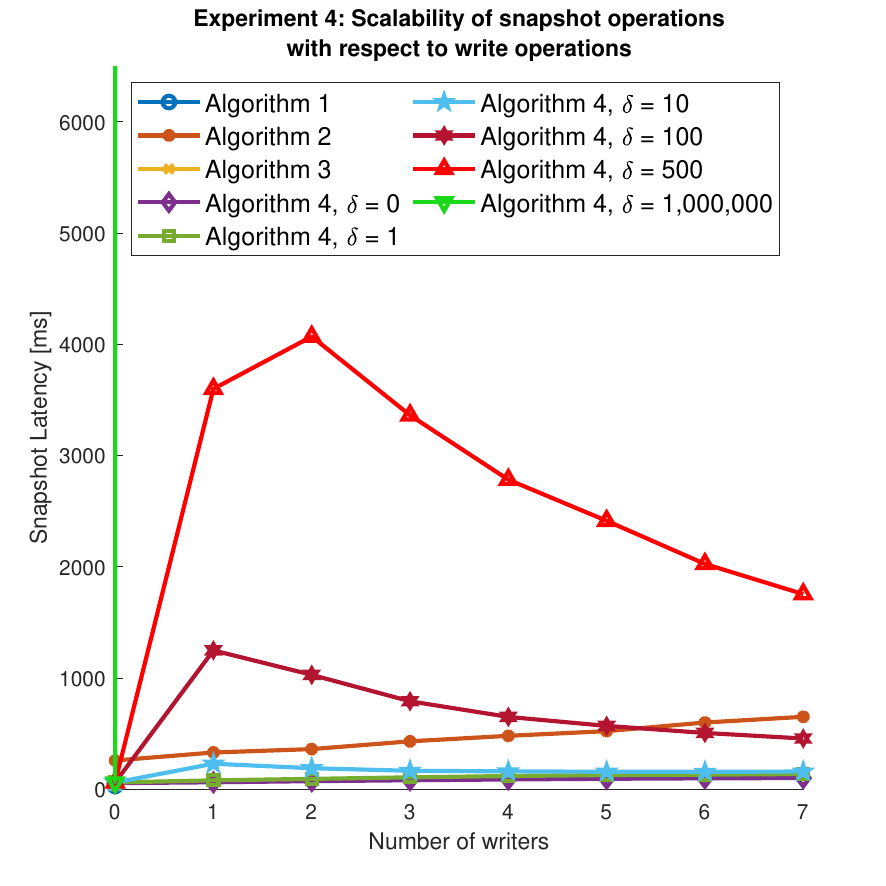}
    \captionof{figure}{The result of Experiment 4, showing an overview of the snapshot latencies.}
    \label{fig:exp4_latency_overview}
\end{minipage}%
\hspace{0.5cm}
\begin{minipage}{.45\textwidth}
    \includegraphics[width=\textwidth]{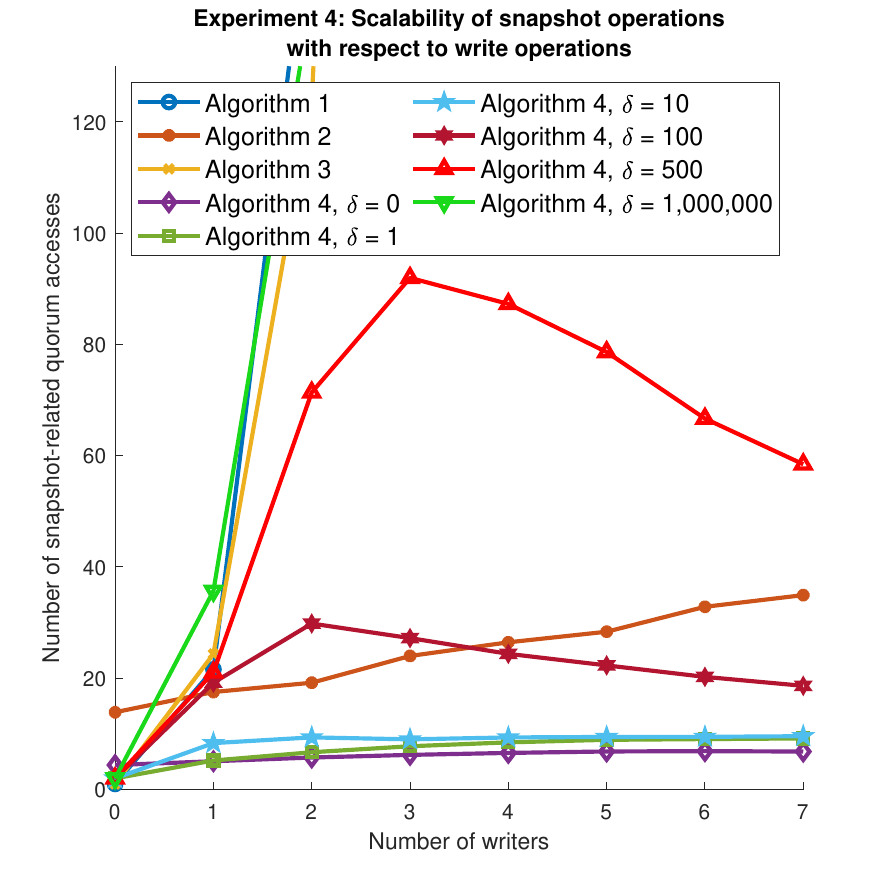}
    \centering
    \captionof{figure}{The result of Experiment 4, showing an overview of the number of snapshot-related quorum accesses.}
    \label{fig:exp4_quorum_overview}
\end{minipage}
\end{figure}

Figure~\ref{fig:exp4_latency_overview} shows the snapshot latency results of Experiment 4. There are several interesting observations to make. (1) As the number of writers changes from 0 to 1, the snapshot latency of algorithms~\ref{alg:0disCongif}, \ref{alg:disCongif} and \ref{alg:terminating} $\delta = 1,000,000$ goes to infinity. This is because they never pause write operations to help snapshot operations. (2) Higher $\delta$ means that the snapshot latency is longer. Thus, $\delta$ indeed can serve as a parameter that can balance the trade-off between write latency and snapshot latency. (3) For low $\delta$ values, increasing the number of writers leads to a longer snapshot latency since the job-stealing scheme is activated more often at the writers. (4) For high $\delta$ values, increasing the number of writers leads to a shorter snapshot latency. This is because a new snapshot quorum access has to be made (line \ref{ln:waitUntilSNAPSHOTackReg2}) even if $reg$ is different from $prev$ on only a single entry, which will happen with only 1 writer. If there are more writers, this will also happen. However, with more writers, the snapshot task will become concurrent with $\delta$ writes in a shorter amount of time, which in turn means that all nodes must block their writes faster, leading to a lower snapshot latency.

\begin{table*}[t!]
\centering
\caption{The correlation coefficients between the snapshot latency and the number of snapshot quorum accesses for Experiment~4, \ie the correlation between Figure~\ref{fig:exp4_latency_overview} and Figure~\ref{fig:exp4_quorum_overview}.}
\label{tab:correlation}
\begin{tabular}{|l|l|}
\hline
\textbf{Algorithm and $\delta$}          & \textbf{Correlation} \\ \hline
Algorithm 1                     & 0.76        \\ \hline
Algorithm 2                     & 1.00        \\ \hline
Algorithm 3                     & 0.83        \\ \hline
Algorithm 4, $\delta=0$         & 0.98        \\ \hline
Algorithm 4, $\delta=1$         & 0.97        \\ \hline
Algorithm 4, $\delta=10$        & 0.77        \\ \hline
Algorithm 4, $\delta=100$       & 0.69        \\ \hline
Algorithm 4, $\delta=500$       & 0.52        \\ \hline
Algorithm 4, $\delta=10^6$ & 0.76        \\ \hline
\end{tabular}
\end{table*}

Figure~\ref{fig:exp4_quorum_overview} shows the number of snapshot-related quorum accesses for Experiment~4. We observe a resembling behavior as for the latencies in this experiment. The reason is that if the latency is low, not many quorum accesses have the time to happen, and hence the number of quorum accesses is also low. From Table~\ref{tab:correlation}, we see that the correlation coefficients related to algorithms~\ref{alg:9disCongif} and~\ref{alg:terminating} (for the case of $\delta \in \{0,1\}$) indicate a strong correlation between the latency and the number of quorum accesses. For the cases with higher values of $\delta$, this correlation slowly diminishes, but still, their plots allow us to observe some degree of similarity despite the somewhat more complex behavior. Namely, the number of quorum accesses goes down slightly after the latency goes done.

The major difference is that the number of snapshot-related quorum accesses does not go to infinity as fast as the snapshot latency does (for algorithms~\ref{alg:0disCongif}, \ref{alg:disCongif} and \ref{alg:terminating} $\delta = 1,000,000$). This is due to the way the metrics are measured, see the formulas in Section \ref{sec:system_setup}, repeated here for convenience: $\frac{1}{|S|} \sum_{p_i \in S} \frac{60}{nso_i}$ for snapshot latency and ${\sum_{p_i \in \sP} (nsq_i + nsq'_i + nrb_i)}/{\sum_{p_i \in S} nso_i}$ for number of quorum accesses per snapshot operation. If only a single node completes no snapshot operation ($nso_i = 0$ for some $i$), the latency formula equals infinity. But if some node completes at least one snapshot operation, meaning the denominator $\sum_{p_i \in S} nso_i > 0$, the number of quorum accesses does not equal infinity. As the number of writers increases, fewer snapshot operations are be completed for all nodes, and eventually no node completes any snapshot operation, but based on the above reasoning, this is noticed quicker in the latency results compared to the number of quorum accesses.

\subsection{Discussion}
\label{sec:discussion}
Our first research question asks how well the system scales with respect to writers and snapshotters. From the results, we see that write operations scale well for all algorithms. Even when there are concurrent snapshot operations, the write latencies do not grow unbounded. The number of quorum accesses needed per write operation is always around one. For snapshot operations, the picture is different. For Algorithm~\ref{alg:0disCongif} and Algorithm~\ref{alg:disCongif}, snapshot operations take constant time when there are no writes, but if there are concurrent writes, the snapshot latency becomes unbounded. While Algorithm~\ref{alg:9disCongif} does not starve snapshot operations, the snapshot latency grows linearly as the number of snapshotters increases. For Algorithm~\ref{alg:terminating}, by varying $\delta$, we can get various latencies. Selecting $\delta = 0$, the write latency scales fairly well and the snapshot latency is constant, only twice as slow as Algorithm~\ref{alg:0disCongif}'s and Algorithm~\ref{alg:disCongif}'s snapshot latency.

Our second research question asks what overhead is needed to pay when comparing algorithms~\ref{alg:0disCongif} and~\ref{alg:disCongif} as well as algorithms~\ref{alg:9disCongif} and~\ref{alg:terminating}. In all experiments we observed that the overhead of Algorithm~\ref{alg:disCongif} compared to Algorithm~\ref{alg:0disCongif} is negligible. We also see that in all experiments, Algorithm~\ref{alg:terminating} with $\delta = 0$ is always faster, or very close to, Algorithm~\ref{alg:9disCongif}. The potential overhead of self-stabilization is more than compensated for by the other improvements of Algorithm~\ref{alg:terminating}. Furthermore, by tuning $\delta$ to suit the pattern of write and snapshot operations the current application of Algorithm~\ref{alg:terminating} generates, Algorithm~\ref{alg:terminating} can be made significantly faster than Algorithm~\ref{alg:9disCongif}.

Our third research question asks what role $\delta$ plays. We see that when there are only writers or snapshotters (experiments 1 and 3), the value of a non-zero $\delta$ does not play a significant role. In Experiment~2, we see that, on the one hand, a higher $\delta$ leads to shorter write latency, while on the other hand, in Experiment~4, we see that a higher $\delta$ leads to longer snapshot latency. Based on the number of nodes in the system and the pattern of writes and snapshots, one can decide a value of $\delta$ that gives the desired trade-off between write and snapshot latency. There is no single best $\delta$, but in our experiment setup, it seems that $\delta = 10$ gives a good trade-off and is in most cases faster than Algorithm~\ref{alg:9disCongif}.

\section{Conclusions}
\label{sec:disc}
We showed how to transform the two non-self-stabilizing algorithms of Delporte-Gallet \etal~\cite{DBLP:journals/tpds/Delporte-Gallet18} into ones that can recover after the occurrence of transient faults. This requires some non-trivial considerations that are imperative for self-stabilizing systems, such as the explicit use of bounded memory and the reoccurring clean-up of stale information. Interestingly, these considerations are not restrictive for the case of Delporte-Gallet \etal~\cite{DBLP:journals/tpds/Delporte-Gallet18}. For our self-stabilizing atomic snapshot algorithm that always terminates, we chose to use safe registers for storing the results of recent snapshot operations, rather than a mechanism for reliable broadcast, which is more expensive to implement. Moreover, instead of dealing with one snapshot operation at a time, we deal with several at a time. In addition, we consider a tunable input parameter, $\delta$, for allowing the system to balance a trade-off between the latency of snapshot operations and communications costs, which range from $\bigO(n)$ to $\bigO(n^2)$ messages per snapshot operation. 

One future direction emanating from this work is to consider the $\bigO(n)$ gap in the number of messages when 
designing future applications. For example, one might prefer the use of repeated snapshots (using the proposed solution) over a replicated state machine, which always costs $\bigO(n^2)$ messages per state transition. Another future direction is to consider the techniques presented here for providing self-stabilizing versions of more advanced snapshot algorithms, such as the one by Imbs \etal~\cite{DBLP:conf/icdcn/ImbsMPR18,DBLP:books/sp/Raynal18}.

%\input{TRv2.bbl}
%
%\bibliographystyle{plain}
%\bibliography{referancesnew}

\end{document}